\documentclass[journal,onecolumn,12pt]{IEEEtran}
\usepackage[doublespacing,nodisplayskipstretch]{setspace}
\usepackage[utf8]{inputenc} 
\usepackage[T1]{fontenc}
\usepackage{url,comment}
\usepackage{epsfig,xcolor,graphics,graphicx}
\usepackage{amsmath,amssymb,amsthm}

\usepackage{soul}

\newtheorem{theorem}{Theorem}
\newtheorem{lemma}{Lemma}
\newtheorem{corollary}{Corollary}
\newtheorem{definition}{Definition}

\newtheorem{proposition}{Proposition}

\newcommand{\cO}{\mathcal{O}}
\newcommand{\cC}{\mathcal{C}}
\newcommand{\pr}{\mathbb{P}}
\newcommand{\R}{\mathbb{R}}
\newcommand{\Z}{\mathbb{Z}}
\newcommand{\N}{\mathbb{N}}
\newcommand{\E}{\mathbb{E}}
\newcommand{\F}{\mathbb{F}}
\newcommand{\hX}{\hat{X}}
\newcommand{\hY}{\hat{Y}}
\newcommand{\hx}{\hat{x}}

\newcommand{\hz}{\hat{z}}
\newcommand{\hU}{\hat{U}}
\newcommand{\hpi}{\hat{\pi}}
\newcommand{\tX}{\tilde{X}}
\newcommand{\tY}{\tilde{Y}}

\newcommand{\nrd}[1]{\textcolor{red}{#1}}
\newcommand{\nbl}[1]{\textcolor{blue}{#1}}

\usepackage[hidelinks]{hyperref}
\allowdisplaybreaks[1]

\title{Diversity in Coded TE-QKD Channels: Achieving Infinite Diversity out of Finite System Resources}
\author{Shaikha S. Al-Qahtani, Siyao Li, and Joseph J. Boutros\\
shaikha@tamu.edu, lis14@erau.edu, boutros@ieee.org}
\date{July 2025}

\begin{document}
\maketitle
\vspace{-2cm}
\begin{abstract}
%%\nrd{SL edited on July 24.\\}
In this paper, we establish conditions and give proofs on how an error-correcting code can attain infinite diversity in a time-entanglement quantum key distribution (TE-QKD) reconciliation. 
The shocking result, never encountered in the literature on coding and communication theory, is that a decoder exhibits an infinite diversity
order while the channel has finite diversity and the code has a relatively short finite length. 
Indeed, TE-QKD encodes multiple raw key bits per entangled photon pair over discrete time bins. Ideally, time entanglement yields identical photon-arrival times, and thus identical raw keys, at Alice's and Bob's detectors. However, practical detector timing jitter introduces measurement errors. Modeled as Gaussian noise, this jitter causes disagreements between Alice's and Bob's bin indices and necessitates information reconciliation. 
This paper studies the  diversity order of coded TE-QKD reconciliation, defined by the asymptotic slope of the error probability as the signal-to-noise ratio (SNR) $\gamma$ is asymptotically large, i.e., in the low noise regime. At high SNR, the channel exhibits two error mechanisms: a jump to a neighboring bin has probability $\Theta(\gamma^{-1/2})$, while a jump of two bins or more has exponentially decaying probability 
as $\cO(e^{-\tfrac{\gamma}{4}})$. 
 This separation permits a code to transform the uncoded diversity order $1/2$ into infinite diversity by eliminating all decoding failures composed solely of single-bin jumps. 
 For bounded-distance algebraic decoding, we derive a necessary and sufficient condition in terms of the number of photons per codeword and the decoding radius. 
 For soft-decision decoding, we introduce the maximal finite diversity (MFD) property and prove that infinite diversity is achieved if and only if the code is MFD deficient. 
 The resulting conditions further yield information rate bounds and show that soft reconciliation can attain infinite diversity at substantially higher rates, or with fewer bins per frame, than algebraic reconciliation. 
 This behavior has no counterpart in classical fading channels, where decoding can only multiply a finite diversity order by a finite factor. 
 Examples of short codes based on Golay, Reed-Solomon, Bose-Chaudhuri-Hocquenghem (BCH), and Reed-Muller codes validate the analysis and illustrate how the TE-QKD system parameters and the relatively short code parameters affect the achievable diversity for both hard and soft information reconciliation.
\end{abstract}


%%--------------------------------------------------------------
%%--------------------------------------------------------------
\section{Introduction}
Quantum Key Distribution (QKD) protocols enable two authenticated parties, Alice and Bob, to establish a shared secret key whose security rests on the laws of quantum mechanics~\cite{Bennett2014,Ekert1991}.
However, practical QKD is fundamentally constrained by optical attenuation, detector imperfections, background noise, and the classical post-processing required to convert correlated measurements into identical secret keys~\cite{Xu2020,Pirandola2020}. Therefore, the design of efficient information reconciliation codes is central to converting physical-layer correlations into a useful secret-key rate,
e.g., see the non-exhaustive list on error-correcting codes for information reconciliation in polarization-based QKD \cite{traisilanun2007}-\cite{Tarable2024}.

The term \emph{diversity} appears in several different contexts in quantum communications, ranging from the variety of QKD protocols and physical implementations~\cite{Rehman2025} to the use of multiple spatial paths, optical modes, or time intervals in free-space and satellite links~\cite{Zhu2002,Scarfe2023,Cao2020}.
These techniques mitigate turbulence, beam wandering, and time-varying loss, which can reduce the secret key rate and appear as excess noise in continuous-variable QKD~\cite{Wang2018}.
Previous work has considered post-selection, transmittance clustering, pilot-assisted estimation, and security analysis under fading (see, e.g., \cite{Usenko2012,Hosseinidehaj2015,Pirandola2021}). Spatial-mode and space-division multiplexing have also been explored to increase aggregate key rates or improve robustness against stochastic channel impairments. In particular, recent work on continuous-variable quantum communications has shown that combining independently faded spatial modes can increase the average secret key rate relative to single-mode transmission in some operating regimes~\cite{Koudia2025}. This physical-layer diversity is analogous to path or antenna diversity in classical wireless systems.

In this work, we study the \emph{diversity order} of information reconciliation in time-entanglement quantum key distribution (TE-QKD) under algebraic/hard-decision and soft-decision decoding. The diversity order is defined as $d = -\lim_{\gamma \to \infty} \tfrac{\log P_{\rm e}(\gamma)}{\log \gamma}$ where $\gamma$ denotes the signal-to-noise ratio (SNR) and $P_{\rm e}$ is the reconciliation error probability. 
In coded and uncoded fading channels, this quantity describes the asymptotic slope of an error-probability curve on a log--log scale.
Independent fading realizations, multipath components, or properly designed channel codes can increase the diversity order and hence accelerate the polynomial decay of the error probability~\cite{Boutros1998,Malkamaki2002,Tse2005}. For a channel having an uncoded diversity order $d_0$, $d_0 < \infty$, 
bounded-distance hard decoding yields a diversity equal to $d_0 \times(t+1)$, whereas soft-decision decoding yields a diversity equal to $d_0 \times d_{Hmin}$, 
where $d_{Hmin}$ is the minimum Hamming distance of the code
and $t$ is the decoding radius, $t=\lfloor \tfrac{d_{Hmin}-1}{2}\rfloor$. Typical conditions for such finite-factor diversity improvement in standard wireless communications is recalled below
in Section~\ref{sec_div_wireless}. In both hard and soft decision cases, a finite channel diversity remains finite after decoding.
To our knowledge, there exists no special case in the literature 
where infinite diversity is created out of a finite one.

Time-entanglement-based QKD (TE-QKD) provides a markedly different setting by exploiting the strong correlation between the arrival times of entangled photon pairs~\cite{Boutros2023,Dolecek2023}. The observation interval is partitioned into bins and frames, and the occupied bin index provides a high-dimensional raw key symbol, similarly to pulse position modulation (PPM)~\cite{Zhong2015}.
With $N=2^m$ bins, one detected pair can convey as many as $m$ raw key bits. This photon efficiency is attractive, but timing jitter perturbs the measured arrival positions at Alice and Bob and can move their observations across bin boundaries. The resulting raw-key disagreements must be corrected through information reconciliation, during which Alice transmits syndrome or parity information to Bob over an authenticated public channel. 
Existing studies of QKD reconciliation have primarily focused on the reconciliation efficiency, decoding complexity, frame error rate, and resulting secret key rate (see, e.g.,~\cite{Boutros2023,Wang2022,Almeida2023}).
In TE-QKD, coding gains have been demonstrated numerically for Reed--Solomon, BCH, and LDPC reconciliation~\cite{Boutros2023}. 
However, it is not straightforward to guess how conventional coding parameters such as length, rate, and minimum Hamming distance do determine the decay rate of the TE-QKD reconciliation error,  
whether it decays polynomially in SNR (finite diversity) or exponentially in SNR (infinite diversity). This paper establishes the proofs and the conditions to achieve infinite diversity.

\paragraph*{Our Main Contributions} 
This work addresses the question of creating infinite diversity in a finite-diversity system by developing a diversity theory for coded TE-QKD reconciliation. The main contributions are summarized as follows.
%%~\\
%%\nrd{Joseph and Siyao, we resume on July 29 at this point.}
%%~\\
\begin{itemize}
%% single, multi-bin jump
\item We derive high-SNR transition laws for the hard- and soft-output TE-QKD channels via two equivalent models, and establish two asymptotically distinct error mechanisms. We prove that single-bin jumps occur with probability $\Theta(\gamma^{-1/2})$, whereas jumps of two or more bins decay exponentially with $\gamma$ as~$\cO(e^{-\tfrac{\gamma}{4}})$. 
See Lemmas~\ref{lem_U-X-Y-1-jump}--\ref{lem_X-Y-2-jump} and 
Proposition~\ref{prop:single-double-jumps}.

%\nbl{We formulate a discrete channel model for the TE-QKD raw-key disagreement induced by detector jitter, building on the PPM-based time-binning framework of~\cite{Boutros2023}. We show that the hard-output TE-QKD channel has two distinct high-SNR error mechanisms: i) single-bin jumps, whose probabilities decay polynomially and yield the uncoded diversity order $1/2$, and ii) multi-bin jumps, whose probabilities decay exponentially and yield the infinite diversity.}
%%% hard decoding
\item For bounded-distance algebraic decoding, we establish a necessary and sufficient condition for infinite diversity in terms of  the number of photons required per codeword and the decoder correction radius. We also extend the necessary and sufficient condition to non-binary codes over finite fields of characteristic~2.
See Theorem~\ref{thm:infinite-diversity-hard} and  Corollary~\ref{cor_infinite-diversity-hard-nonbinary}.
%%% soft decoding
\item For soft-decision decoding, we introduce the maximal-finite-diversity (MFD) property in Definition~\ref{def_MFD_code}, which captures whether a non-zero codeword can be formed entirely from labels reachable through single-bin jumps. 
We develop a demilitarized-zone (DMZ) analysis in Lemma~\ref{lem_DMZ} that isolates the dominant soft-decision error region and proves that all decoding errors outside the DMZ have
exponentially decaying probability under a condition on the DMZ width.
\item The most important result in this paper is the proof that soft reconciliation has infinite diversity if and only if the code is MFD deficient, see Theorem~\ref{thm:infinite-diversity-soft}.
A simple sufficient condition for infinite diversity based on the code minimum Hamming distance is provided.  
%%% rate corollaries
\item The Singleton bound converts the infinite diversity condition into an upper bound on the coding rate in Corollaries~\ref{cor_rate_loss_hard} and~\ref{cor_rate_loss_soft}, and we show that the rate penalty under soft-decision decoding is half that of algebraic decoding. 
\item Finally, examples based on short linear codes illustrate the
conditions proved in Theorems~\ref{thm:infinite-diversity-hard} and \ref{thm:infinite-diversity-soft}. Examples include the binary self-dual Golay code, short Reed-Solomon codes, binary and quaternary BCH codes, as well as binary Reed-Muller codes. 
See Section~\ref{sec_code_examples}. 
The examples show that soft decoding generally reaches infinite diversity with fewer bins per frame, or equivalently with a number of photons per word not as small as for algebraic decoding.
\end{itemize}

These results reveal a non-standard diversity behavior of coded TE-QKD channels. 
Unlike conventional fading channels, where finite input diversity remains finite after decoding, the TE-QKD channel decoder can exhibit infinite diversity when the code structure eliminates all dominant single-bin-jump error events.

\paragraph*{Paper Organization} Section~\ref{sec_model} introduces the TE-QKD channel model and notation. Section~\ref{sec_finite} reviews diversity in classical fading channels and contrasts it with the diversity behavior of uncoded and conventionally coded TE-QKD channels.
Section~\ref{sec_probability} analyzes the probabilities of single-bin and multi-bin jumps.
Section~\ref{sec_infinity-div} establishes the main infinite-diversity conditions under algebraic and soft-decision decoding respectively.
Section~\ref{sec_code_examples} illustrates these conditions with short binary and non-binary codes.
Section~\ref{sec_conclusion} concludes this work.

%%--------------------------------------------------------------
%%--------------------------------------------------------------
\section{TE-QKD Channel Model and Notations\label{sec_model}}

\begin{figure}[!h]
\begin{center}
\includegraphics[width=10cm]{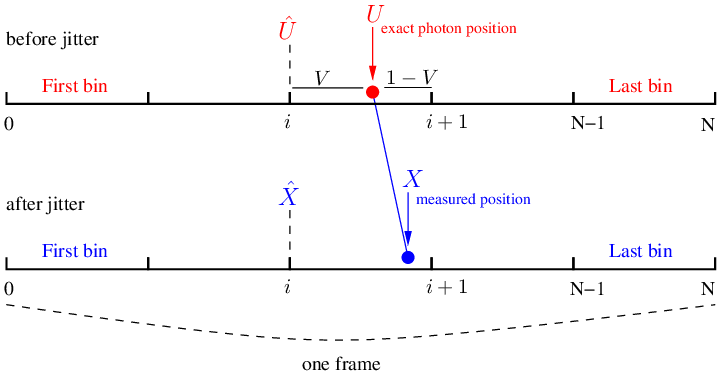}
\caption{The frame structure on Alice's side showing the photon position before and after detection jitter noise.\label{fig_modeloneside}}
\end{center}
\end{figure}

We consider the TE-QKD channel involving two distant parties, Alice and Bob. We assume independent photon-pair realizations and independent detector jitters across channel uses. Therefore, the induced discrete or soft-ouput channel is memoryless under the adopted model.
Following the same notations used in \cite{Boutros2023}, the measured \emph{soft} positions by Alice and Bob are modeled as 
\begin{align}
\tilde{X} = U + Z_1, ~~~~  \tilde{Y} = U + Z_2, 
\label{eq:channel-model}
\end{align}
$\tilde{X},\tilde{Y} \in \R$, where $Z_1$ and $Z_2$ are independent identically distributed (i.i.d.) Gaussian noise with zero mean and variance $\sigma^2$, denoted as $\mathcal N(0, \sigma^2)$, modeling the detection jitter. 
The random variable $U \sim \mathrm{Unif}([0,N[)$ is uniform in the real range $[0,N[$, $N \in \N$, $N \ge 2$. The variable $U$ is the jitter-free photon position within a frame of $N$ bins, bin number $i$ being the interval $[i,i+1[$, $i\in\Z_N=\{0, 1, \ldots, N-1\}$.  A frame is \emph{valid} for a party if its measured soft position is in $[0,N[$. Alice and Bob reject invalid (empty) frames and frames with more than one received photon. Let $X$ and $Y$ denote the instances of $\tilde{X}$ and $\tilde{Y}$ within the interval $[0, N[$, and let $\hat{X}$ and $\hat{Y}$ be the bin number inside a frame. Conditioned on both frames being valid, define \emph{hard} decisions \begin{align}
   \hX=\lfloor X\rfloor \in \Z_N, ~~~~ \hY=\lfloor Y\rfloor \in \Z_N,
\end{align}
where 
\begin{align}
 X=\tilde{X}, \text{ for } \tilde{X} \in [0, N[, ~~~~  Y=\tilde{Y}, \text{ for } \tilde{Y} \in [0, N[.   
\end{align}  
The notation $\lfloor x \rfloor$ refers to the floor function of $x$, $x \in \R$. Figure~\ref{fig_modeloneside} illustrates a frame on Alice's side.
The frame structure is similar on Bob's side. The exact position $U$ is identical
in both frames thanks to the time entanglement; however, the photon detectors
generate different jitters $Z_1$ and $Z_2$. This mathematical model does not
take into account the sleeping phase of a detector after receiving a photon.
Such a simplification has no effect on the code design and the error-rate performance evaluation. Similar to $\hX$ and $\hY$, we define the bin number of the exact photon position as $\hU=\lfloor U\rfloor \in \Z_N$. As illustrated in Figure~\ref{fig_modeloneside},
the distance to the left bin is $V=U-\hU$ and the distance to the right bin is $1-V$.

For reconciliation in the QKD protocol, it is considered that $\hat{X}$ is the channel input (on Alice 's side), $\hat{Y}$ and $Y$ are the hard output and the soft output respectively of the channel (on Bob's side). Appendix~\ref{app_fundamental} includes a list of fundamental equations characterizing the $\hat{X}-\hat{Y}$ hard-output (algebraic) channel and the $\hat{X}-Y$ soft-output channel. 
The apriori probability of a bin $\hat{\pi}_i=\pr(\hat{X}=i)$ is given by (\ref{equ_hpi}) with a simplified expression (\ref{equ_hpi_simple}) at high SNR. For analyzing the performance and the diversity order, and for the sake of simplicity, taking $\hat{\pi}_i=1/N$ (equal probability bins) is sufficient and does not hurt the analysis. However, (\ref{equ_hpi}) is utilized for exact numerical evaluations especially at low and moderate SNR. Appendix~\ref{app_fundamental} also summarizes the expressions of the channel likelihood, i.e., the conditional density $p(y|\hat{x})$ in (\ref{equ_Y_cond_hX_tY}) for the soft-output channel, the transition probability $p_{ij}$ for the algebraic channel in (\ref{equ_pij}), and finally the density $p(u|\tX, \tY \in [0,N[)$ of $U$ conditioned on both Alice and Bob's frames being valid in~(\ref{equ_pdfU_valid}). At high SNR, as given by (\ref{equ_pdfU_highsnr}), $U$ becomes uniform in the range $[0,N[$.\\ 

The main channel parameter $\gamma$ is the signal-to-noise ratio (SNR) parameter defined as 
\begin{align}
    \gamma = \frac{E_s}{\sigma^2} = \frac{1}{\sigma^2}
\end{align}
where the average symbol energy is normalized to $E_s=1$, representing the energy cost per transmitted photon. 
In the absence of error-correcting codes, the information rate of the channel is $R_b=m=\log_2(N)$ bits per photon, or equivalently $R_b$ bits per channel use (bpcu). 
In the presence of a rate $k/n$ error-correcting code in the reconciliation process,
the information rate becomes $R_b=\tfrac{k}{n}\cdot m=\tfrac{k}{n}\cdot \log_2(N)$ information bits per photon, for $m$ coded bits per photon. 

For this TE-QKD model, when the eavesdropper has access only to the public communication and no correlated side information, the secrecy capacity equals the mutual information between Alice's and Bob's measurement~\cite[p.567]{ElGamal2011}, i.e., $I(X; Y)$, as derived in~\cite[Theorem 3]{Boutros2023}. This secrecy capacity serves as an achievable upper bound on the rate after privacy amplification. In the low-noise regime, as stated by Corollary~2 in~\cite{Boutros2023}, $I(X; Y)$ is asymptotic to $\tfrac{1}{2}\log(\tfrac{\gamma}{4\pi e})$ expressing the addition of Gaussian jitter noises on both Alice's and Bob's sides.\\
 
\begin{definition}
\label{def_div}
Let $P_e=P_e(\gamma)$ be the probability of error of a communication system. 
The latter could be coded or uncoded. The error probability $P_e$ could refer to 
a bit error probability, or a symbol error probability, or a word error probability,
before or after decoding. The diversity order $d$ associated to $P_e$ is the
quantity
\begin{equation}
d=-\lim_{\gamma \rightarrow \infty} \frac{\log(P_e)}{\log(\gamma)}.
\end{equation}
\end{definition}
Let $P_{e1}(\gamma)$ have diversity $d_1$ and $P_{e2}(\gamma)$ have diversity $d_2$;
if there exists $\gamma_0$ such that $P_{e1}(\gamma) \le P_{e2}(\gamma)$, $\forall \gamma \ge \gamma_0$, then we have $d_2 \le d_1$.
Two cases are of great interest for our work. Firstly, the finite-diversity case
where the probability of error has the general shape
\begin{equation}
\label{equ_pe_findiv}
P_e(\gamma)=\frac{K(\gamma)}{\gamma^{d+\alpha(\gamma)}}+\beta(\gamma)=\Theta(\gamma^{-d}),
\end{equation}
where the diversity order is a constant $0<d<\infty$ (not necessarily an integer),
with the following mild conditions on the remaining terms:  
a) $|\alpha(\gamma)|=O(\gamma^{-\lambda})$ for some constant $\lambda>0$, b) $K(\gamma)>0$ and $\lim_{\gamma \rightarrow \infty} \log(K(\gamma))/\log(\gamma)=0$, and c) $\lim_{\gamma \rightarrow \infty} \beta(\gamma)\cdot\gamma^{d+\alpha(\gamma)}/K(\gamma)=0$. 
At asymptotic SNR, $P_e$ decreases as $K/\gamma^d$ where the slope in a double-logarithmic scale is equal to $-d$. After writing $K/\gamma^d=1/(K^{-1/d}\gamma)^d$, the {\em coding gain} refers to the factor $K^{-1/d}$. Thus, $K$ affects the coding gain with no effect on the slope, i.e., the $P_e(\gamma)$ plot moves left or right without a change in the slope. This paper mainly studies the diversity of a coded TE-QKD system, not its coding gain.\\
The second diversity case corresponds to an exponential decrease of the error probability,
\begin{equation}
\label{equ_pe_infdiv}
P_e(\gamma)=K \cdot \exp(-\alpha \gamma^{\beta}), ~~~~K,\alpha,\beta >0,
\end{equation}
where the constant $K$ could also be a polynomial or a sub-exponential function of $\gamma$. In this second case, $-\lim_{\gamma \rightarrow \infty} \tfrac{\log(P_e)}{\log(\gamma)}=\infty$, i.e., we have an infinite diversity in (\ref{equ_pe_infdiv}) as opposed to the finite diversity $d$ achieved in (\ref{equ_pe_findiv}). 

In communication theory and information theory, the literature usually focuses on the diversity order as given in Definition~\ref{def_div}, 
see for example \cite{Tse2005}\cite{Biglieri2005}\cite{Boutros2010}. 
It is also possible to study the so-called effective diversity at finite
SNR as in \cite{Vashakidze2024}, which is the slope of $P_e(\gamma)$ in a doubly logarithmic scale for a given finite $\gamma$. In this paper, 
we focus on formal diversity as stated by Definition~\ref{def_div}. 
Section~\ref{sec_finite} recalls the general behavior of diversity on standard fading channels with additive Gaussian noise and its ``expected'' 
behavior on time-entangled QKD channels, before describing ``unexpected''
behavior in Section~\ref{sec_infinity-div} where infinite diversity is created out of a finite diversity.

%%--------------------------------------------------------------
%%--------------------------------------------------------------
\section{From Standard Fading Channels to TE-QKD \label{sec_finite}}
In wireless communication, signals often travel from a transmitter to a receiver via multiple paths due to reflections off buildings, terrain, and other objects. This multipath propagation can cause the signal strength at the receiver to fluctuate randomly, which is known as fading~\cite{Tse2005}. A simple mathematical model for a fading channel in baseband, referred to as coherent frequency-non-selective (no inter-symbol interference), is
\begin{equation}
\label{eq_fadingchannel}
Y_t=\alpha_t X_t + Z_t,~~~~\alpha_t> 0,
\end{equation}
for linear modulations \cite{Proakis2008}, where $t \in \Z$ is the channel instance.
The integer $t$ could represent discrete-time, a channel propagation path index, 
or the index of a transmit or a receive antenna, or the number of a communication link with a specific carrier frequency. Typically, the random variable $\alpha_t$ is a fading that follows a Rayleigh distribution with density 
$p(\alpha)=2\alpha e^{-\alpha^2}$,
for $\alpha>0$ and a unit second order moment. 
Under some particular conditions, fading could follow other types of
distributions such as Rice and Nakagami distributions.   

When the receiver has access to a unique Rayleigh fading $\alpha_t$ during a signal detection,
we say that diversity is $1$, which is equivalent to no diversity or absence of diversity.
The weak received signal corresponding to $\alpha_t$ close to $0$ (called a deep fading) makes the detector
fail in finding $X_t$ from $Y_t$, with a high probability. It can be proven that the probability of error varies as $P_e=\Theta(\tfrac{1}{\gamma})$ \cite{Tse2005}\cite{Proakis2008}. 
If the receiver has access to $d$ independent Rayleigh fadings, then diversity is $d$.
You may say that the diversity order is $d$. In this case, it can also be proven that 
$P_e=\Theta(\tfrac{1}{\gamma^d})$ if the receiver exploits the $d$-diversity.
Correlation between the $d$ fading variables usually impacts the coding gain only, 
e.g., see~\cite{Veeravalli2001}. The $d$~independent or correlated fadings could
be observed on a wireless channel with multiple paths or with multiple antennas.
Frequency hopping and time variations also create a high-order diversity known
as frequency diversity and time diversity respectively. At this point, the reader
should conclude that $d$ is always a non-negative integer when dealing with Rayleigh fading.
In some urban wireless channels where fading follows a Nakagami-$d$ distribution, of density $p(\alpha)=\tfrac{2}{\Gamma(d)}d^d\alpha^{2d-1}e^{-d\alpha^2}$, 
$\alpha>0$, $\Gamma()$ is the gamma function, the diversity order $d$ can take any non-negative real value,  i.e., non-integer diversity is encountered on Nakagami fading channels (even less than 1).\\

Understanding diversity in the standard context, summarized above for wireless communications, provides a clear foundation for the analysis of the diversity order in TE-QKD channels, where the source of performance degradation is not signal fading but the exact photon position with respect to the bin borders.

%%--------------------------------------------------------------------------
\subsection{Coded Diversity in Standard Wireless Fading Channels \label{sec_div_wireless}}
Consider a binary $[n,k,t]_2$ code \cite{MacWilliams1977}. Assume a hard-decision (algebraic) decoder with a diversity $d_0$ at its input, i.e., the bit error rate at the channel output is
$P_e=\Theta(\tfrac{1}{\gamma^{d_0}})$. After algebraic decoding with a correction radius of $t$, the probability of error per word $P_{ew}$ at the decoder output satisfies
\begin{equation}
\label{equ_pe_bounds}
P_e^{t+1} (1-P_e)^{n-t-1} \le P_{ew} \le \sum_{i=t+1}^n {n \choose i} P_e^i (1-P_e)^{n-i}.
\end{equation}
Hence, $P_{ew}$ has diversity $d_0 \times (t+1)$. The probability of error per bit after decoding, $P_{eb}$, will also achieve a diversity order of $d_0 \times (t+1)$
because $P_{eb} \le P_{ew} \le nP_{eb}$. The reader should note that it is mathematically impossible to achieve infinite diversity after such an algebraic decoding of a finite-length binary error-correcting code on a finite-diversity wireless communication channel.\\

For a non-binary $[n,k,t]_q$ code defined over a finite field $\F_q$, bounds similar to~(\ref{equ_pe_bounds}) could be established and therefore the diversity order
after algebraic decoding is also $d_0 \times (t+1)$. The non-binary nature of the error-correcting code cannot achieve an infinite diversity, not on wireless communication channels.\\

Now, assume soft-decision decoding of a binary $[n,k,t]_2$ code whose
digits are transmitted via BPSK (bipolar symbols) over the fading channel
defined in (\ref{eq_fadingchannel}) with i.i.d. fadings $\alpha_t$ following a
Nakagami-$d_0$ distribution, for $t=1 \ldots n$. 
Then, the average probability of error per bit 
before decoding is $P_e=E[Q(\sqrt{\alpha_t^2\cdot 2\gamma})]=\Theta(\tfrac{1}{\gamma^{d_0}})$, and $\gamma=E_s/N_0$ is the signal-to-noise ratio per coded bit \cite[Chapter 9]{Simon2004fading}. Conditioned on the fading coefficients, the probability of error per word after decoding satisfies
\begin{equation}
Q\left(\sqrt{\sum_{i=1}^{d_{Hmin}}\alpha_i^2 \cdot 2\gamma}\right) \le P_{ew}(\alpha) \le
\sum_{w=d_{Hmin}}^n A_w Q\left(\sqrt{\sum_{i=1}^w\alpha_i^2 \cdot 2\gamma}\right).
\end{equation}
The union bound in the right inequality involves the Hamming weight enumerator
$A(x)=\sum_{w=d_{Hmin}}^n A_w x^w$ of the binary code 
with minimum Hamming distance $d_{Hmin} \ge 2t+1$ \cite{MacWilliams1977}. 
For simplicity, without any loss after expectation, we located the $w$ positions 
from $1$ to $w$ (a Hamming distance between two distinct codewords does not always
involve the same bit positions). 
After expectation over the Nakagami-$d_0$ distributions, 
we get $P_{ew}=\E[P_{ew}(\alpha)]=\Theta(\tfrac{1}{\gamma^{d_0\cdot d_{Hmin}}})$ \cite{Al-Hussaini1985}.
Consequently, the diversity order after soft-decision decoding of a binary code 
is $d_0 \times d_{Hmin} \ge d_0 \times (2t+1)$, almost the double than the case
of algebraic hard-decision decoding. Defining the code over $\F_q$, $q>2$,
and using a larger modulation alphabet, e.g. QAM or PSK, does not change the diversity
outcome after soft-decision decoding: the diversity order is $d_0 \times d_{Hmin} \ge d_0 \times (2t+1)$. It is improved by a factor of at least $2t+1$ with respect to the uncoded diversity $d_0$ but can never be infinite because the expectation of the Gaussian tail function does not yield an exponential decay in signal-to-noise ratio. To our knowledge, there exist no cases in wireless communication channels where a finite diversity before decoding is turned into an infinite diversity after decoding, for all types of decoders.

%%---------------------------------------------------------
%%---------------------------------------------------------
\subsection{Quick Review of Information Reconciliation in TE-QKD \label{sec_reconciliation}}
The aim of the reconciliation process is to guarantee that Alice and Bob are exactly sharing the same binary digits. Errors introduced by the TE-QKD channel are to be corrected via a code $\cC$ of parameters $[n,k,d_{Hmin}]_q$ defined over a finite field $\F_q$, assumed to be of characteristic $2$ in this paper. The $k \times n$ generator matrix $G$ and the $(n-k)\times n$ parity-check matrix $H$ of $\cC$ are publicly known by all parties. Let $L=\lceil\tfrac{n\log_2(q)}{m}\rceil$ be the number of photons per codeword, for a TE-QKD frame of $N=2^m$ bins. For the sake of simplicity, we describe the reconciliation in this section when $m$ divides $n\log_2(q)$ and so $L=\tfrac{n\log_2(q)}{m} \in \Z$, $L\ge 1$. The case $\tfrac{n\log_2(q)}{m} < 1$, with $m$ multiple of $n\log_2(q)$ to ensure label alignment, can be easily derived by the reader.

Let $\phi: \F_2^m \rightarrow \Z_N$ represent the bijective map that 
converts Gray-labeled $m$ bits from the code into an integer associated to a bin of the TE-QKD frame. An illustration is given in Table~\ref{tab:graycode} for $N=8$ bins per frame, $m=3$ bits per photon, where $\phi(110)=2$ and $\phi(111)=5$. As usual,
$\phi()$ or $\phi^{-1}()$ applied to a vector in $(\F_2^m)^L$ 
or in $\Z_N^L$ indicates that the map is applied to each of the $L$ coordinates.\\

\begin{table}[!h]
\begin{center}
\caption{Gray code with $m=3$ bits for a TE-QKD frame of $N=8$ bins.\label{tab:graycode}}
\includegraphics[width=0.8\columnwidth]{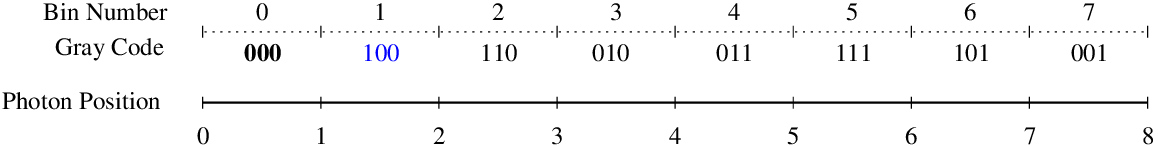}
\end{center}
%%\centering
%%\begin{tabular}{ccccccccc}
%%         Bin & 0 & 1 & 2 & 3 & 4 & 5 & 6 & 7  \\
%%         \hline 
%%         Gray Code & {\bf 000} & \textcolor{blue}{100} & 110 & 010 & 011 & 111 & 101 & 001  \\
%%    \end{tabular}
\end{table}

\begin{table}[!h]
\begin{center}
\caption{A centered version of the Gray code for $m=3$ bits per photon.\label{tab:graycode2}}
\includegraphics[width=0.8\columnwidth]{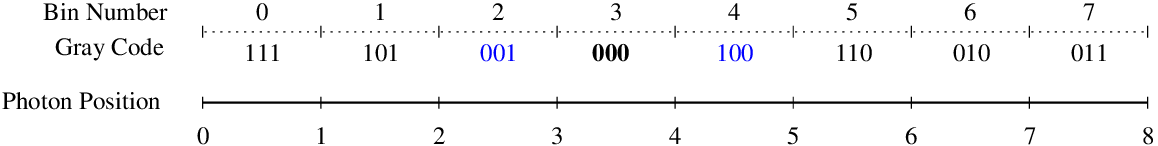}
\end{center}
%%\centering
%%\begin{tabular}{ccccccccc}
%%Bin & 0 & 1 & 2 & 3 & 4 & 5 & 6 & 7  \\
%%\hline 
%%Gray Code & 111 & 101 & \textcolor{blue}{001} & {\bf 000} & \textcolor{blue}{100} & 110 & 010 & 011  \\
%%\end{tabular}
\end{table}

Table~\ref{tab:graycode2} shows another version of a Gray code for $m=3$ where
the all-zero label is centered. Obviously, the Gray labeling of the TE-QKD bins is not unique. Fortunately, at high signal-to-noise ratio, the frame borders have
no effect on the diversity analysis. Our results do not depend on the Gray labeling version chosen by a user. Also, a property valid for all linear codes tells us that all $q^{n-k}$ cosets $e_A+\cC$, used in the reconciliation defined below by (\ref{equ_complete_dec}) and (\ref{equ_soft_recon}), admit the same Hamming distance distribution.\\ 

Alice's detector measures $L$ photon positions $x=(x_1, x_2, \ldots, x_L)$, $0 \le x_i <N$.
Alice extracts $L\cdot m= n \cdot \log_2(q)$ binary digits from $x$ by measuring the bin numbers, $\hat{x}=(\hat{x}_1, \hat{x}_2, \ldots, \hat{x}_L)$ where $\hat{x}_i=\lfloor x_i \rfloor \in \Z_N$ is a threshold detection by the floor function. Let $c_A \in \F_q^n$ be the word associated to the vector $\hat{x}$, $c_A=\phi^{-1}(\hat{x}) \in e_A+\cC$. Alice computes the syndrome $s_A=c_A\cdot H^t \in \F_q^{n-k}$ and sends $s_A$ to Bob via a public channel.\\

Bob's detector also measures $L$ time-entangled photon positions $y=(y_1, y_2, \ldots, y_L)$, $0 \le y_i <N$. Threshold detection yields $\hat{y}=(\hat{y}_1, \hat{y}_2, \ldots, \hat{y}_L)$, where $\hat{y}_i=\lfloor y_i \rfloor$. Let $c_B \in \F_q^n$ be Bob's word associated to the vector $\hat{y}$, $c_B=\phi^{-1}(\hat{y})$. 
In the ideal situation, $c_B=c_A$, or equivalently $\hat{y}=\hat{x}$, and no reconciliation is needed.
However, as described in Section~\ref{sec_model}, both detectors suffer from a jitter represented by a Gaussian noise. Given the syndrome $s_A$, the aim of the information reconciliation is to correct any errors in $c_B$ and make it identical to $c_A$ (when reconciliation succeeds). We distinguish two cases as follows.\\

{\bf Hard-decision (algebraic) decoding.} Maximum {\em a posteriori} (MAP) or equivalently maximum likelihood (ML) decoding in the Hamming space $\F_q^n$ is considered. The TE-QKD information-theoretical channel input is $\hat{X} \in \Z_N^L$ (equivalent to $c_A \in \F_q^n$), and the channel output is $\hat{Y} \in \Z_N^L$ (equivalent to $c_B \in \F_q^n$). Given the syndrome $s_A$ and given $\hat{Y}=\hat{y}=\phi(c_B)$, Bob's task is to decode $c_B$ in the coset $e_A+\cC$, where
$s_A=e_AH^t$ and $e_A\in \F_q^n$ is a representative of that coset. Here, Bob applies an algebraic (hard-decision input) decoder that minimizes the Hamming distance to decode $c_B$ into $\hat{c}_A$ via coset decoding
\begin{equation}
\label{equ_complete_dec}
\hat{c}_A = \arg \min_{\upsilon \in e_A+\cC} d_H(\upsilon, c_B).
\end{equation}
The algebraic decoder described in (\ref{equ_complete_dec}) is a complete decoder, capable of decoding wherever is the position of $c_B$ in the Hamming space $\F_q^n$. For small codes, this complete decoding could be achieved via the Slepian table \cite{Slepian1956}\cite[\S 3.4]{Blahut2003}. The latter is a decomposition of $\F_q^n$ into cosets of the code $\cC$, i.e.,
$\F_q^n=\cC+\F_q^n/\cC$, where $\F_q^n/\cC$ is the quotient group of the whole space by the code. As a first example, consider the $[3,1,3]_2$ code $\cC=\{000, 111\}$ listed on the first row of the Slepian table in Figure~\ref{fig_slepian_3_1_3}. We have $c_A=000$ and $c_B=001$.
In this case, (\ref{equ_complete_dec}) correctly decodes $c_B$ into $c_A$ since $d_H(001,000)=1$ while $d_H(001,111)=2$. As a second example,
we consider the $[6,3,3]_2$ code, a shortened version of the famous length-$7$ binary Hamming code. As illustrated in Figure~\ref{fig_slepian_6_3_3} where the code fills the first row of the Slepian table, we have $c_A=110~111$ and $c_B=110~110$.
The decoder of (\ref{equ_complete_dec}) also yields the correct answer by flipping the last bit in $c_B$. Both examples assume a TE-QKD frame of $N=8$ bins, or equivalently $m=3$ bits per photon. For large codes, complete decoding via the Slepian table becomes intractable, so fast algebraic decoders are used to solve $e$ from the syndrome equation $s=e\cdot H^t$. Hence, in general, we have $c_A=c+e_A$ with syndrome $s_A=e_A \cdot H^t$ and $c_B=c+e_B$ with syndrome $s_B=e_B \cdot H^t$,
where $c \in \cC$. Bob computes $s_A-s_B=(e_A-e_B)H^t$ and finds
$e_A-e_B$ using any algebraic decoder (complete or $t$-bounded, $2t+1\le d_{Hmin}$, or beyond $t$). Then, Bob does the final reconciliation step by adding $c_B+(e_A-e_B)=c+e_B+(e_A-e_B)=c+e_A=c_A$. Since all finite fields considered in this paper are of characteristic $2$, the "-" sign
can be replaced by the "+" sign.

\begin{figure}[!h]
\begin{center}
\includegraphics[width=0.5\columnwidth]{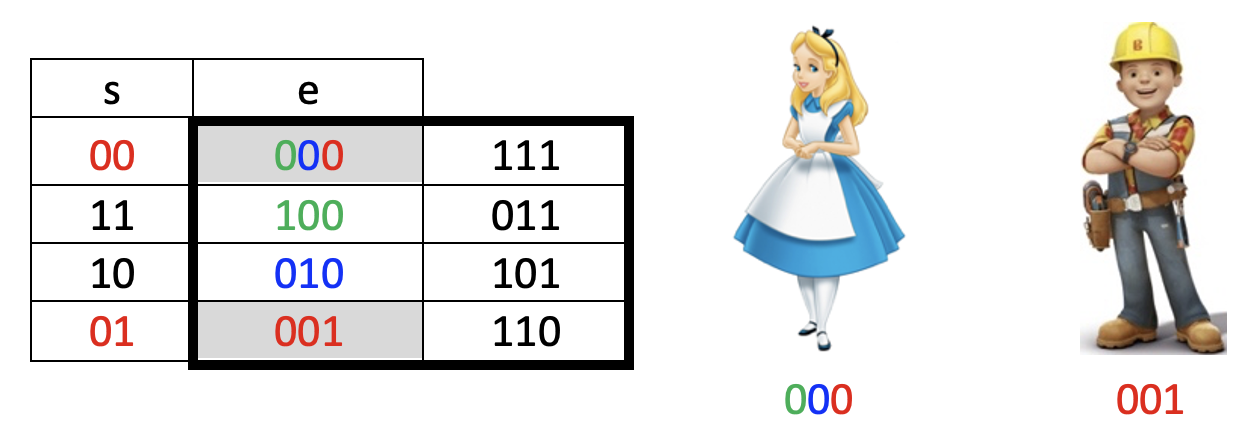} 
\caption{Slepian table of the binary repetition $[3,1,3]_2$ code (within the bold frame) and the syndrome on the left.
\label{fig_slepian_3_1_3}}
\end{center}
\end{figure}

\begin{figure}[!h]
\begin{center}
\includegraphics[width=0.85\columnwidth]{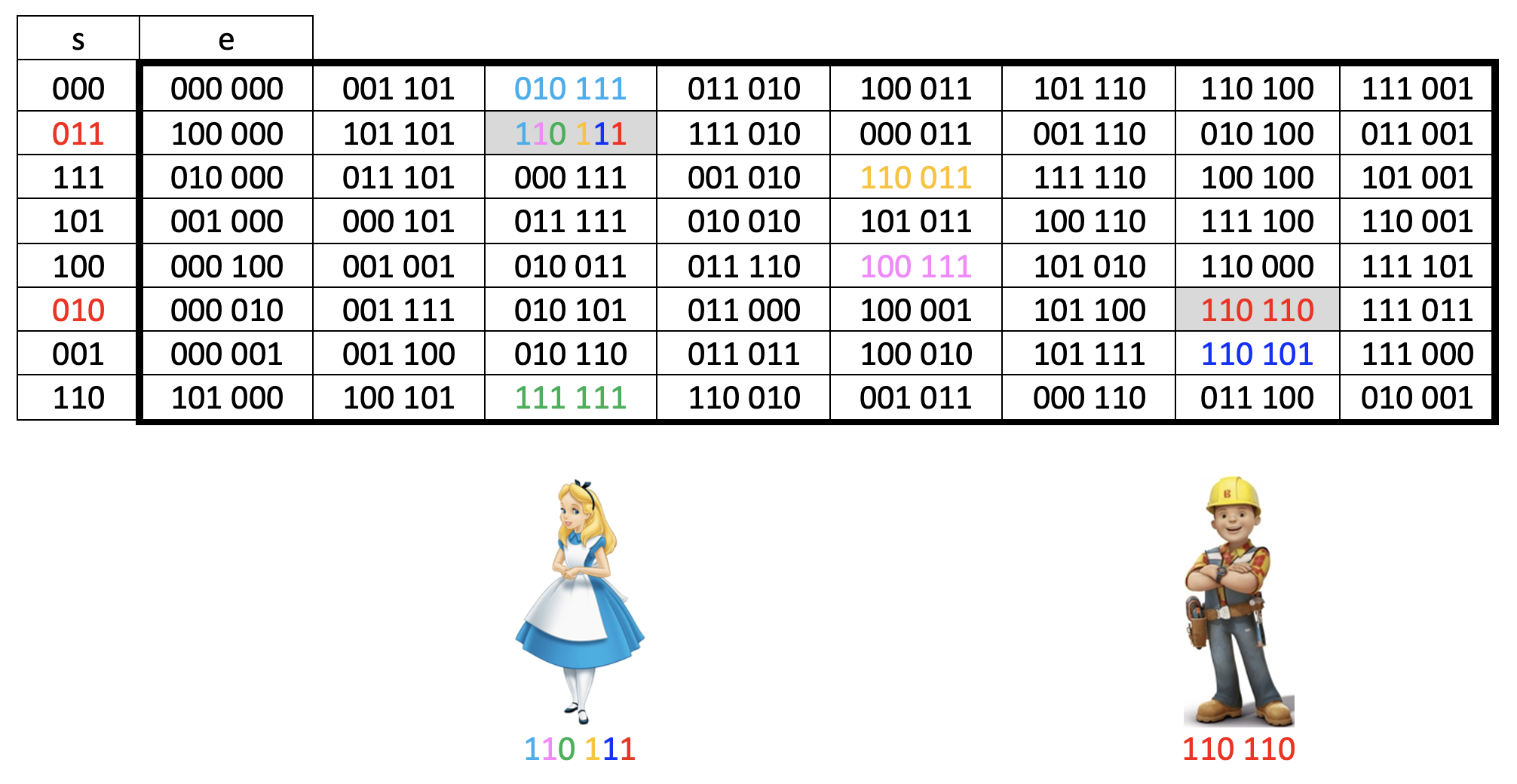} 
\caption{Slepian table of the shortened Hamming code $[6,3,3]_2$ (within the bold frame) and the syndrome on the left.
\label{fig_slepian_6_3_3}}
\end{center}
\end{figure}

In the first example of Figure~\ref{fig_slepian_3_1_3}, $s_A+s_B=00+01=01$ which corresponds to $e=001$, so we flip the third bit to get $c_A$ from $c_B$. In the second example of Figure~\ref{fig_slepian_6_3_3}, $s_A+s_B=011+010=001$ which corresponds to $e=000 001$, so we flip the sixth bit to get $c_A$ from $c_B$.\\

{\bf Soft-decision decoding.} The TE-QKD information-theoretical channel input is $\hat{X} \in \Z_N^L$ (equivalent to $c_A \in \F_q^n$), but now the channel output is soft given by $Y = (Y_1, \cdots, Y_L) \in [0, N[^L$ instead of the hard bin indices $\hat{Y}$. Given that the decoder input is now soft, the maximum {\em a posteriori} decoder that minimizes
the probability of error per word $P_{ew}$ is not based anymore on the Hamming distance as in (\ref{equ_complete_dec}).\\

For $\upsilon=(\upsilon_1, \upsilon_2, \ldots, \upsilon_L) \in e_A+\cC \in \F_q^n \sim \F_2^{n\log_2(q)}=\F_2^{mL}$ in the coset of Alice, $\upsilon_{\ell} \in \F_{2}^m$ and $mL=n\log_2(q)$, assuming that the channel introduces no memory between the $L$ photons, define the {\em a posteriori} probability as
\begin{equation}
\label{equ_APP_word}
APP(\upsilon)=\pr(\hat{c}_A=\upsilon|Y=y)
=\prod_{\ell=1}^L \pr(\phi(\upsilon_{\ell})|y_{\ell})
=\prod_{\ell=1}^L APP(\phi(\upsilon_{\ell})), ~~~\upsilon_{\ell}\in \F_2^m, ~~~\upsilon \in \F_2^{mL}.
\end{equation}
For $\hat{y}_{\ell}=j$ and $\phi(\upsilon_{\ell})=i \in \Z_N$, 
the APP could be solved
via Bayes rule from the likelihood expression 
in~(\ref{equ_Y_cond_hX_tY}) listed in Appendix~\ref{app_fundamental}. 
In this paper, we will use the simplified expressions (48) and (49) from~\cite{Boutros2023}: 
\begin{align}
\text{If}~i=j,~~APP(i) &\propto \left[1-\tfrac{1}{2}e^{-\tfrac{(y_{\ell}-i)^2}{2\sigma^2}}
-\tfrac{1}{2}e^{-\tfrac{(y_{\ell}-i-1)^2}{2\sigma^2}}\right], \label{equ_app1_approx}\\
\text{If}~i\ne j,~~APP(i) &\propto sign(j-i) \cdot \left[\tfrac{1}{2}e^{-\tfrac{(y_{\ell}-i-1)^2}{2\sigma^2}}
-\tfrac{1}{2}e^{-\tfrac{(y_{\ell}-i)^2}{2\sigma^2}}\right]. \label{equ_app2_approx}
\end{align}
For a better clarity, if needed, the notations $APP(i)$ could be replaced by $APP(\upsilon_{\ell}|y_{\ell})$ or $APP(i,j|y_{\ell})$.
Following the three expressions (\ref{equ_APP_word})--(\ref{equ_app2_approx}), for soft information reconciliation, the MAP decoder does
\begin{equation}
\label{equ_soft_recon}
\hat{c}_A = \arg \max_{\upsilon \in e_A+\cC} APP(\upsilon).
\end{equation}

If the code $\cC$ has a tractable trellis representation, variations of the MAP decoder could include a word-wise Viterbi algorithm applied to minimize the metric $-\sum_{\ell=1}^L \log(APP(\phi(\upsilon_{\ell})))$.
At high signal-to-noise ratio, it is well known that the BCJR algorithm (the forward-backward algorithm) behaves exactly like the Viterbi algorithm, with respect to the best route in the trellis. Thus, it is possible to perform soft decoding
via bit-wise BCJR without any diversity loss. In all cases, as indicated
in Section~\ref{sec_div_wireless}, the probability of error per bit $P_{eb}$ and the probability of error per word $P_{ew}$ exhibits an identical diversity, whether diversity is finite or infinite.
Finally, when belief propagation is used
to decode an LDPC code, which is sub-optimal in the block-wise MAP sense, 
the syndrome $s_A$ is used to flip the parity of the
check equations in the Tanner representation of the code.

%%\nrd{Oops. Change all previous dH into dHmin. oh la la}.\\

%%---------------------------------------------------------
%%---------------------------------------------------------
\subsection{Coded Diversity in TE-QKD Channels \label{sec_coded_div_QKD}}
While the TE-QKD channel is free from multipath fading, its analogous behavior necessitates diversity analysis. The ``fading'' in TE-QKD is not caused by signal strength fluctuation but by the random position of the jitter-free photon, $U$, relative to the discrete bin boundaries. A ``deep fade'' in TE-QKD occurs when the true photon position $U$ is extremely close to a bin boundary, where the system is vulnerable. Even a very small amount of detector jitter can be enough to push the measured position across the boundary, causing a symbol error. Because the position $U$ is a uniform random variable, there is always a non-zero probability of it being near a boundary. As shown in \cite[Propositions 1\&2]{Boutros2023}, for \emph{uncoded} TE-QKD (hard decisions), the symbol error probability is 
\begin{align}
    P_e(\gamma) = \frac{2}{\sqrt{\pi} } (1- \frac{1}{N}) \gamma^{-\frac{1}{2}} + \cO\left( \exp(-\frac{\gamma}{4})\right),
    \label{eq:uncoded-Pe}
\end{align}
achieving a diversity order of  $d_0=\frac{1}{2}$. Following the wireless analogy, we can apply a standard error-correcting code to increase the diversity order, which has been confirmed by~\cite{Boutros2023}
where Reed-Solomon codes ($n=63$ over $GF(64)$), binary Bose-Chaudhuri-Hocquenghem (BCH) codes ($n=378$ bits),
and binary low-density parity-check (LDPC) codes ($n=384$ and $9999$ bits) were employed in a similar manner as on wireless communication channels. An illustration is given here with a binary shortened $[30,20,t=2]$ BCH code. Figure~\ref{fig_30_20}
shows the performance of reconciliation on TE-QKD for both algebraic and soft-decision decoding. As expected, diversity becomes $\tfrac{1}{2} \times (t+1) =1.5$ after algebraic decoding and $\tfrac{1}{2} \times (2t+1) = 2.5$ after soft-decision decoding of the $[30,20,t=2]$ code. 

As a final illustration in this section, we consider a much shorter code:
the length-3 binary repetition $[3,1,3]$ code considered in Section~\ref{sec_reconciliation}. Its error probability performance is shown versus signal-to-noise ratio in Figure~\ref{fig_3_1_3}. 
Shockingly, the plots are out of control! Diversity orders of $1$ and $1.5$ were expected for algebraic and soft-decision decoding respectively.
The observed effective diversity is $d=10$ between $10^{-6}$ and $10^{-7}$.
The probability of error follows the shape of a Gaussian tail function,
suggesting that diversity is infinite. In the sequel, we prove how infinite diversity is created after decoding from a finite diversity channel,
and we establish necessary and sufficient conditions for infinite diversity taking into account the TE-QKD frame structure and the code structure. 

\begin{figure}[!h]
\begin{center}
\vspace{-8mm}
\includegraphics[angle=270, width=0.7\linewidth]{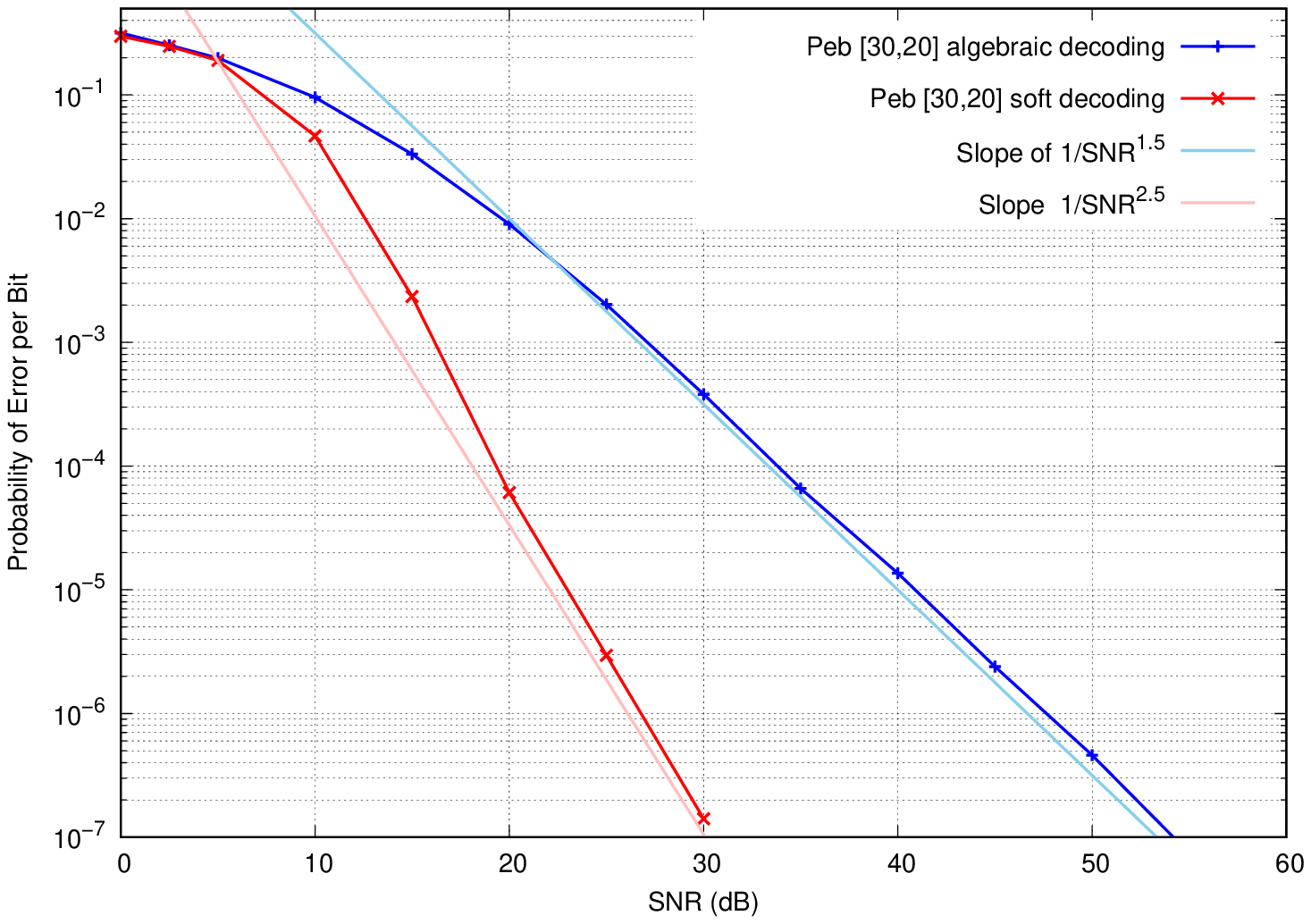} 
\vspace{9mm}
\caption{Hard-decision (algebraic) versus soft-decision of a binary shortened $[30,20,t=2]$ BCH code, $N=8$ bins per frame.
\label{fig_30_20}}
\end{center}
\end{figure}

\begin{figure}[!h]
\begin{center}
\vspace{-8mm}
\includegraphics[angle=270, width=0.7\linewidth]{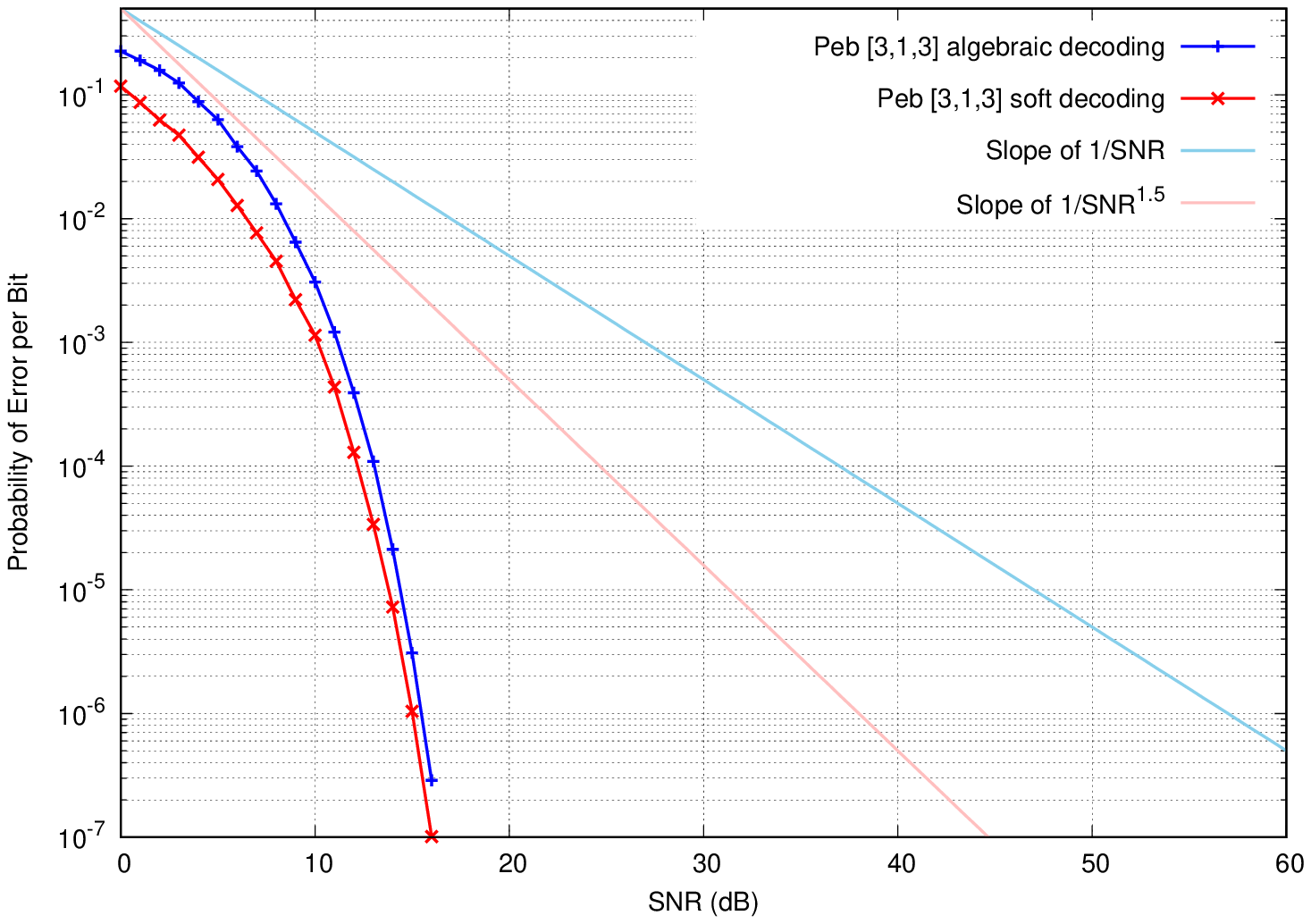} 
\vspace{9mm}
\caption{Hard-decision (algebraic) versus soft-decision of a 3-fold repetition binary $[3,1,3]$ code, $N=8$ bins per frame.
\label{fig_3_1_3}}
\end{center}
\end{figure}

%%--------------------------------------------------------------
%%--------------------------------------------------------------
\clearpage
\section{Probability of Single and Multiple Jumps \label{sec_probability}}

It is pivotal to understand events within a TE-QKD frames 
in order to establish the conditions leading to infinite diversity after decoding. Given the exact photon position $U$, Proposition~1 in \cite{Boutros2023} neglects double bin jumps of $X$ and $Y$ on Alice's
and Bob's sides respectively, see Figure~3 in \cite{Boutros2023}. 
Similarly, Proposition~2 in \cite{Boutros2023} at high SNR states 
that the dominant term in $p_{ij}=\pr(\hY=j|\hX=i)$ is $\tfrac{1}{\sqrt{\pi \gamma}}$ for $|i-j|=1$, while $p_{ij}=\cO(e^{—\tfrac{\gamma}{4}})$ could be forced to zero in numerical calculations when $|i-j|\ge 2$. As we will see later, after establishing the infinite diversity conditions, neglecting double jumps and longer bin jumps $|i-j|\ge 2$
has no influence on estimating the performance of long error-correcting codes, since long codes do not usually satisfy infinite diversity conditions.

The fundamental equations that govern the TE-QKD channel are listed 
in Appendix~\ref{app_fundamental}. Exact equations in \cite{Boutros2023}
are mandatory to compute exact mutual information and the exact channel capacity. In this paper, we deal with diversity at high signal-to-noise ratio. It is sufficient to reconsider fundamental equations and let $\sigma^2$ vanish to extract simpler models that are strong enough to determine whether diversity is finite or infinite, models that are also accurate enough to indicate the rate of decay of the probability of main events. At high signal-to-noise ratio, firstly we consider all $N$ bins to be equiprobable thanks to (\ref{equ_hpi}) \& (\ref{equ_hpi_simple}) in Appendix~\ref{app_fundamental}-a. 
Secondly, if $\hX=i$, we consider that $X$ is uniform in the interval $[i, i+1[$ and $Y=U+Z_2$, $Z_2 \sim \mathcal{N}(0,\sigma^2)$. Indeed, 
the integral in the numerator of (\ref{equ_Y_cond_hX_tY}) in Appendix~\ref{app_fundamental}-b becomes a convolution of the $\mathcal{N}(0,\sigma^2)$ Gaussian density with a uniform density 
over $[i, i+1[$ when $\sigma^2 \ll 1$. By symmetry, we also get the high-SNR assumption $X=U+Z_1$, $Z_1 \sim \mathcal{N}(0,\sigma^2)$, $Z_1$ and $Z_2$ independent. Thirdly, instead of using (\ref{equ_pij}) from Appendix~\ref{app_fundamental}-c, we will re-establish the high-SNR expression of $p_{ij}$ via two simple models described below, emphasizing the critical events inside the TE-QKD frame. Finally, conditioned on both Alice's and Bob's frames being valid, we will assume that $U$ is uniform in $[0,N[$ as a result of (\ref{equ_pdfU_valid}) \& (\ref{equ_pdfU_highsnr}) in Appendix~\ref{app_fundamental}-d. 

Consequently, from the fundamental equations, we establish two equivalent models valid at high signal-to-noise ratio $\gamma=\tfrac{1}{\sigma^2} \gg 1$, or equivalently in the small noise regime $\sigma^2 \ll 1$:
\begin{itemize}
\item $U-X-Y$ model: $X=U+Z_1$, $Y=U+Z_2$, $U \sim \mathrm{Unif}([0,N[)$, 
$Z_1$ and $Z_2$ are i.i.d. $\mathcal{N}(0, \sigma^2)$. 
\item $X-Y$ model: $Y=X+Z$, where $Z$ is $\mathcal{N}(0, 2\sigma^2)$ (from $Z=Z_2-Z_1$), and $X \sim \mathrm{Unif}([0,N[)$. 
\item The frame border effects are neglected in both models. 
\end{itemize}

In the sequel, if $h(\gamma)=f(\gamma)\pm g(\gamma)$, where 
$g(\gamma)=\cO(e^{-\alpha \gamma})$, for $\alpha>0$ and $\gamma \gg 1$, we will write
$h(\gamma)=f(\gamma)+\cO(e^{-\alpha \gamma})$ since the absolute
value of the vanishing term is the most relevant, not its sign.
%%-----------------------------------------
%%-----------------------------------------
\begin{figure}[!t]
\begin{center}
\includegraphics[width=8cm]{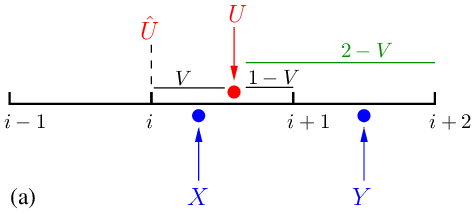}~~~~
\includegraphics[width=8cm]{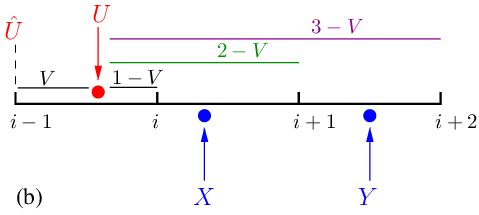}
\caption{Frame illustration for (a) $\hU=i$ and (b) $\hU=i-1$, while $\hX=i$ and $\hY=i+1$.\label{fig_lemma1}}
\end{center}
\end{figure}
%%-----------------------------------------
%%-------------- Lemma 1 ------------------
%%-----------------------------------------
\begin{lemma}
\label{lem_U-X-Y-1-jump}
\normalfont
    Consider the $U-X-Y$ model where $i \le X < i+1$ and $i+1 \le Y <i+2$. Thus, $\hat{X} =i$ and $\hat{Y} = i+1$. Then, we have 
    \begin{align}
    \label{equ_1jump_final}
        \pr( \hat{Y} = i+1 ~|~ \hat{X} =i) = \frac{1}{\sqrt{\pi \gamma}} + \cO( e^{-\frac{\gamma}{4}}) =\Theta(\gamma^{-1/2}).
    \end{align}
\end{lemma}
\begin{IEEEproof}
We start with the conditional probability of the measured bin position by Alice and Bob given the exact photon's position $U$. Assuming $\hU=i$ as illustrated in Figure~\ref{fig_lemma1}-(a),
we can write
\begin{align*}
        \pr( \hat{X} = i ~|~ U = u) &= 1 - \pr( X \notin [i, i+1[ ~|~u )
        \\& = 1 - \pr(X<i ~|~ U = u) - \pr( X \geq i+1 ~|~ U = u)
        \\& = 1 - Q\left(\frac{v}{\sigma}\right) - Q\left(\frac{1-v}{\sigma} \right)
\end{align*}
    where $v$ is the realization of $V = ( U$ mod $[0,1])$ given $\hat{U} = i$, i.e., the distance between $U$ (the exact photon position) and the left border of the bin $\hat{U}$. 
Similarly, we have
\begin{align}
    \pr(\hat{Y} = i+1 ~|~ U = u) = Q\left(\frac{1-v}{\sigma}\right) - Q\left(\frac{2-v}{\sigma} \right),
\end{align}
for $\hU=i$. The case $\hU=i+1$ is identical 
to the case $\hU=i$ by symmetry, you just need to switch the expressions of $\pr(\hat{X} = i ~|~ U = u)$ and $\pr(\hat{Y} = i+1 ~|~ U = u)$.

\noindent
Now assume that $\hU=i-1$ as illustrated in Figure~\ref{fig_lemma1}-(b), then we get
\begin{equation}
\label{pr_X_i_U_i-1}
\pr(\hat{X} = i ~|~ U = u) = Q\left(\frac{1-v}{\sigma}\right) - Q\left(\frac{2-v}{\sigma} \right)
\end{equation}
and
\begin{equation}
\label{pr_Y_i_U_i-1}
\pr(\hat{Y} = i+1 ~|~ U = u) = Q\left(\frac{2-v}{\sigma}\right) - Q\left(\frac{3-v}{\sigma} \right).
\end{equation}
The case $\hU=i+2$ is identical 
to the case $\hU=i-1$ by symmetry,
we just switch (\ref{pr_X_i_U_i-1}) and (\ref{pr_Y_i_U_i-1}).

\noindent
The probability of $\hY$ conditioned on $\hX$ becomes
\begin{align}
    \pr( \hat{Y} = i+1 | \hat{X} =i )
    & = \int_{u= 0}^{N} \pr(\hat{Y} = i+1, U = u  | \hat{X} =i ) du
     \notag
    \\&
    = \int_{u=0}^{N} \frac{\pr(\hat{Y} = i+1, \hat{X} =i, U = u  ) }{\pr(\hat{X} = i )}du
    \notag
    \\& 
    = 
    \int_{u=0}^{N}  \frac{\pr(\hat{Y} = i+1, \hat{X} =i| U = u  ) }{\pr(\hat{X} = i )} p(u) du
    \notag
    \\& 
    = 
    \int_{u=0}^{N}  \pr(\hat{Y} = i+1, | U = u  )  \pr(\hat{X} = i| U=u) du
\label{eq:1/N} 
\\&
= \sum_{j=0}^{N-1} \int_{u=j}^{j+1} \pr(\hat{Y} = i+1, | U = u  )  \pr(\hat{X} = i| U=u) du
\label{eq_sum_overj} 
\end{align}
where~\eqref{eq:1/N} holds since $p(u) = \pr(\hat{X} =i ) =\frac{1}{N}$ according to our high-SNR assumptions. Furthermore, the two events $\hX=i$ and $\hY=i+1$ are independent when conditioning on $U$.   
We neglect bin jumps of three bins and more in (\ref{eq_sum_overj}), events with probability bounded from above by $Q(2/\sigma)=\cO(e^{-2\gamma})$, to obtain
\begin{align}
    \pr( \hat{Y} = i+1 | \hat{X} =i ) 
    &= 
    \sum_{j=i-1}^{i+2} \int_{u=j}^{j+1} \pr(\hat{Y} = i+1 |U =u ) \pr(\hat{X} = i | U =u) du ~+~ \cO(e^{-2\gamma}) \notag 
    \\& 
    = 2 \int_{0}^1 \left[ Q\left(\frac{1-v}{\sigma}\right) - Q\left(\frac{2-v}{\sigma} \right) \right] 
    \cdot 
    \left[ Q\left(\frac{2-v}{\sigma}\right) - Q\left(\frac{3-v}{\sigma} \right)   \right] dv 
    \label{equ_1jump_i-1}
    \\&
    \quad + 
    2 \int_{0}^1 \left[ 1 - Q\left(\frac{v}{\sigma}\right) - Q\left(\frac{1-v}{\sigma} \right) \right] 
    \cdot 
    \left[ Q\left(\frac{1-v}{\sigma}\right) - Q\left(\frac{2-v}{\sigma} \right)   \right] dv 
    \label{equ_1jump_i}
    \\& \quad + \cO(e^{-2\gamma}) \nonumber
\end{align}
where~\eqref{equ_1jump_i-1} is for $j=i-1$ and $j = i+ 2$ and~\eqref{equ_1jump_i} is for $j = i$ and $j = i+1$. Then,
\begin{align*}
    \pr( \hat{Y} = i+1 | \hat{X} =i ) 
    &= 
    2 \int_{0}^1 Q\left(\frac{1-v}{\sigma}\right) dv
    -
    2 \int_{0}^1 Q\left(\frac{2-v}{\sigma}\right) dv
    \\&
    \quad
    - 2 \int_{0}^1 Q\left(\frac{v}{\sigma}\right)
    Q\left(\frac{1-v}{\sigma}\right) dv
    +
    2 \int_{0}^1 Q\left(\frac{v}{\sigma}\right)
    Q\left(\frac{2-v}{\sigma}\right) dv
    \\&
    \quad
    - 2 \int_{0}^1 Q\left(\frac{1-v}{\sigma}\right)
    Q\left(\frac{1-v}{\sigma}\right) dv
    +
    4 \int_{0}^1 Q\left(\frac{1-v}{\sigma}\right)
    Q\left(\frac{2-v}{\sigma}\right) dv
    \\&
    \quad
    - 2 \int_{0}^1 Q\left(\frac{1-v}{\sigma}\right)
    Q\left(\frac{3-v}{\sigma}\right) dv
    -
    2 \int_{0}^1 Q\left(\frac{2-v}{\sigma}\right)
    Q\left(\frac{2-v}{\sigma}\right) dv
    \\&
    \quad
    + 2 \int_{0}^1 Q\left(\frac{2-v}{\sigma}\right)
    Q\left(\frac{3-v}{\sigma}\right) dv + \cO(e^{-2\gamma})
\end{align*}  
\begin{align*}  
    &=
    2 \frac{1}{\sqrt{2\pi \gamma}} 
    + \cO( e^{- \frac{\gamma}{2}}) + \cO( e^{- \frac{\gamma}{2}})
    + \cO( e^{- \frac{\gamma}{4}})+ \cO( e^{-\gamma})
    -2 \frac{(\sqrt{2} -1)}{2 \sqrt{\pi \gamma}} + \cO( e^{- \gamma})\\
    &\quad + \cO( e^{- \frac{\gamma}{2}})
    + \cO(e^{- 2\gamma })
    + \cO(e^{- \gamma })
    + \cO( e^{- \frac{5}{2}\gamma})
    + \cO(e^{- 2\gamma }) \\
    &= \frac{1}{\sqrt{\pi \gamma}} + \cO( e^{- \frac{\gamma}{4}}),
    ~~~\text{as stated by the Lemma}.
\end{align*}  
The big $\cO$ notations of each integration term are derived in Appendix~\ref{app:bigO}.
\end{IEEEproof}

~\\
\noindent
The next Lemma derives an identical result via the simplified model.
%%-----------------------------------------
%%-------------- Lemma 2 ------------------
%%-----------------------------------------
\begin{lemma}
\label{lem_X-Y-1-jump}
\normalfont
    Consider the $X-Y$ model  where  $i \le X < i+1$ and $i+1 \le Y <i+2$, as in Lemma~\ref{lem_U-X-Y-1-jump}. 
    Thus, $\hat{X} =i$ and $\hat{Y} = i+1$. 
    This model yields an identical result, i.e., 
    \begin{align}
        \pr( \hY = i+1 ~|~ \hX =i) = \frac{1}{\sqrt{\pi \gamma}} + \cO( e^{-\frac{\gamma}{4}}) = \Theta(\gamma^{-1/2}).
    \end{align}
\end{lemma}
\begin{IEEEproof}
According to the $X-Y$ model we have
    \begin{align}
        \pr( \hY = i+1 ~|~ X= x)
        = Q\left( \frac{i+1 -x}{ \sigma \sqrt{2}}\right) 
        - 
        Q\left( \frac{i+2 -x}{ \sigma \sqrt{2}}\right),
    \end{align}
so we get
\begin{align*}
        \pr(\hY=i+1 ~|~ \hX=i)
        &= \int_{x=i}^{i+1}
        \left[ 
        Q\left( \frac{i+1 -x}{ \sigma \sqrt{2}}\right) 
        - 
        Q\left( \frac{i+2 -x}{ \sigma \sqrt{2}}\right)
        \right] dx
        %\notag
        \\&
       =
       \int_{v=0}^{1}
        \left[ 
        Q\left( \frac{1 -v}{ \sigma \sqrt{2}}\right) 
        - 
        Q\left( \frac{2 -v}{ \sigma \sqrt{2}}\right)
        \right] dv
        \\& 
        =
        \frac{1}{\sqrt{\pi \gamma}} + \cO(e^{-\frac{\gamma}{4}}).
\end{align*}
Note that the factor $\sqrt{2}$ of $\sigma$ in the above equations is due to the double variance of the additive Gaussian noise in the $X-Y$ model. 
The last result and the big $\cO$ notations are derived in Appendix~\ref{app:bigO}.\\
\end{IEEEproof}

%--------------------------------------------
Lemma~\ref{lem_X-Y-1-jump} requires less algebra than Lemma~\ref{lem_U-X-Y-1-jump} to yield an identical result for the transition probability $\pr(\hY=i+1 ~|~ \hX=i)$ in the algebraic TE-QKD channel model. Obviously, we should prefer Lemma~\ref{lem_X-Y-1-jump} to undergo performance analysis. However, the true events behind the transition error are only visible in Lemma~\ref{lem_U-X-Y-1-jump}.
Indeed, (\ref{equ_1jump_final}) results from (\ref{equ_1jump_i}), since (\ref{equ_1jump_i-1}) has a higher decay rate. 
Half of the term $\tfrac{1}{\pi\gamma}$ in (\ref{equ_1jump_final}) comes from Alice's detector not changing the measured photon position while Bob's detector making it jump by one bin,
and from both detectors jumping by one bin on the same side. 
The second half of $\tfrac{1}{\pi\gamma}$ is obtained by switching Alice and Bob in the photon detection events.
The dominant infinite-diversity term $\cO(e^{-\gamma/4})$ results from $\int_{0}^1 Q\left(\frac{v}{\sigma}\right) Q\left(\frac{1-v}{\sigma}\right) dv$, which is equivalent to an event where Alice's detector makes one error to the left by one bin and Bob's detector makes one error to the right by one bin.
%%-----------------------------------------
%%-------------- Lemma 3 ------------------
%%-----------------------------------------
\clearpage
\begin{lemma}
\label{lem_U-X-Y-2-jump}
\normalfont
   Consider the $U-X-Y$ model  where  $i \le X < i+1$ and $i+2 \le Y <i+3$. Thus,  $\hat{X} =i$ and $\hat{Y} = i+2$. Then, we have
   \begin{align}
   \label{equ_p_i+2_i}
       \pr( \hat{Y} = i+2 | \hat{X} =i) = \cO( e^{-\frac{\gamma}{4}}).
   \end{align}
\end{lemma}
\begin{IEEEproof}
    According to the $U-X-Y$ model and similar to the proof of Lemma~\ref{lem_U-X-Y-1-jump}, we have
    \begin{align}
        \pr( \hat{Y} = i+2 | \hat{X} =i) 
        & =
        \int_{u=0}^N \pr( \hat{Y} = i+2, U = u | \hat{X} =i) du 
        \notag
        \\&
        = \sum_{j=i}^{i+2} \int_{u=j}^{j+1} \pr( \hat{Y} = i+2, U = u | \hat{X} =i) du
        ~+~\cO(e^{-2\gamma})
        \notag
        \\& 
        = \sum_{j=i}^{i+2} \int_{u=j}^{j+1} \pr( \hat{Y} = i+2 | U = u ) \pr(\hat{X} =i | U =u) du
        ~+~\cO(e^{-2\gamma})
        \label{eq:follow-example1}
        \\& 
        = 2 \int_{v=0}^1 \left[ 
        Q\left( \frac{2-v}{ \sigma} \right)
        - 
        Q\left( \frac{3-v}{ \sigma} \right)
        \right] 
        \cdot 
        \left[ 
        1 - Q\left( \frac{v}{ \sigma} \right)
        - 
        Q\left( \frac{1-v}{ \sigma} \right)
        \right] dv 
        \label{eq:example3-j=i}
        \\& 
        %%%%%%%%%%%%
        \quad + 
        \int_{v=0}^1 \left[ 
        Q\left( \frac{1-v}{ \sigma} \right)
        - 
        Q\left( \frac{2-v}{ \sigma} \right)
        \right] 
        \cdot 
        \left[ 
         Q\left( \frac{v}{ \sigma} \right)
        - 
        Q\left( \frac{1+ v}{ \sigma}\right)
        \right] dv \label{eq:example3-j=i+1}
        \\ & \quad +~\cO(e^{-2\gamma}), \nonumber
    \end{align}
    where~\eqref{eq:follow-example1} follows from \eqref{eq:1/N} in the proof of Lemma~\ref{lem_U-X-Y-1-jump}, \eqref{eq:example3-j=i} is the sum for $j=i$ and $j = i+2$ by symmetry, and \eqref{eq:example3-j=i+1} is for $j= i+1$.
    Then, 
    \begin{align*}
        \pr( \hat{Y} = i+2 | \hat{X} =i) 
        & = 2 \int_{v=0}^1 Q\left( \frac{2-v}{ \sigma} \right) dv 
        -
        2 \int_{v=0}^1 Q\left( \frac{3-v}{ \sigma} \right) dv
        -
        2 \int_{v=0}^1 Q\left( \frac{2-v}{ \sigma} \right) Q\left( \frac{v}{ \sigma} \right) dv
        \\&
         \quad + 
         2 \int_{v=0}^1 Q\left( \frac{3-v}{ \sigma} \right) Q\left( \frac{v}{ \sigma} \right) dv
         -
         2 \int_{v=0}^1 Q\left( \frac{2-v}{ \sigma} \right) Q\left( \frac{1 - v}{ \sigma} \right) dv
         \\& 
         \quad +
         2 \int_{v=0}^1 Q\left( \frac{3-v}{ \sigma} \right) Q\left( \frac{1-v}{ \sigma} \right) dv
         +
          \int_{v=0}^1 Q\left( \frac{1-v}{ \sigma} \right) Q\left( \frac{v}{ \sigma} \right) dv
          \\&
          \quad 
          -  \int_{v=0}^1 Q\left( \frac{2-v}{ \sigma} \right) Q\left( \frac{v}{ \sigma} \right) dv
          -  \int_{v=0}^1 Q\left( \frac{1-v}{ \sigma} \right) Q\left( \frac{1+v}{ \sigma} \right) dv
          \\&
          \quad
          + \int_{v=0}^1 Q\left( \frac{2-v}{ \sigma} \right) Q\left( \frac{1 + v}{ \sigma} \right) dv
          \\&
          = \cO( e^{-\frac{\gamma}{2}}) 
          + \cO( e^{-2\gamma}) 
          + \cO( e^{-\gamma}) 
          + \cO( e^{-\frac{5}{2}\gamma}) % 3-v v
          + \cO( e^{-\frac{\gamma}{2}}) %2-v 1-v
          + \cO( e^{-2\gamma}) %3-v 1-v
          \\&
          \quad
          + \cO( e^{-\frac{\gamma}{4}}) %1-v v
          + \cO( e^{-\gamma }) %2-v v
          +\cO( e^{-\gamma}) %1-v 1+v
          + \cO( e^{-\frac{9}{2}\gamma}) %2-v 1+v
          \\&
          = \cO( e^{-\frac{\gamma}{4}}).
    \end{align*}
    The big $\cO$ notations of each integration term are derived in Appendix~\ref{app:bigO}.
\end{IEEEproof}

~\\~\\
\noindent
From the proof of Lemma~\ref{lem_U-X-Y-2-jump}, it is understood that (\ref{equ_p_i+2_i}) is the result of Alice's detector making an error by one bin in one direction, and Bob's detector making an error by one bin in the opposite direction. Lemma~\ref{lem_X-Y-2-jump} yields (\ref{equ_p_i+2_i}) but without revealing which photon events created this transition probability.

%%-----------------------------------------
%%-------------- Lemma 4 ------------------
%%-----------------------------------------
\begin{lemma}
\label{lem_X-Y-2-jump}
\normalfont
    Consider the $X-Y$ model in where  $i \le X < i+1$ and $i+2 \le Y <i+3$. Thus, $\hat{X} =i$ and $\hat{Y} = i+2$.  Then, we have
    \begin{align}
        \pr( \hat{Y} = i+2 | \hat{X} =i)= \cO( e^{-\frac{\gamma}{4}}).
        \label{eq:X-Y-2-jump}
    \end{align}
\end{lemma}

\begin{IEEEproof}
    In this case, 
    \begin{align*}
        \pr( \hat{Y} = i+2 | \hat{X} =i)
        & =
        \int_{x=0}^1 \left[ 
        Q\left( \frac{2-x}{\sqrt{2} \sigma}         \right)
        -
        Q\left( \frac{3-x}{\sqrt{2} \sigma} \right)
        \right] dx 
        \\&
        = \cO(e^{-\frac{\gamma}{4}}) + \cO(e^{-\gamma}). 
    \end{align*}
    Equation~\eqref{eq:X-Y-2-jump} follows immediately.
\end{IEEEproof}

~\\
\noindent
Combining the results of Lemmas~\ref{lem_U-X-Y-1-jump}-\ref{lem_X-Y-2-jump}, we reach the following proposition on transition probabilities in the algebraic (hard-output) TE-QKD channel. We use the frame symmetry where Alice and Bob can be switched. Also, we assume that signal-to-noise ratio is large enough to neglect border effects and use the assumptions established in Appendix~\ref{app_fundamental}. 

\begin{proposition} \label{prop:single-double-jumps}
Consider the uncoded TE-QKD channel with $N$ bins per frame, 
SNR $\gamma = \frac{1}{\sigma^2} \gg 1$. 
The transition probabilities decompose into two regimes:
\begin{itemize}
\item Single-bin jump.
\begin{align}
\quad \forall i, j \in \Z_N,~ |j-i| =1,~ \text{then}~ p_{ij} = \pr(\hY=j|\hX=i)=\frac{\sigma}{ \sqrt{\pi}} + \cO\left( e^{-\gamma/4} \right) = \Theta  \left( \gamma^{-1/2} \right).
\end{align}
Therefore, the uncoded symbol error probability scales as $\Theta\left( \gamma^{-1/2} \right)$, achieving diversity order $\frac{1}{2}$.
\item Multi-bin jump.
\begin{align}
\forall i, j \in \Z_N,~ |j-i| \geq 2,~ \text{then}~ p_{ij} =  \pr(\hY=j|\hX=i) = \cO\left( e^{-\gamma/4} \right).
\end{align}
The errors decay exponentially in $\gamma$ and have infinite diversity order.
\end{itemize}
\end{proposition}

%%\begin{figure}[htb]
%%\centering 
%%\includegraphics[angle=270, width=0.8\linewidth]{density_of_U_N=8_vs_sigma.pdf} 
%%\caption{Probability density of $\pi_i$ when $N=8$}
%%     \label{fig:pdf}
%%\end{figure}

Now, we are ready to expose in the next section the conditions under which a finite-length code on the half-diversity TE-QKD channel is going to perform in a non-standard (rather shocking) manner by exhibiting an infinite diversity. 

%%--------------------------------------------------------------
%%--------------------------------------------------------------
%%\clearpage
\section{From Finite to Infinite Diversity \label{sec_infinity-div}}
For a coded TE-QKD channel, the relevant question is whether all dominant single-bin jumps can be absorbed by the error-correcting code. If this is possible, then the remaining decoding errors must involve at least one multi-bin jump, whose probability is exponentially small. The decoding error probability then has infinite diversity. If this is not possible, then there exists a decoding error made of single-bin jumps. Since each such jump has probability $\Theta(\gamma^{-1/2})$, the overall error probability has a polynomial lower bound, and the diversity remains finite.\\
%%\nrd{Sentence to be improved: The impatient reader is advised to check the infinite diversity conditions in Theorems 1 \&2,
%%then read Section VI that shows practical examples before going through the details of Section V.}
The reader may find it helpful to first review the infinite-diversity conditions presented in Theorems~\ref{thm:infinite-diversity-hard} and~\ref{thm:infinite-diversity-soft}, and then the practical examples in Section~\ref{sec_code_examples}, before studying the detailed derivations in Section~\ref{sec_infinity-div}.
\subsection{Infinite Diversity under Algebraic Decoding \label{sec_alg_dec}}
For reconciliation between Alice and Bob, we assume in this section that Bob has a bounded-distance algebraic (hard-decision input) decoder with an error-correction radius of $t$, where $t=\lfloor\tfrac{d_{Hmin}-1}{2}\rfloor$, in contrast to a complete algebraic decoder or a decoder that corrects beyond $t$ \cite{Guruswami1999}. A $t$-bounded algebraic decoder for a linear code is a deterministic, polynomial-time algorithm that maps any received word to a unique codeword if one exists within a specified Hamming distance $t$, or outputs a decoding failure otherwise \cite{Blahut2003}.

With a Gray-labeled TE-QKD frame of $N = 2^m$ bins, each photon carries $m = \log_2 N$ coded bits. Hence, for a binary code of length $n$ bits,
one codeword is transmitted using $L = \lceil\frac{n}{m}\rceil$ photons. A single-bin jump changes only one bit in the label of the affected photon. Therefore, if every photon in the codeword undergoes at most one single-bin jump, then at most $L$ coded bits are corrupted. 
The algebraic decoder can fix all such dominant error patterns only when its error-correction radius $t$ is at least $L$. This gives the threshold condition in the following theorem. For non-trivial codes, $0 < k < n$, $t$ is null for $n\le 2$, and $t\le (n-k)/2 < n/2$ from the Singleton bound, i.e., there exists no linear code such that $n/2 \le t$. Hence, we consider $n\ge 3$ and $m\ge 1$, knowing that the inequality of the next theorem can never be satisfied for $m=1$ and $m=2$ bits per photon. Furthermore, 
for the sake of practicality, only well-aligned systems are considered where $n$ is multiple of $m$, or $m$ is multiple of $n$.
%%-------------------------------------------
%%-------------- Theorem 1 ------------------
%%-------------------------------------------
\clearpage
\begin{theorem}
\label{thm:infinite-diversity-hard}
  Let $\cC$ be a linear binary code with parameters~$\left[n, k, d_{Hmin}, t \ge 1\right]_2$.  Assume that $\cC$ is used for reconciliation over a TE-QKD frame of length $N=2^m$ bins, $n\ge 3$, $m\ge 1$, and $\gcd(n,m)=\min(n,m)$. The TE-QKD bins are labeled via a binary Gray code of $m$ bits per label. Then, a $t$-bounded algebraic decoder for $\cC$ achieves infinite diversity if and only if 
    \begin{align}
        \label{equ_n_m_t}
        L=\frac{n}{m} \leq t.
    \end{align}
\end{theorem}
\begin{IEEEproof}
We know from Proposition~\ref{prop:single-double-jumps} that a single-bin jump has a probability $\Theta\left(\gamma^{-1/2}\right)$ 
and a multi-bin jump has a probability 
$\cO\left(e^{-\gamma/4}\right)$.\\
$\bullet$ For the special case $n/m < 1 \le t$, there exists a positive integer $\ell \ge 2$ such that $m=n\ell$ since $\gcd(m,n)=n$. In this case, one photon carries $\ell$ codewords. A single-bin jump flips one bit 
in one of those $\ell$ codewords. Such a single error is solved by the algebraic decoder having $t\ge 1$. Thus, all error events with
probability $\Theta\left(\gamma^{-1/2}\right)$ are successfully decoded. The decoder fails only when multiple jumps occur, with non-neighboring labels flipping more than a single bit.
Consequently, the probability of decoding error is $\cO\left(e^{-\gamma/4}\right)$ (infinite diversity) for $n/m < 1 \le t$.\\
$\bullet$ Now, we consider $L=n/m \ge 1$, where $L$ is integer since 
$\gcd(m,n)=m$, $L$ being the number of photons per codeword.  
If $L\le t$, all single-bin jumps create a maximum of $L$ erroneous
bits thanks to the Gray labeling. The decoder fails only
for multi-bin jumps with probability $\cO\left(e^{-\gamma/4}\right)$.
We conclude that diversity is infinite if $L=n/m\le t$. The sufficient condition is proved.\\
$\bullet$ For the necessary condition, consider now $L=n/m >t$, i.e.,
the error-correction radius $t$ is less than the number of photons per codeword. An error event of $t+1$ single-bin jumps has a probability $\Theta\left(\gamma^{-(t+1)/2}\right)$ of finite diversity equal to $\tfrac{1}{2} \times (t+1)$ and makes the $t$-bounded decoder fail.
\end{IEEEproof}
%%Remark: Some $t$-bounded algebraic decoders are capable of detecting
%%the presence of $t+1$ errors if the noisy word does not fall inside
%%the radius-$t$ ball of another codeword. In this case, the decoder
%%stops the decoding procedure. Nevertheless, the $t+1$ errors
%%introduced by the channel noise remain. If the noisy word with $t+1$ errors falls inside the ball of another codeword, then the decoder confirms its mistake by decoding the other codeword. Finally,
%%in the special case of an improved algebraic decoder (complete or not) that decodes some error patterns beyond a weight of $t$, 
~\\
\noindent
We saw in the proof of Theorem~\ref{thm:infinite-diversity-hard} that diversity is $\tfrac{1}{2} \times (t+1)$ when the code and its $t$-bounded algebraic decoder fail to attain infinite diversity. This is in perfect agreement with $d_0 \times (t+1)$ as in (\ref{equ_pe_bounds}) for standard ergodic Rayleigh fading channels.

%% The following sentence is wrong, 
%% checked by Joe & Siyao on 3 July 2025
%%However, for long codes, the MFD$(d_{Hmin}, t+1)$ property is satisfied with high probability making the necessary condition valid for all algebraic decoders. In fact, any MFD$(\omega, \ell\ge t+1)$ property will force the diversity to be finite at an order of $\tfrac{\ell}{2} \ge \tfrac{1}{2} \times (t+1)$.\\
Next, we generalize Theorem~\ref{thm:infinite-diversity-hard} to non-binary linear codes over a finite field $\F_q$, $q \ge 2$, $q$ power of $2$.
Similar to well-aligned linear codes, for the sake of practicality, we assume that we have $m$ multiple of $\log_2(q)$,
or $\log_2(q)$ multiple of $m$. We did not introduce the binary image of the non-binary code, avoiding to apply Theorem~\ref{thm:infinite-diversity-hard} to the binary image, because most available algebraic decoders work directly on the non-binary alphabet of the code.

%%---------------------------------------------
%%-------------- Corollary 1 ------------------
%%---------------------------------------------
\begin{corollary}
\label{cor_infinite-diversity-hard-nonbinary}
  Let $\cC$ be a linear code defined over $\F_q$ with parameters~$\left[n, k, d_{Hmin}, t \ge 1\right]_q$, where $\F_q$ is a finite field of characteristic $2$.  Assume that $\cC$ is used for reconciliation over a TE-QKD Gray-labeled frame of length $N=2^m$ bins, $n\ge 3$, $m\ge 1$, $\gcd(m,\log_2(q))=\min(m,\log_2(q))$, and $L=\tfrac{n\log_2(q)}{m}$ integer. The Gray code is binary with $m$ bits per label. Then, a $t$-bounded algebraic decoder for $\cC$ achieves infinite diversity, if and only if 
    \begin{align}
    \label{equ_n_logq_m_t}
        L=\frac{n\log_2(q)}{m} \leq t.
    \end{align}
\end{corollary}
\begin{IEEEproof} The proof is similar to that of Theorem~\ref{thm:infinite-diversity-hard}. We distinguish the two following cases:\\
$\bullet$ $\log_2(q)=\ell m$, where $\ell$ is integer, $\ell \ge 1$.
The number of photons per codeword is $L=\tfrac{n\log_2(q)}{m}=n\ell$. In this case, we always have $L>t$ with finite 
diversity of order $\tfrac{1}{2} \times (t+1)$.\\
$\bullet$ $m=\ell\log_2(q)$, where $\ell$ is integer, $\ell \ge 2$.
The number of photons per codeword is $L=\tfrac{n\log_2(q)}{m}=\tfrac{n}{\ell}$. Infinite diversity is possible
only for $\ell \ge 3$ if the code satisfies $L\le t$.
Otherwise, diversity is $\tfrac{1}{2} \times (t+1)$. 
\end{IEEEproof}

The practicality conditions $\gcd(m,n)=\min(m,n)$ and $\gcd(m,\log_2(q))=\min(m,\log_2(q))$ make the binary codewords in Theorem~\ref{thm:infinite-diversity-hard} and the binary representation of a field element in Corollary~\ref{cor_infinite-diversity-hard-nonbinary} aligned with the photon labels. Without such an alignment, the partial Hamming weight enumeration of a code over a piece of the photon label would be necessary to analyze the diversity order. The partial Hamming weight enumeration depends on the specific code structure and cannot lead to a general condition such as (\ref{equ_n_m_t}) or (\ref{equ_n_logq_m_t}) for infinite diversity. 

The necessary condition in Theorem~\ref{thm:infinite-diversity-hard} 
and Corollary~\ref{cor_infinite-diversity-hard-nonbinary} is based on the fact that the bounded-distance decoder fails if it encounters $t+1$ errors or more. Is Theorem~\ref{thm:infinite-diversity-hard} 
still valid under improved decoders (complete or not)
capable of solving some error patterns beyond a Hamming weight of $t$?
Assume there exists an error pattern of weight $\nu \ge t+1$ 
formed by $\nu$ single-bin jumps that makes the noisy word
fall inside another codeword $t$-radius Hamming ball, or
makes the noisy word fall outside all $t$-radius balls but cannot be decoded,
then the decoder fails in finding Alice's word with probability
$(\Theta(\gamma^{-\tfrac{1}{2}}))^{\nu}$ of finite diversity equal
to $\nu/2$. Hence, the existence of such a weight-$\nu$ error pattern keeps the condition of Theorem~\ref{thm:infinite-diversity-hard}
applicable for improved non-$t$-bounded decoders. 

%% WRONG
%%One would think that better decoders, complete or correcting beyond $t$, could improve the necessary condition. However, the code always includes two codewords $c_A$ and $c'_A$ separated by a Hamming distance of $d_{Hmin}=2t+1$, by definition. Therefore, jumping in the direction $c_A \rightarrow c'_A$ by $t+1$ single-bin jumps will unfortunately
%%generate an error of diversity $\tfrac{1}{2} \times (t+1)$.
%%If $d_{Hmin}=2t+2$, we randomly choose $c_A$ or $c'_A$ when at equal Hamming distance $t+1$ from both of them, and hence diversity is still stuck at $\tfrac{1}{2} \times (t+1)$.\\ 

\noindent
Forcing infinite diversity comes naturally at a price in the coding
rate $R_c=\tfrac{k}{n}$. Indeed, (\ref{equ_n_m_t}) or (\ref{equ_n_logq_m_t}) 
are asking the error-correction radius $t$ to be greater than a fraction of the code length $n$. This increase in $t$ leads to a decrease in the coding rate $R_c$ as stated below.
%%---------------------------------------------
%%-------------- Corollary 2 ------------------
%%---------------------------------------------
\begin{corollary}
\label{cor_rate_loss_hard}
Consider a linear code $\cC[n,k,d_{Hmin},t]_q$ of rate $R_c=k/n$ defined over a finite field $\F_q$ of characteristic 2. Assume that $\tfrac{n\log_2(q)}{m} \le t$, i.e., $\cC$ achieves infinite diversity in a TE-QKD reconciliation as per Theorem~\ref{thm:infinite-diversity-hard} or Corollary~\ref{cor_infinite-diversity-hard-nonbinary}.
Then, we have 
\begin{equation}
\label{equ_Rc_bound}
R_c \le 1-\frac{2\log_2(q)}{m}.
\end{equation}
\label{cor_Rc_loss_hard}
\end{corollary}
\begin{IEEEproof}
The Singleton bound states that $2t+1 \le d_{Hmin} \le n-k+1$ for any linear code.
We get the inequality $2t \le n-k=n \cdot (1-R_c)$.
Then, $\tfrac{n\log_2(q)}{m} \le t \le n \cdot (1-R_c)/2$.
After simplifying $n$ on both sides, the inequality becomes
$\tfrac{2\log_2(q)}{m} \le (1-R_c)$ which leads to the announced result.
\end{IEEEproof}

Improvements of the bound in (\ref{equ_Rc_bound}) could be obtained by applying the Hamming bound or the Sphere Packing bound in $\F_q^n$ given by $\sum_{i=0}^t {n \choose i} (q-1)^i \le q^{n-k}$ or the finite-length Gilbert-Varshamov bound $\sum_{i=0}^{d-2} {n-1 \choose i} (q-1)^i \le q^{n-k}$ \cite{MacWilliams1977}, but these enhancements have no closed form expressions.\\ 

The very short binary repetition $[3,1,3]_2$ code satisfies 
the infinite-diversity condition of Theorem~\ref{thm:infinite-diversity-hard}, under algebraic decoding for $m=3$ bits per photon, $n=3$, and $t=1$. Its error-rate performance was shown in Figure~\ref{fig_3_1_3}. Under normal conditions (no rough approximations, no dramatic sub-optimality), soft-decision decoding outperforms algebraic decoding for the same error-correcting code. Hence, the $[3,1,3]_2$ code also attains infinite diversity under soft-decision decoding. Figure~\ref{fig_3_1_3}
shows a probability of error per bit of $10^{-6}$ at signal-to-noise ratio $\gamma=15$ dB. The gap between both decoders depends on the code itself and could be estimated using performance analysis tools.
In the case of the very short $[3,1,3]_2$ code, the gap is small suggesting that algebraic decoding on a TE-QKD frame of $8$ bins is sufficient when the infinite diversity condition $\tfrac{n\log_2(q)}{m} \le t$ is valid. Other examples are given in the sequel in Section~\ref{sec_code_examples}. We will see in the next section that the infinite diversity condition is relaxed under soft-decision decoding, i.e.,
it is easier to get infinite diversity with soft reconciliation rather than hard/algebraic reconciliation.

%%\nrd{Note 1: we never mentioned the ergodic property when referring
%%to standard wireless channels. We will add it.\\
%%Note 2: Check again where do we need $N \gg 1$.
%%I believe $\gamma \gg 1$ is sufficient.\\
%%Note 3: After Corollary 1, add the rate loss occurred by the condition
%%of infinite diversity. See Singleton bound, Hamming bound, and GV bound. Put the note on RS, LDPC, and binary BCH codes.\\
%%Note 4: Rice even at K>0 is still stuck at diversity=1 like Rayleigh.\\
%%Note 5: $q=8=2^3$, m=6 bits per photon, n=8, Code is $RS[n=8,k=4,d_H=5,t=2]_8$\\
%%$n\cdot log_2(q)/m=4$ not less than or equal to $t=2$.\\
%%$n\cdot log_2(q)/m=4$ photons per RS codeword.\\
%%Algebraic decoding will give a finite diversity.\\
%%However, $n\cdot log_2(q)/m=4=2t<d_H$, then soft-decision decoding will give infinite diversity.\\
%%}

%%----------------------------------------------------------
%%----------------------------------------------------------
\subsection{The MFD Property of a Code\label{sec_mfd}}
Before analyzing soft reconciliation and its diversity in the next section, we introduce a special definition for codes with respect to the diversity achieved by a codeword. A permutation of the symbol positions in a code yields another code version but could influence its diversity in presence of fading. The effect of permutations, known as multiplexing, was studied on non-ergodic block fading channels \cite{Boutros2004}\cite{Boutros2005}. The TE-QKD channel is ergodic by definition, however the code symbol positions could affect the diversity because the $n\log_2(q)$ binary digits of a word are spread over $L$ photons by groups of $m=\log_2(N)$ bits. Definitions~\ref{def_MFD_code} and \ref{def_MFD_code_2} of the MFD property will be crucial to determine whether diversity is finite or infinite as stated by Theorem~\ref{thm:infinite-diversity-soft} under soft-decision decoding.

%% Next sentence is wrong.
%Definition~\ref{def_MFD_code} will also be useful in commenting complete decoding versus $t$-bounded decoding in Theorem~\ref{thm:infinite-diversity-hard} under algebraic decoding.

\begin{definition}[Neighboring Labels in a Gray Code]
\label{def_neighboring_label}
Consider a TE-QKD frame of length $N=2^m$ bins where each bin has an $m$-bit binary label built from a one-dimensional Gray code. 
A label is called a zero-neighboring label, or a neighbor of zero, 
if it corresponds to a single-bin jump away from zero.
Thanks to the Gray code, the neighbor of zero has a unit Hamming weight.
The current definition is easily extended to a neighboring label of a non-zero label.
\end{definition}

Consequently, whatever is the considered Gray code version, the number of neighboring labels cannot exceed $2$. In Table~\ref{tab:graycode},
the Gray code has a single neighboring label of zero (indicated in blue color),
while the centered version in Table~\ref{tab:graycode2} has two neighboring labels of zero. It is highly important to note, according to Proposition~\ref{prop:single-double-jumps} and Definition~\ref{def_neighboring_label}, a jump to a neighboring bin (with a neighboring label) has probability $\Theta(\gamma^{-1/2})$ of diversity $1/2$, while other multi-bin jumps have probability $\cO(e^{-\gamma/4})$ of infinite diversity even if one bit only is flipped. E.g., non-neighboring labels $010$ and $001$ associated to bins $i=3$ and $i=7$ in Table~\ref{tab:graycode} correspond to $\cO(e^{-\gamma/4})$ jumps from $000$ of bin $i=0$. In Table~\ref{tab:graycode2}, label $010$ associated to bin $i=6$ is also $\cO(e^{-\gamma/4})$ when a triple jump occurs from the zero-label bin at $i=3$.

\begin{definition}[Maximal Finite Diversity, MFD$(\omega,\ell)$]
\label{def_MFD_code}
Let $\cC$ be a linear $[n,k,d_{Hmin}]_q$ code defined over a finite field 
$\F_q$ of characteristic $2$. 
Let $L=\tfrac{n\log_2(q)}{m}$ be the number of photons per word,
assuming that a TE-QKD frame has $N=2^m$ bins and that $L\ge 1$ is integer, i.e., $\F_q^n \sim \F_2^{n\log_2(q)}=\F_2^{mL}$.
The code $\cC$ is called MFD$(\omega,\ell)$ if there exists a non-zero codeword $c \in \cC$ 
such that the length-$mL$ binary image of $c$ has binary Hamming weight 
$\omega_H(c \in \F_2^{mL})=\omega$ and the word $c$ admits $\ell$ neighboring labels of the all-zero label,
where $d_{Hmin} \le \omega \le mL$ and $0 \le \ell \le \min(\omega,L)$.
\end{definition}

Let $\cC_b$ denote the binary image of $\cC$ over $\F_2^{mL}$. From Definition~\ref{def_MFD_code}, we deduce that the MFD$(\omega, \ell)$ property applies to both $\cC$ and $\cC_b$ at the same time.

As examples of the MFD property, in the case of the binary $[6,3,3]$ code listed on the first row of the Slepian table in Figure~\ref{fig_slepian_6_3_3}, for $m=2$ bits per photon with the Gray code  $11~10~00~01$, the second codeword $c=00~11~01$ makes the code MFD$(3,1)$, the third codeword $c=01~01~11$ makes the code MFD$(4,2)$, and the fourth codeword $c=01~10~10$ makes the code MFD$(3,3)$. 
With the same code and the Gray labeling of Table~\ref{tab:graycode2} for $m=3$,
the second codeword $c=001~101$ makes the code MFD$(3,1)$, the third codeword $c=010~111$ makes the code MFD$(4,0)$, while the code is MFD$(3,0)$ thanks
to the fourth codeword $c=011~010$. Table~\ref{tab_MFD_6_3_3} shows the complete MFD distribution of the $[6,3,3]_2$ code.

\begin{definition}[Full MFD versus MFD Deficient]
\label{def_MFD_code_2}
Let $\cC$ be a linear $[n,k,d_{Hmin}]_q$ code defined over a finite field 
$\F_q$ of characteristic $2$.
Assume that there exists a non-zero codeword $c \in \cC$ whose binary image has  Hamming weight $\omega$ such that $\cC$ is MFD$(\omega,\omega)$, then we say that $\cC$ is full MFD for the Hamming weight~$\omega$.
If $\cC$ is not MFD$(\omega,\omega)$, for all $\omega \ge d_{Hmin}$, 
we say that $\cC$ is MFD deficient.
\end{definition}

From these examples, the reader
should conclude that the most critical property is MFD$(\omega,\omega)$ when dealing with finite diversity. 
Furthermore, as it is obvious from the above definition and the examples, the MFD$(\omega, \ell)$ property of a code depends on $m=\log_2(N)$; it should be denoted MFD$(\omega, \ell, m)$ but we dropped the $m$ for the sake of simplifying the notation.

\begin{table}[!h]
\begin{center}
\caption{Full MFD distribution of the shortened $[6,3,3]_2$ Hamming code.\label{tab_MFD_6_3_3}}
\includegraphics[width=0.55\columnwidth]{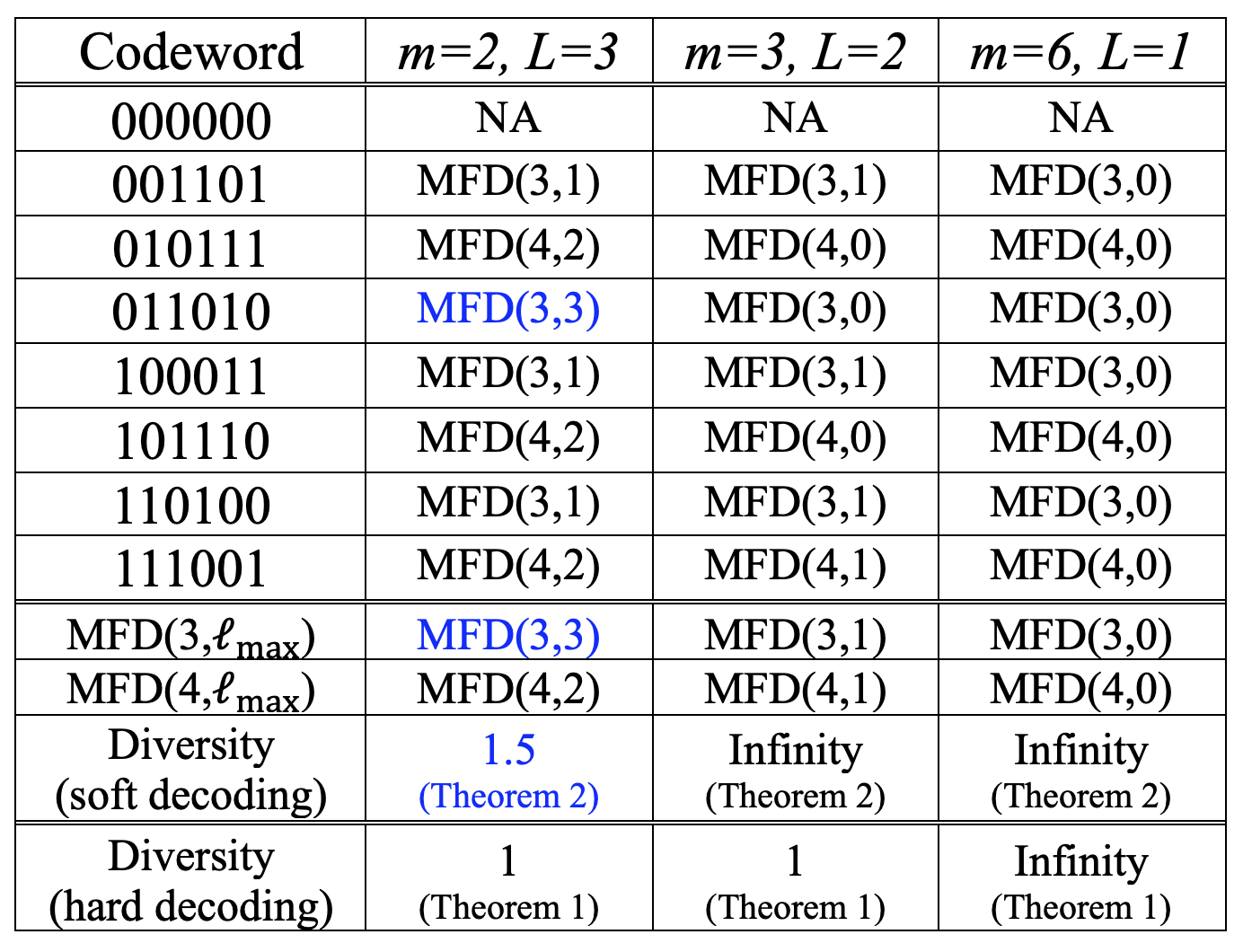}
\end{center}
\end{table}

%%----------------------------------------------------------
%%----------------------------------------------------------
\subsection{Infinite Diversity under Soft-Decision Decoding \label{sec_soft_dec}}
We start this section by a simple example, before generalizing the analysis to any binary or non-binary linear code. Such an approach helps the reader in understanding the abstraction made in Lemma~\ref{lem_necessary_soft} and Theorem~\ref{thm:infinite-diversity-soft}. The APP soft-decision decoder based on (\ref{equ_soft_recon}) is complete, while the algebraic decoder of Section~\ref{sec_alg_dec}
is bounded using (\ref{equ_complete_dec}) with the additional constraint $d_H(\upsilon,c)\le t$. This difference makes the analysis of soft reconciliation much more challenging. 

The $[n=6,k=3,d_{Hmin}=3,t=1]_2$ code of Figure~\ref{fig_6_3_3} and Table~\ref{tab_MFD_6_3_3} does not satisfy the condition of Theorem~\ref{thm:infinite-diversity-hard} for infinite diversity under algebraic decoding, when the TE-QKD frame has $N=8$ bins. 
In such a case, two photons are needed to carry a codeword
while the error-correction radius is $t=1$. Algebraic reconciliation attains a finite diversity of $\tfrac{1}{2}(t+1)=1$ as illustrated in Figure~\ref{fig_6_3_3}. A surprising result is also shown in this same figure: Under soft-decision decoding, the reconciliation probability of error exhibits an abnormally large diversity; the effective diversity measured between $10^{-6}$ and $10^{-7}$ is equal to $9$. Under normal conditions, as in wireless communication channels, diversity for soft decoding should have been $\tfrac{1}{2}d_{Hmin}=1.5$.

\begin{figure}[!h]
\centering 
\vspace{-8mm}
\includegraphics[angle=270, width=0.7\linewidth]{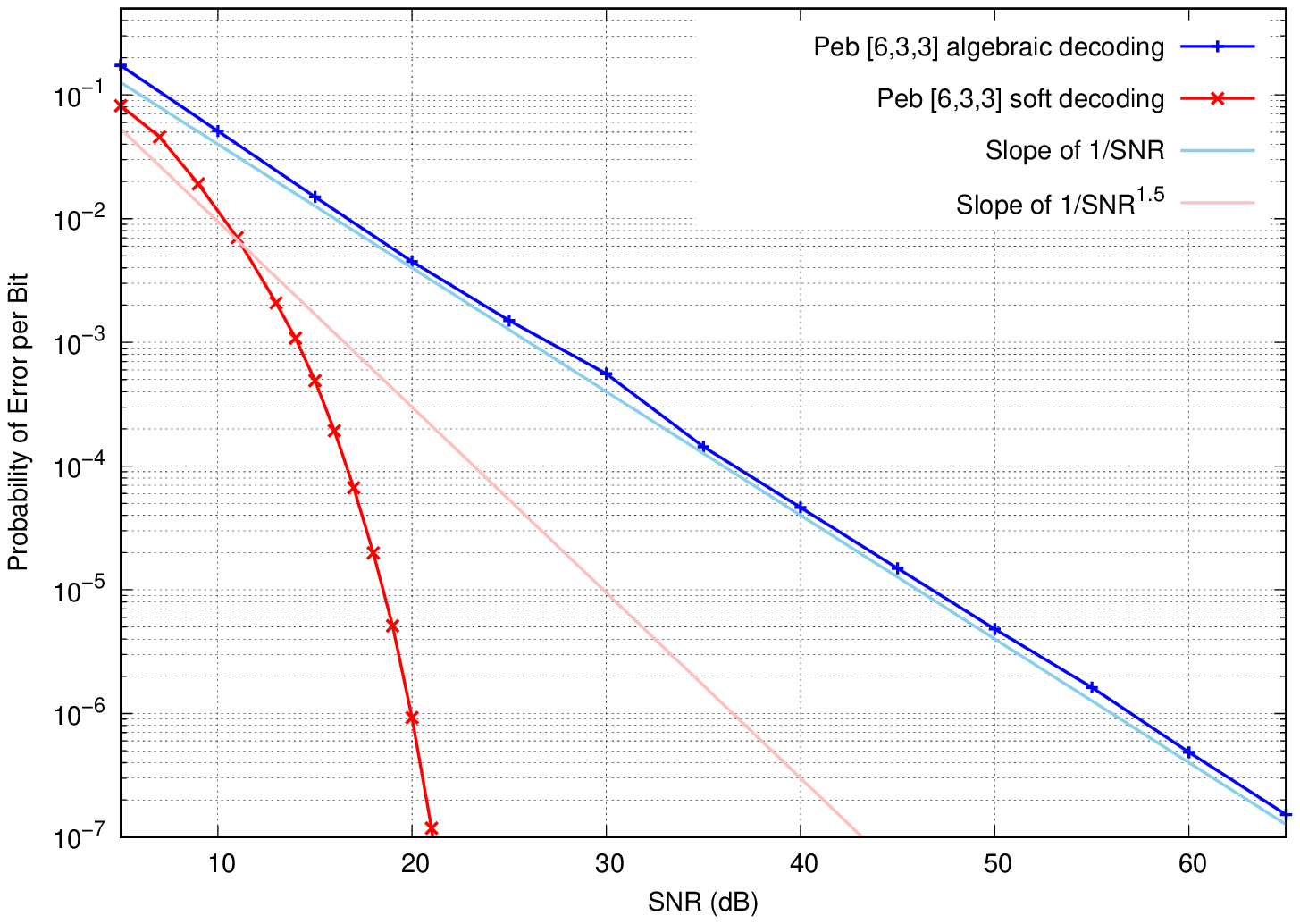}
\vspace{9mm}
\caption{Hard-decision (algebraic) versus soft-decision of a $[6,3,3]_2$ shortened binary Hamming code, $N=8$ bins per frame.
\label{fig_6_3_3}}
\end{figure}

%%\nrd{I moved the MFD definition before Section V-A on 13 June 2026.\\
%%I reviewed the whole section of Theorem 1 and Corollaries 1 and 2.\\
%%15 June 2026: Start this section with the analysis of the $[6,3,3]_2$ code before considering the Lemma of the necessary condition in the general case, and finally Theorem 2.\\}

The fantastic performance of the $[6,3,3,t=1]_2$ in soft reconciliation is due to its infinite diversity, as we will prove in the sequel. It also implies that two single-bin jumps corresponding to photon positions close to the bin border are corrected by the decoder, otherwise diversity becomes $\tfrac{1}{2}d_{Hmin}$. In other words, in presence of soft information, the $[6,3,3,t=1]_2$ code is correcting critical double errors, those that would create a finite diversity. It is well known that a soft-decision decoder can correct multiple errors (not all) up to $d_{Hmin}-1$ erroneous binary digits. A simple way to understand this capability is to assume that $d_{Hmin}-1$ erasures are declared on the $d_{Hmin}-1$ less reliable positions then filling these erasures via an algebraic decoder. Hence, many errors of weight up to $d_{Hmin}-1$ are corrected, but not all of them.\\

Now, we consider one critical example of two single-bin jumps. The number of bins is set to $N=8$ with the Gray labeling of Table~\ref{tab:graycode} (this choice has no effect on the analysis result). 
The $[6,3,3]_2$ code $\cC$ is employed for soft reconciliation, where $L=n/m=6/3=2$ photons per word.
Without loss of generality, we assume that Alice's syndrome is $s_A=000$ ($e_A=000~000$) and Alice's word is $c_A=000~000$. Bob's word before reconciliation is $c_B=100~100$ with syndrome $s_B=111$, this is the seventh word on the third row of the Slepian table in Figure~\ref{fig_slepian_6_3_3}. Given $c_A$ and $c_B$, we have $0 \le X_1 < 1$ and $0 \le X_2 < 1$ for the photons positions for Alice, and $1 \le Y_1 < 2$ and $1 \le Y_2 < 2$ for the photons positions for Bob. The nearest word to $c_B$ in the coset $e_A+\cC$ 
is $c'_A=110~100$ at Hamming distance $1$ located in the seventh column of the first row in the Slepian table. Given the soft channel output $(Y_1=y_1, Y_2=y_2)$,
the {\em a posteriori} probability of $c_A$ is denoted by $APP_0$ (first codeword on the first row) and the {\em a posteriori} probability of $c'_A$ is denoted by $APP_6$ (seventh codeword on the first row). Using (\ref{equ_app1_approx}) and (\ref{equ_app2_approx}), we obtain
\begin{align}
APP_0 &= APP(000|y_1) \cdot APP(000|y_2) \nonumber \\ 
&= APP(i=0,j=1|y_1) \cdot APP(i=0,j=1|y_2) \nonumber \\ 
& \propto \left[ \tfrac{1}{2} e^{-\tfrac{(y_1-1)^2}{2\sigma^2}} - \tfrac{1}{2} e^{-\tfrac{y_1^2}{2\sigma^2}}\right] 
\cdot \left[ \tfrac{1}{2} e^{-\tfrac{(y_2-1)^2}{2\sigma^2}} - \tfrac{1}{2} e^{-\tfrac{y_2^2}{2\sigma^2}}\right], \label{equ_app0}
\end{align}
and
\begin{align}
APP_6 &=APP(110|y_1) \cdot APP(100|y_2) \nonumber \\
&=APP(i=2,j=1|y_1) \cdot APP(i=1,j=1|y_2) \nonumber \\ 
& \propto \left[ \tfrac{1}{2} e^{-\tfrac{(y_1-2)^2}{2\sigma^2}} - \tfrac{1}{2} e^{-\tfrac{(y_1-3)^2}{2\sigma^2}}\right] 
\cdot \left[ 1-\tfrac{1}{2} e^{-\tfrac{(y_2-1)^2}{2\sigma^2}} - \tfrac{1}{2} e^{-\tfrac{(y_2-2)^2}{2\sigma^2}}\right]. \label{equ_app6}
\end{align}

\begin{figure}[!h]
\centering 
\includegraphics[width=0.6\linewidth]{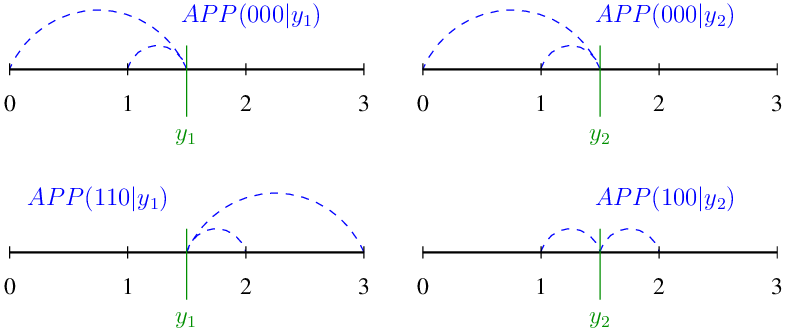} 
\caption{Euclidean distances involved in the APP evaluation in (\ref{equ_app0}) and (\ref{equ_app6}). The APP per photon includes the squared Euclidean distances to the two borders of a bin. Binary labels are defined in Table~\ref{tab:graycode}.
\label{fig_app_distances}}
\end{figure}

For $\sigma^2 \ll 1$, the soft reconciliation of $c_B$ will yield one of the two competing words, $c_A$ or $c'_A$, while all other words have $APP_i$, $i\ne 0$ and $i\ne 6$, negligible with respect to $APP_0$ and $APP_6$.
Note that words $110~110$ and $100~110$ do not belong to Alice's coset
and hence do not compete with $c_A$ or $c'_A$ in the soft reconciliation process. Figure~\ref{fig_app_distances} illustrates the Euclidean distances to the bin borders appearing in expressions (\ref{equ_app0}) and (\ref{equ_app6}). The smallest distance shall define the dominating term for $\sigma^2 \ll 1$, i.e., at high signal-to-noise ratio.

For the purpose of determining the diversity order in this competition between $c_A$ and $c'_A$, let us find the decision boundary in the $(y_1, y_2)$ bi-dimensional plane, for $(y_1, y_2) \in [1,2[^2$. We write $APP_0=APP_6$. 
In the equality $(\ref{equ_app0})=(\ref{equ_app6})$, we neglect larger distances in the three {\em a posteriori} probabilities per photon, $APP(000|y_1), APP(000|y_2), APP(110|y_1)$, as illustrated in Figure~\ref{fig_app_distances}, i.e., $e^{-\tfrac{d_2^2}{2\sigma^2}} \ll e^{-\tfrac{d_1^2}{2\sigma^2}}$, if $d_2 > d_1$, for $\sigma^2 \ll 1$.
We get the following equation for the decision boundary
\begin{align*}
\left[ \tfrac{1}{2} e^{-\tfrac{(y_1-1)^2}{2\sigma^2}} + \cO(e^{-\gamma/2}) \right] \cdot \left[ \tfrac{1}{2} e^{-\tfrac{(y_2-1)^2}{2\sigma^2}} + \cO(e^{-\gamma/2}) \right] \\
= \left[ \tfrac{1}{2} e^{-\tfrac{(y_1-2)^2}{2\sigma^2}} + \cO(e^{-\gamma/2}) \right] \cdot \left[ 1 - \tfrac{1}{2} e^{-\tfrac{(y_2-1)^2}{2\sigma^2}} - \tfrac{1}{2} e^{-\tfrac{(y_2-2)^2}{2\sigma^2}} \right],
\end{align*}
which yields
\begin{equation} 
y_1=-\sigma^2 \cdot \log\left[ 2 e^{\tfrac{(y_2-1)^2}{2\sigma^2}} - e^{\tfrac{2(y_2-3/2)}{2\sigma^2}}\right] + \frac{3}{2} + \cO(e^{-\gamma/2}), ~~~y_2 \in [1,2[.
\end{equation}
At vanishing noise variance, the decision boundary equation simplifies to
\begin{equation} 
y_1= - \frac{1}{2} (y_2-1)^2+\frac{3}{2}, ~~~y_2 \in [1,2[,
\end{equation}
or equivalently
\begin{equation} 
\label{equ_decision_boundary}
y_2=\sqrt{-2y_1+3}+1, ~~~y_1 \in [1, 3/2].
\end{equation}

The left plot in Figure~\ref{fig_decision_boundary} shows the decision boundary (\ref{equ_decision_boundary}) in the bi-dimensional plane within the range~$[1,2[^2$. The right plot emphasizes a square of side $\Delta$ such that the bottom left corner is $(1,1)$ and the top right corner is found by the intersection of the boundary with the first bisector line $y_1=y_2$. In this special case, for the $[6,3,3]_2$ code, we have
$\Delta=\sqrt{2}-1$ since the solution of (\ref{equ_decision_boundary}) for $y_1=y_2$ is $\sqrt{2}$. This special square area will be referred to as the demilitarized zone (DMZ) between the two competing words $c_A$ and $c'_A$, knowing that an error occurs if $c'_A$ wins. Hence, the integral outside the DMZ over the distribution of $(Y_1, Y_2) \in [1,2[^2$ yields an upper bound of the probability of error. The next lemma determines how the probability of being outside the DMZ, beyond $\Delta$, is vanishing at high signal-to-noise ratio. Only the first statement in Lemma~\ref{lem_DMZ} is used in the proof of Theorem~\ref{thm:infinite-diversity-soft}. The other three statements could be used to validate the choice of $\Delta$ and to undergo the performance analysis for a particular code. 

%%\nrd{We put here the Lemma on the DMZ, $e^{-\Delta^2\gamma/4}$.\\
%%The analysis above that is limited to the $[6,3,3]$ code will now be genetalized to any code in the next lemma (Necessary Condition) and the next theorem on the infinite diversity condition under soft reconciliation.}
%%~\\

\begin{figure}[!h]
\vspace{-5mm}
\centerline{
\hspace{-0.5cm}\includegraphics[angle=270,width=0.65\linewidth]{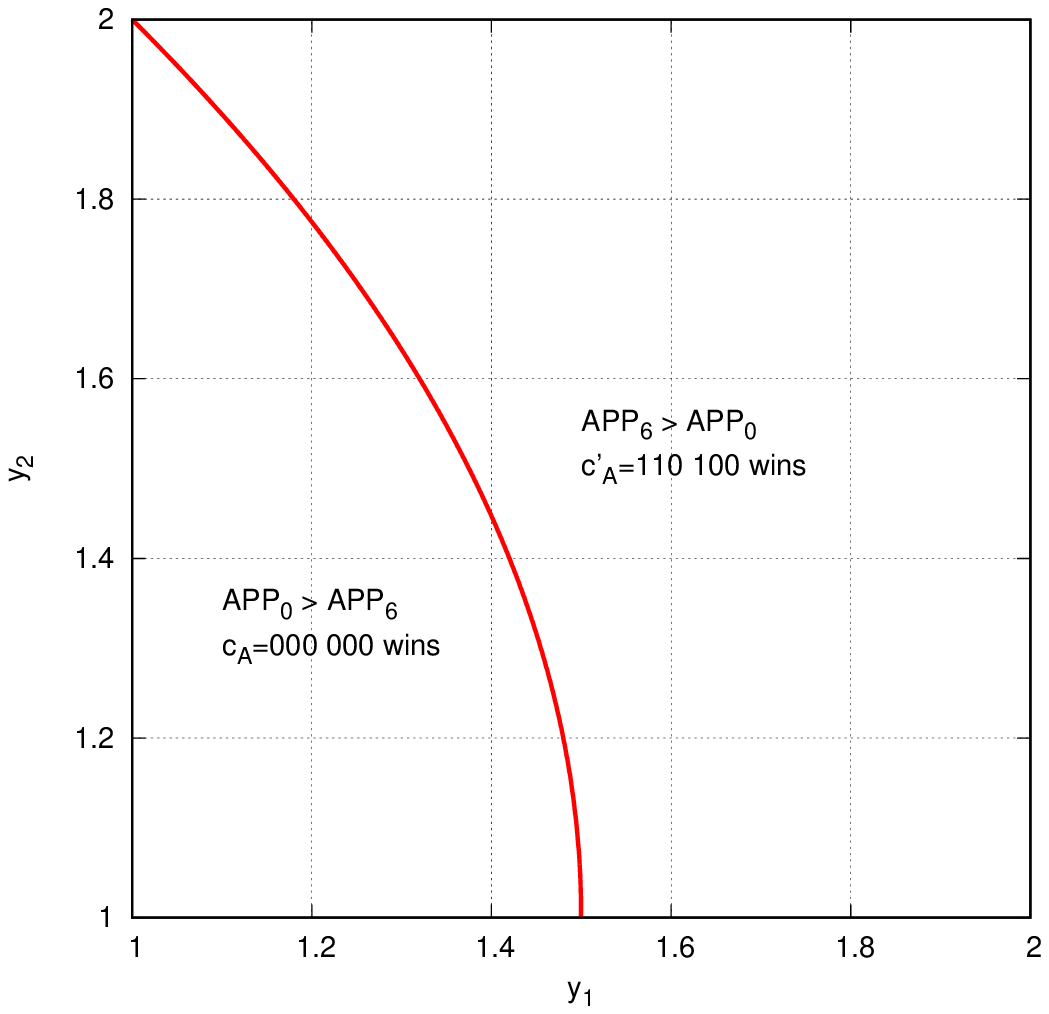}~~~~ 
\hspace{-3cm}\includegraphics[angle=270,width=0.65\linewidth]{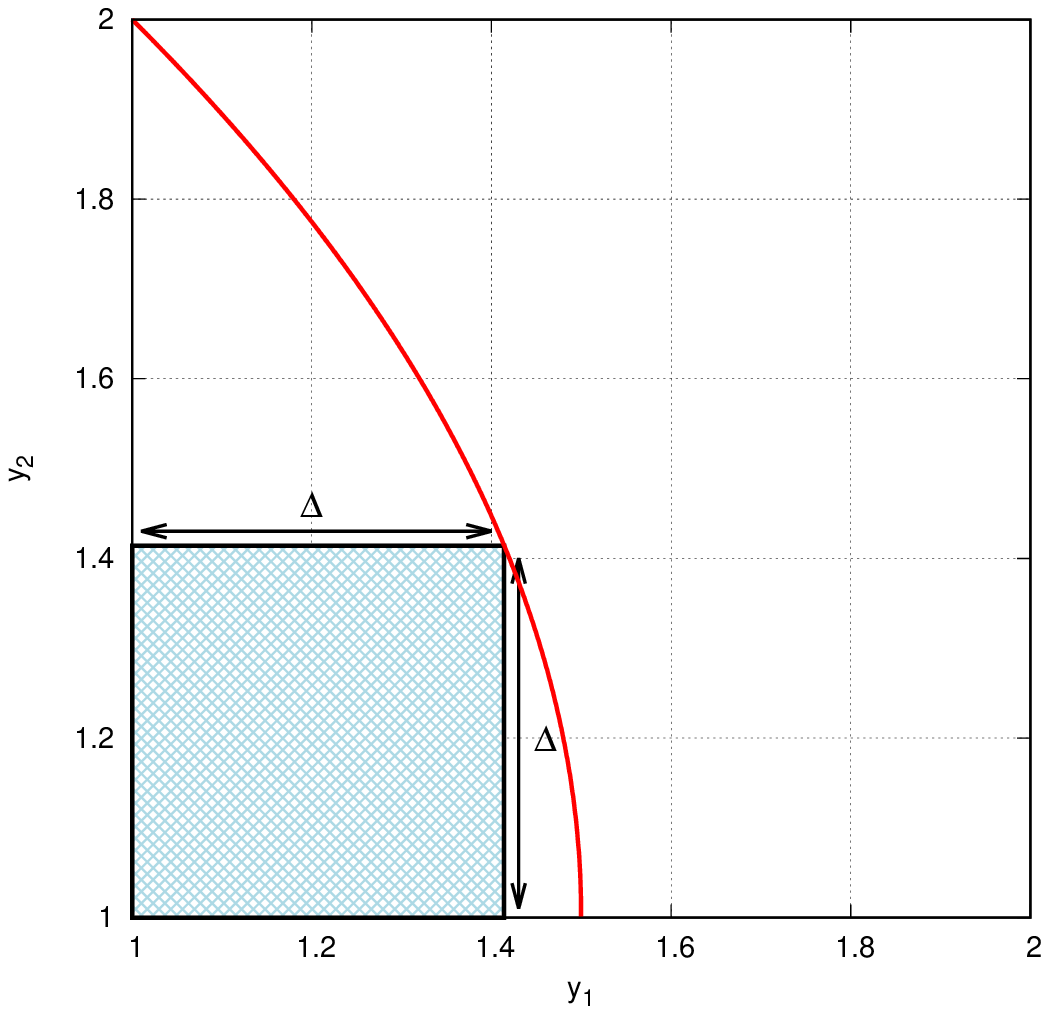}}
\begin{center}
\vspace{5mm}
\caption{(a) Left: the decision boundary separating the region of $c_A$ (below) from the region of $c'_A$ (above).
(b) Right: the demilitarized zone (hashed) defined by the square of side $\Delta$. \label{fig_decision_boundary}}
\end{center}
\end{figure}

%%\nrd{Lemma 5 updated on June 25.\\}
%%\nrd{Siyao reviewed Lemmas 5 and 6 on June 25.\\}

\begin{lemma}[Demilitarized Zone (DMZ)]
\label{lem_DMZ}
Assume that Alice's detector measures the photon position $X$
in the range $[i, i+1[$, i.e., the bin is $\hX=\hx=i$.
Define a demilitarized zone by the range $[i+1, i+1+\Delta]$,
for $0<\Delta<1$.
Then, for $\sigma\ll 1$ (or equivalently $\gamma \gg 1$), 
the probability that Bob's detector measures a photon position beyond the DMZ is given by (four cases):
\begin{itemize}
\item Fixed width $\Delta$,  $0<\Delta<1$ and vanishing $\sigma$.\\
\begin{equation}
\pr\left( Y \in [i + 1 + \Delta, i+2[ ~|~ X \in [i, i+1[\right) 
~=~ \cO\left(\frac{e^{-\frac{\Delta^2 \gamma}{4}}}{\gamma^{3/2}}\right),
\end{equation}
and could be simplified to
\begin{equation}
\pr\left( Y \in [i + 1 + \Delta, i+2[ ~|~ X \in [i, i+1[\right) 
~=~\cO\left(e^{-\frac{\Delta^2 \gamma}{4}}\right),
\end{equation}
In other words, jumping over a DMZ of any non-zero fixed width 
$\Delta$ is an event of infinite diversity.
\item For both $\Delta$ and $\sigma$ vanishing, and a ratio $\sigma/\Delta \rightarrow 0$.
\begin{equation}
\pr\left( Y \in [i + 1 + \Delta, N[ ~|~ X \in [i, i+1[\right) 
~=~\cO\left(\frac{e^{-\frac{\Delta^2 \gamma}{4}}}{\sqrt{\gamma}}\right),
\end{equation} 
stating that diversity is infinite when $\sigma$ vanishes faster than the DMZ width.
\item For both $\Delta$ and $\sigma$ vanishing but a fixed ratio $\Delta/\sigma=\kappa > 0$.
\begin{equation}
\pr\left( Y \in [i + 1 + \Delta, N[ ~|~ X \in [i, i+1[\right) 
~=~\cO\left(\frac{f(\kappa)}{\sqrt{\gamma}}\right)=\cO\left(\frac{1}{\sqrt{\gamma}}\right),
\end{equation} 
where $f(\kappa)=\tfrac{e^{-\kappa^2/4}}{\sqrt{\pi}}-\kappa~Q(\tfrac{\kappa}{\sqrt{2}})$. The function $f(\kappa)$
is asymptotic to $-\kappa/2+1/\sqrt{\pi}$ for small $\kappa$
and asymptotic to $2e^{-\kappa^2/4}/(\kappa^2\sqrt{\pi})$ for large $\kappa$. Thus, diversity is finite of order $1/2$ when $\sigma$ and $\Delta$ vanish at the same rate.
\item For both $\Delta$ and $\sigma$ vanishing, and a ratio $\Delta/\sigma \rightarrow 0$.
\begin{equation}
\pr\left( Y \in [i + 1 + \Delta, N[ ~|~ X \in [i, i+1[\right) 
~=~\frac{1}{\sqrt{\pi\gamma}}-\tfrac{\Delta}{2}+\cO\left(e^{-\frac{\gamma}{4}}\right)
~=~ \Theta(\frac{1}{\sqrt{\gamma}}).
\end{equation}
In other words, jumping over a DMZ of width vanishing faster than the noise variance is an event of finite diversity of order $1/2$.
\end{itemize}
\end{lemma}

\begin{figure}[!h]
\centerline{\includegraphics[width=0.4\linewidth]{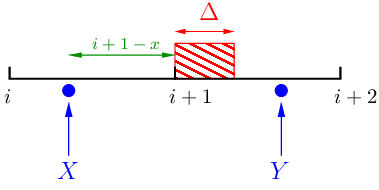}}
\begin{center}
\caption{Representation of the two bins $i$ and $i+1$ with the demilitarized zone of width $\Delta$.\label{fig_DMZ}}
\end{center}
\end{figure}

\begin{IEEEproof} 
The demilitarized zone is illustrated by the red rectangle in Figure~\ref{fig_DMZ}. The first case where $\Delta$ is fixed can be easily solved by using the $X-Y$ model as in Lemmas~\ref{lem_X-Y-1-jump} and \ref{lem_X-Y-2-jump}. We have
\begin{align}
& \pr\left( Y \in [i + 1 + \Delta, i+2[ ~|~  X \in [i, i+1[\right)  \nonumber \\
&= 
Q \left( \frac{i+1 - x + \Delta}{\sigma \sqrt{2}} \right)
-
Q \left( \frac{i+1 - x + 1}{\sigma \sqrt{2}} \right) \label{equ_Q_minus_Q}\\
& \le  
Q \left( \frac{i+1 - x + \Delta}{\sigma \sqrt{2}} \right) \nonumber \\
&\le Q\left( \frac{\Delta}{ \sigma \sqrt{2}}\right) 
= \cO\left( e^{- \frac{\Delta^2 \gamma}{4}} \right) \nonumber
\end{align}
by considering the high SNR regime, $\gamma = \frac{1}{\sigma^2} \gg 1$. Notice that we did not even integrate over $x$. However, such a quick upper bound fails in the case of a vanishing width $\Delta$.
Now, we proceed in a more Babylonian way that exactly solves the four cases. The distance $i+1-x$ in (\ref{equ_Q_minus_Q}) is replaced by $v$.
Also, at high signal-to-noise ratio, $Q(\tfrac{v+1}{\sigma\sqrt{2}})$ and $Q(\tfrac{v+N-(i+1)}{\sigma\sqrt{2}})$ are small with respect to $Q(\tfrac{v+\Delta}{\sigma\sqrt{2}})$, so we dropped the second $Q()$ function and replaced the interval $[i+1+ \Delta, i+2[$ by $[i + 1 + \Delta, N[$. In other words, when frame border effects are neglected, $N$ acts as $\infty$. The exact integration of the $Q()$ function that measures the probability beyond $i+1-x+\Delta=v+\Delta$ yields

\begin{align}
&\pr\left( Y \in [i + 1 + \Delta, N[ ~|~ X \in [i, i+1[\right) 
~=~\int_0^1 Q\left( \tfrac{v+\Delta}{\sigma\sqrt{2}}\right) dv \nonumber\\
&=~\tfrac{\sigma}{\sqrt{\pi}}\left[e^{-\tfrac{\Delta^2}{4\sigma^2}}-e^{-\tfrac{(\Delta+1)^2}{4\sigma^2}}\right]+
\left[(\Delta+1)Q\left(\tfrac{\Delta+1}{\sigma\sqrt{2}}\right)-\Delta Q\left(\tfrac{\Delta}{\sigma\sqrt{2}}\right)\right]. \label{equ_I_a_b}
\end{align}
For the sake of simplicity and space, the calculus details are not shown. At a fixed $\Delta$ and a vanishing $\sigma$, 
(\ref{equ_I_a_b}) becomes
\begin{equation*}
\frac{2\sigma^3}{\sqrt{\pi}\Delta^2} ~e^{-\tfrac{\Delta^2}{4\sigma^2}} ~-~ \frac{4\sigma^5}{\sqrt{\pi}\Delta^4} ~e^{-\tfrac{\Delta^2}{4\sigma^2}}
~+~\cO\left(\sigma^3 e^{-\tfrac{(\Delta+1)^2}{4\sigma^2}}\right),
\end{equation*}
and for $\sigma/\Delta \rightarrow 0$, when both $\sigma$ and $\Delta$ are vanishing, (\ref{equ_I_a_b}) becomes
\begin{equation*}
\frac{e^{-\frac{\Delta^2 \gamma}{4}}}{\sqrt{\pi\gamma}}~+~
\cO\left( e^{-\frac{\gamma}{4}}\right).
\end{equation*}
Similarly, for a fixed ratio $\Delta/\sigma=\kappa > 0$,
when both $\Delta$ and $\sigma$ are vanishing, (\ref{equ_I_a_b}) becomes
\begin{equation*}
~\sigma \left( \frac{e^{-\kappa^2/4}}{\sqrt{\pi}} 
- \kappa ~Q\left(\frac{\kappa}{\sqrt{2}}\right)\right) ~+~\cO\left( e^{-\frac{\gamma}{4}}\right),
\end{equation*} 
and finally, at a vanishing $\Delta$ smaller than $\sigma$, (\ref{equ_I_a_b}) becomes
\begin{equation*}
\frac{\sigma}{\sqrt{\pi}} ~-~ \frac{\Delta}{2} ~+~\cO\left( e^{-\frac{\gamma}{4}}\right),
\end{equation*}
which proves all four cases in this lemma, where $\gamma=1/\sigma^2$.
\end{IEEEproof}

For the $[6,3,3]$ binary code, we have established that $\Delta=\sqrt{2}-1$ from (\ref{equ_decision_boundary}). Denote by $\pr(c_A \rightarrow c'_A)=\pr(APP_6>APP_0)$ the probability of error that corresponds to decoding $c'_A$ instead of $c_A$. Then, using Lemma~\ref{lem_DMZ} for $X_1,X_2 \in [0,1[$
and $Y_1,Y_2 \in [1,2[$, we have the following union bound of the probability of error
\[
\pr(c_A \rightarrow c'_A) \le \pr(Y_1\ge 1+\Delta)+\pr(Y_2\ge 1+\Delta)
~=~\cO\left(e^{- \frac{\Delta^2 \gamma}{4}}\right),
\]
which proves that $\pr(c_A \rightarrow c'_A)$ has infinite diversity. 
Hence, as explained above, the soft reconciliation via the $[6,3,3]$ code is
correcting all double errors ($d_{Hmin}-1=2$) inside the DMZ, i.e., close to the correct bin observed by Alice. Only double errors of infinite diversity, beyond the DMZ, could push the soft decoder to fail.

As a summary for this specific example where the $[6,3,3,t=1]_2$ code requires $L=2$ photons per word over a TE-QKD frame of $8$ bins, we can state: a) One single-bin jump of diversity $1/2$ is always corrected by the soft decoder, b) Two single-bin jumps of width less than $\Delta$ are always corrected by the soft decoder, c) Only double single-bin jumps beyond $\Delta$ and multi-bin jumps could create decoding errors, but they all have infinite diversity as per Proposition~\ref{prop:single-double-jumps} and Lemma~\ref{lem_DMZ}. Consequently, soft reconciliation with the $[n=6,k=3,d_{Hmin}=3,t=1]_2$ binary code over a frame of $N=2^m=8$ bins achieves infinite diversity, as previously shown in the surprising result of Figure~\ref{fig_6_3_3}.
The reader should notice that $L=n/m=2 \le d_{Hmin}-1=2$ in this example. Also, the third column for $m=3$ in Table~\ref{tab_MFD_6_3_3} does not include any MFD$(\omega, \omega)$ term; the $[6,3,3]_2$ code is MFD deficient for $m=3$.\\ 

\begin{proposition}
\label{prop_w_w_L}
Let $\cC_b$ be a linear $[n,k,d_{Hmin}(\cC_b)]_2$ binary code. 
Consider a TE-QKD frame of $N=2^m$ bins and 
let $L=\tfrac{n}{m}$ be the number of photons per codeword, assuming
$n$ is multiple of $m$. If $\cC_b$ is MFD$(\omega, \omega)$,
for some Hamming weight~$\omega$, $d_{Hmin}(\cC_b) \le \omega \le n$, 
then $L \ge \omega \ge d_{Hmin}(\cC_b)$.\\
So, the full-MFD property implies that $L\ge d_{Hmin}(\cC_b)$.
Equivalently, $L<d_{Hmin}(\cC_b)$ implies MFD deficiency.
\end{proposition}
\begin{IEEEproof}
$\cC_b$ is MFD$(\omega, \omega)$. From Definition~\ref{def_MFD_code},
we know that there exists a non-zero codeword $c \in \cC_b$ such that $c$ contains $\ell=\omega$ labels that are neighbors of zero. Hence, $c$ has at least $\omega$ labels or equivalently it is using at least $\omega$ photons. So, $L$ is greater than or equal to $\omega$. That completes the proof. 
Another way to prove it is by using the Pigeon Hole principle with $L<d_{Hmin}(\cC_b)$.
Then, the code cannot be MFD$(\omega, \omega)$ for all $\omega\ge d_{Hmin}(\cC_b)$.
We conclude that it is MFD deficient. 
\end{IEEEproof}

For the general case of any linear code, the next lemma and the next theorem announce the conditions to let the soft-decision decoder attain infinite diversity in a TE-QKD soft reconciliation.

%%---------------------------------------------------
%%----------- Necessary Condition (Soft) ------------
%%---------------------------------------------------
\begin{lemma}[Necessary Condition of Theorem~\ref{thm:infinite-diversity-soft}]
\label{lem_necessary_soft}
Let $\cC$ be an $[n,k,d_{Hmin}]_q$ linear code, defined over the finite field $\F_q$ of characteristic 2, employed for soft reconciliation on a TE-QKD frame of $N=2^m$ bins. The frame
bins are labeled via a binary Gray code of $m$ bits per label.
Assume that $n\log_2(q)$ is multiple of $m$ for simplicity, i.e.,
$L=\tfrac{n\log_2(q)}{m}$ is an integer.
Let $\cC_b$ be the binary image of $\cC$, i.e., $\cC_b$ has length
$n\log_2(q)$ bits, dimension $k\log_2(q)$, and a minimum Hamming distance $d_{Hmin}(\cC_b) \ge d_{Hmin}(\cC)$.\\
Finally, assume that $\cC$ is MFD$(\omega, \omega)$ for some $\omega$ (full MFD), $d_{Hmin}(\cC_b) \le \omega \le n\log_2(q)$. Then, soft information reconciliation achieves a finite diversity.
\end{lemma}
\begin{IEEEproof} Note that $n\log_2(q)=mL$. 
So, $\F_q^n$ is isomorphic to $\F_2^{n\log_2(q)} = \F_2^{mL}$;
it could be considered as an Abelian group isomorphism or a vector space isomorphism over $\F_2$. To simplify the notations, we write $c \in \F_q^n$ when referring to the code $\cC$ over $\F_q$ and we write $c \in \F_2^{n\log_2(q)} = \F_2^{mL}$ when referring to the binary image $\cC_b$ over $\F_2$.\\
A codeword $c \in \F_q^n \sim \F_2^{mL}$ contains $mL$ coded bits and  consists of $L$ photons.  Without loss of generality, we assume that Alice's syndrome $s_A$ is zero in $\F_q^{n-k}$, so Alice's coset in $\F_q^n$ is the code $\cC$ itself. Let $c_A \in \cC$ denote the word observed by Alice. Let $x = (x_1, x_2, \ldots, x_L)$ represent the photon positions associated to~$c_A$, with $\phi(c_A) = \hat{x}$. Similarly, let $y = (y_1, y_2, \ldots, y_L)$ be the photon positions observed by Bob with his word $c_B$ satisfying $\phi(c_B) = \hat{y}$. 
The map $\phi()$ was defined in Section~\ref{sec_reconciliation}.

The lemma statement supposes that there exists $\omega\ge d_{Hmin}(\cC_b)$ such that the linear code $\cC_b$ is $MFD(\omega, \omega)$.
Therefore, by linearity, there exists a word $c'_A \in \cC_b$ satisfying $d_H(c_A, c'_A)=\omega$ and $c_A+c'_A$ which has weight $\omega$ includes exactly $\omega$ neighboring labels of zero corresponding to a single-bin jump. Figure~\ref{fig_MFD_omega} sketches $c_A$ and $c'_A$ with the full $MFD(\omega, \omega)$ property. For simplicity, Figure~\ref{fig_MFD_omega} illustrates $c_A$ as the all-zero codeword, but it could be any codeword in $\cC$ or in $\cC_b$. Also, notice that the two neighboring labels of $00\ldots 0$ are $10\ldots 0$ and $00\ldots 1$ if the Gray code is centered as in Table~\ref{tab:graycode2}. 
Without loss of generality,
we assume in this proof that the $\omega$ different bins between $c_A$ and $c'_A$ are the first $\omega$ bins. Proposition~\ref{prop_w_w_L}
guarantees that $L \ge \omega$, thus validating the frame structure illustrated in~Figure~\ref{fig_MFD_omega}.

\begin{figure}[!h]
\centerline{\includegraphics[width=0.55\linewidth]{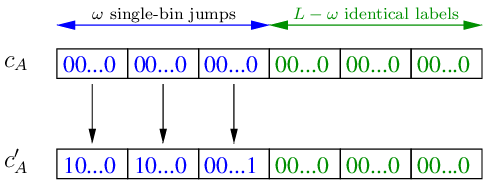}}
\begin{center}
\caption{Alice's word $c_A$ and a competing word $c'_A$ making the code MFD$(\omega, \omega)$ (full MFD), where $\cC$ is linear over $\F_q$.\label{fig_MFD_omega}}
\end{center}
\end{figure}

To maintain homogeneous notations, the bin positions associated to the competing word $c'_A$ are denoted by $\hz=(\hz_1, \hz_2, \ldots, \hz_L)$,
where $c'_A \in \cC \sim \cC_b$ and $\phi(c'_A)=\hz \in \Z_N^L$. Given $y \in [0,N[^L$, a decoding error occurs if $APP(c_A) < APP(c'_A)$. We will make a judicious choice of $y$ to prove that the probability of error satisfies
\begin{equation}
\label{eq:prob_APP_O}
\pr(c_A \rightarrow c'_A)=\pr(APP(c_A ) < APP(c'_A ) ) = \Omega\left(\frac{1}{\gamma^{\omega/2}}\right).
\end{equation}  

For the first $\omega$ bins, $1 \le \ell \le \omega$, where $\hx_{\ell} \le x_{\ell} < \hx_{\ell}+1$ and $\hz_{\ell}=\hx_{\ell}+1$, we consider $y_{\ell} \in ]\hx_{\ell}+1,\hx_{\ell}+2[$,
i.e., $y_{\ell}$ is placed in the bad bin. For the last $L-\omega$ bins, $\omega < \ell \le L$, where $\hx_{\ell} \le x_{\ell} < \hx_{\ell}+1$ and $\hz_{\ell}=\hx_{\ell}$, 
we consider $y_{\ell} \in [\hx_{\ell},\hx_{\ell}+1[$, i.e., $y_{\ell}$ is placed in the good bin. For illustration, Figure~\ref{fig_y_pos_MFD_omega} shows the considered photon positions of Alice, Bob, and the competing word.

\begin{figure}[!h]
\centerline{\includegraphics[width=0.8\linewidth]{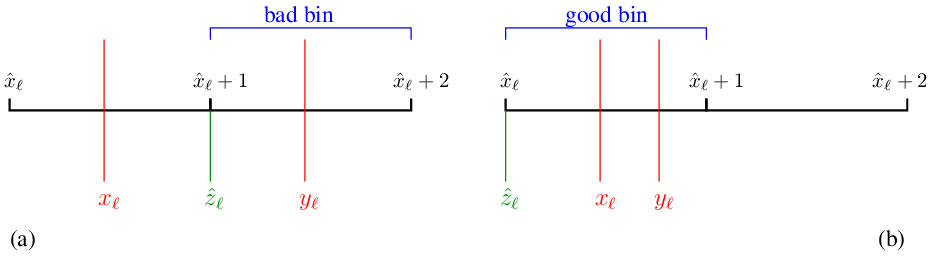}}
\begin{center}
\caption{Photon positions $y_{\ell}$ of Bob for bins $1 \le \ell \le \omega$ (left in (a)) and for bins $\omega< \ell \le L$ (right in (b)).\label{fig_y_pos_MFD_omega}}
\end{center}
\end{figure}

Then, by~\eqref{equ_app1_approx} and~\eqref{equ_app2_approx}, we obtain
%%\nrd{20 June 2026. 3:30am. Joseph reached this point.\\
%%\nrd{We continue starting from here at the next session.\\}
%%\nrd{SL: 3:30 am!!! Joseph, you worked too hard. I have reviewed your edits and really like all the figures.  We can meet on weekends if you want :)}
%%\nrd{==>> Siyao, Sunday after 1pm Florida time if you can. \\}
%%\nrd{Sure, no problem. I am also available 1:30pm -- 3:30 pm on Saturday. }
\begin{align*}
APP(c_A) \propto \prod_{\ell=1}^{\omega} \frac{1}{2} \left( e^{- \frac{(y_{\ell} - \hat{x}_{\ell} -1)^2 }{2 \sigma^2}} - e^{- \frac{(y_{\ell} - \hat{x}_{\ell} )^2 }{2 \sigma^2}} \right) 
\cdot 
\prod_{\ell=\omega+1}^{L}  \left( 1 - \frac{1}{2} e^{- \frac{(y_{\ell} - \hat{x}_{\ell})^2 }{2 \sigma^2}} - \frac{1}{2} e^{- \frac{(y_{\ell} - \hat{x}_{\ell} -1)^2 }{2 \sigma^2}} \right),
\\
APP(c'_A) \propto \prod_{\ell=1}^{\omega}  \left(1 - \frac{1}{2} e^{- \frac{(y_{\ell} - \hat{x}_{\ell} -1)^2 }{2 \sigma^2}} - \frac{1}{2} e^{- \frac{(y_{\ell} - \hat{x}_{\ell} -2 )^2 }{2 \sigma^2}} \right) 
\cdot 
\prod_{\ell=\omega+1 }^{L}  \left( 1 - \frac{1}{2} e^{- \frac{(y_{\ell} - \hat{x}_{\ell})^2 }{2 \sigma^2}} - \frac{1}{2} e^{- \frac{(y_{\ell} - \hat{x}_{\ell} -1)^2 }{2 \sigma^2}} \right).
\end{align*}
Therefore,
\begin{align*}
\frac{APP(c_A)}{APP(c'_A)} = \prod_{\ell=1}^{\omega} \frac{ \frac{1}{2} \left( e^{- \frac{(y_{\ell} - \hat{x}_{\ell} -1)^2 }{2 \sigma^2}} - e^{- \frac{(y_{\ell} - \hat{x}_{\ell} )^2 }{2 \sigma^2}} \right) }{\left(1 - \frac{1}{2} e^{- \frac{(y_{\ell} - \hat{x}_{\ell} -1)^2 }{2 \sigma^2}} - \frac{1}{2} e^{- \frac{(y_{\ell} - \hat{x}_{\ell} -2 )^2 }{2 \sigma^2}} \right)}, ~~~y_{\ell} \in ]\hx_{\ell}+1,\hx_{\ell}+2[,
\end{align*}
which leads to 
\begin{align}
\lim_{\sigma \to 0} ~~\frac{APP(c_A)  }{APP(c'_A) } = \frac{0}{1} = 0, \label{eq:frac_APP}
\end{align}
confirming that the given $y$ corresponds to a reconciliation error. The ratio in (\ref{eq:frac_APP}) also confirms that the open cube on the first 
$\omega$ dimensions
$\prod_{\ell=1}^{\omega} ]\hx_{\ell}+1,\hx_{\ell}+2[ ~\times~ \R^{L-\omega}$ is contained inside the decision region of $c'_A$ in $\R^L$,
where $\times$ represents the Cartesian product. Hence, the pairwise probability of error $\pr(c_A \rightarrow c'_A)$ is bounded from below by the integral over the open cube.

The probability of this error event is
\begin{equation*}
\pr(c_A \rightarrow c'_A) =\pr(APP(c_A) \le APP(c'_A)) 
~\ge~\prod_{\ell=1}^{\omega} \pr(Y_{\ell}\in ]\hat{x}_{\ell} + 1, \hat{x}_{\ell} +2[) 
= \left( \Theta(\gamma^{-\tfrac{1}{2}}) \right)^{\omega},
\end{equation*}
where the last equality is from Proposition~\ref{prop:single-double-jumps}.
Therefore, using Bachmann-Landau notation again, we can write 
$\pr(c_A \rightarrow c'_A) = \Omega(\tfrac{1}{\gamma^{\omega/2}})$, i.e.,
the probability of error has diversity less than or equal to $\omega/2$.
%%\nrd{I am here....Joseph.\\
%%I went back and updated Lemma 5. Ouff.
%%We will continue starting from here. June 21. 1am.\\}
%%\nrd{SL: Yes, the updated Lemma 5 is more rigorous.}
\end{IEEEproof}
%%\nrd{I am here....Joseph. 26 June 2026.\\~\\}
%%\nrd{Siyao is also here.}
%\nrd{=====SL edited on 3 June 2026 ======}

\clearpage
Lemma~\ref{lem_necessary_soft} proves that the MFD$(\omega, \omega)$ 
property (full MFD) implies finite diversity under soft-decision decoding. Equivalently, Lemma~\ref{lem_necessary_soft} just established that infinite diversity under soft-decision decoding requires 
the code to be MFD deficient as a necessary condition.
This is the necessary condition of statement 1) in Theorem~\ref{thm:infinite-diversity-soft} stated below.

For a relatively small $m$, most linear codes 
encountered in practice are MFD$(d_{Hmin},d_{Hmin})$. 
When increasing $m$, as suggested by the infinite diversity conditions in Theorems~\ref{thm:infinite-diversity-hard} and \ref{thm:infinite-diversity-soft}, the code becomes MFD deficient, i.e., it is only MFD$(\omega, \ell)$ for $\ell<\omega$ because non-neighboring labels corresponding to multi-bin jumps appear, thus boosting the diversity to infinity. More MFD examples are found in the sequel in Section~\ref{sec_code_examples}.

The proof of Lemma~\ref{lem_necessary_soft} is based on a lower bound of the probability of error. The lower bound is sufficient to complete the proof, no need for an upper bound or for the analysis of the exact decision regions, until we ask ourselves about the exact diversity order attained by the soft reconciliation process when the code
is MFD$(\omega,\omega)$!
In practice, Monte Carlo simulations show that soft decoding reaches 
a finite diversity of $d_0 \times \omega=\tfrac{1}{2}\omega$ where $\omega=d_{Hmin}(\cC_b)$ in most cases, 
as for coding over wireless fading channels summarized in Section~\ref{sec_div_wireless}.
Is it possible to extend the analysis in the proof of Lemma~\ref{lem_necessary_soft}
to prove that diversity is $\tfrac{1}{2}\omega$ when the code is MFD$(\omega,\omega)$? The answer is yes. We established that proof by analyzing the decision regions boundary under soft decoding in the finite diversity case. We do not include it in this paper to relieve the reader from extra tedious calculus.

%%\nrd{I am here....Joseph. 26 June 2026, afternoon.\\}

%%Without MFD criticality, $L \ge d_{Hmin}$ alone does not automatically guarantee a finite-diversity event. If the relevant minimum-distance codeword differences are one-bit flips that correspond to non-neighboring bin labels, then those events may still require multi-bin jumps and therefore have infinite diversity.

%The following theorem shows that this condition is also sufficient.
%The proof follows the geometry of Gray-labeled TE-QKD frames. 
%When two codewords differ in more bit positions than the number of photons used to transmit them, the disagreement cannot be created solely by single-bin jumps. At least one photon must be displaced by two or more bins, and Proposition~\ref{prop:single-double-jumps} shows that such a displacement has exponential decay probability. Thus, every soft-decision error contains an infinite-diversity component.

%\nrd{=====SL edited on 3 June 2026 ======}

%In TE-QKD, detector jitter and other channel imperfections cause Alice's and Bob's raw keys to be correlated but not identical. Information reconciliation is the process where Alice and Bob communicate over a public classical channel to identify and correct these discrepancies. A common and efficient method for reconciliation is based on \emph{syndrome decoding}, which uses the principles of error-correcting codes.

%%-------------------------------------------
%%-------------- Theorem 2 ------------------
%%-------------------------------------------
\begin{theorem}
\label{thm:infinite-diversity-soft}
%%\nrd{Updated June 19 by Joe and Shaikha.}  
%% Again, another update on 4 July 2026
Let  $\cC$ be a linear $\left[n, k, d_{Hmin}\ge 2, t \right]_q$ code employed for soft reconciliation over a TE-QKD frame of $N=2^m$ bins.
The finite field $\F_q$ over which the code is defined has characteristic 2. Denote by $\cC_b$ the binary image of $\cC$.
Assume that $\gcd(n\log_2(q),m)=\min(n\log_2(q),m)$.\\ 
Then, the following statements are true:\\
1) Soft reconciliation attains infinite diversity if and only if $\cC$ is MFD deficient.\\
2) Soft reconciliation attains finite diversity if and only if $\cC$ is full MFD for some non-zero weight $\omega$, $d_{Hmin}(\cC_b) \le \omega \le n\log_2(q)$.\\
3) If $L=\tfrac{n\log_2(q)}{m} < d_{Hmin}(\cC_b)$, $L \in \N^*$,  then soft reconciliation attains infinite diversity.\\

Statements 1) and 2) given above are equivalent. If $d_{Hmin}=2t+1$, from 3), $L\le 2t$ is a sufficient condition to achieve infinite diversity under soft decoding; to be compared with the condition $L\le t$ under algebraic decoding in Theorem~\ref{thm:infinite-diversity-hard}.
%%
%% July 4, 2026: it is wrong to say that MFD is not needed.
%%\nrd{Note. Add the MFD condition later, after finishing the proof.}
%% No MFD needed in this theorem.
%%
%%
%% Comments below added on July 4, 2026
%%\begin{verbatim}
%%Lemma 6:
%%a) MFD(w,w) => finite
%%b) infinite => non-MFD(w,w)
%%
%%Update for Theorem 2: 
%%1) $\bullet$ MFD(w,w) <=> finite diversity
%%2) $\bullet$ non-MFD(w,w) <=> infinite diversity
%%3) $\bullet$ L<dHmin => non-MFD(w,w) => infinite diversity
%%L < dHmin is a sufficient condition for infinite diversity (soft).
%%But L>= dHmin is not a necessary condition,i.e., infinite diversity could be achieved for some codes even if L>= dHmin because the code is non-MFD(w,w). 
%%\end{verbatim}
\end{theorem}

%%\nrd{SL updated on 25 June. Joseph on June 26.}
\begin{IEEEproof}
We have previously proved in Lemma~\ref{lem_necessary_soft} that infinite diversity implies MFD deficiency. 
To complete the proof of the two equivalent statements 1) and 2), 
we just need to prove that MFD deficiency implies infinite diversity, under soft-decision decoding.

We follow the same notation as used previously in this section.
Let $c_A$ denote the word observed by Alice with $\phi(c_A) = \hat{x}$. 
Without loss of generality and thanks to the code linearity, we can assume that Alice has a zero syndrome and that $c_A \in \cC$ is the all-zero word.
Let $c_A' \in \cC$ be a competing word, with $\phi(c_A') = \hat{z}$ such
that the Hamming distance in $\F_2^{n\log_2(q)}$ is $d_H(c_A, c_A')=\omega \ge d_{Hmin}(\cC_b) \ge d_{Hmin}$.

\noindent
$\bullet$ Consider the case where $m$ is multiple of $n\log_2(q)$, i.e., 
$L=\tfrac{n\log_2(q)}{m}<1$. Hence, one photon is carrying 
$\lambda=\tfrac{m}{n \log_2 (q)}$ words. Consider the new linear code 
$\cC_s=\cC^{\oplus \lambda}$ defined as the direct sum of $\lambda$ copies of~$\cC$. The code $\cC_s$ has parameters 
$[\lambda n,\lambda k, d_{Hmin}]_q$. 
For $\cC_s$, we need one photon to carry a word. Hence, the instance $L=\tfrac{n\log_2(q)}{m}<1$ is now converted to the instance $L=1$ examined in the second case that follows.
\begin{figure}[!h]
\centerline{\includegraphics[width=0.75\linewidth]{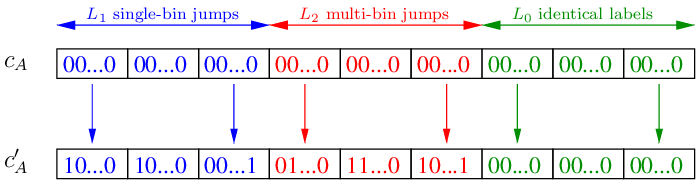}}
\begin{center}
\caption{Alice's word $c_A$ and a competing word $c'_A$ of weight $\omega$, where the MFD$(\omega, L_1)$ property is satisfied, $L_1 < \omega$, and the number of photons per word is $L=L_0+L_1+L_2 \ge 1$.\label{fig_MFD_omega_2}}
\end{center}
\end{figure}

\noindent
$\bullet$ Consider the case where $n\log_2(q)$ is multiple of $m$, i.e., $L=\tfrac{n\log_2(q)}{m}$ is integer, $L\ge 1$, and $L$ photons are required to carry a word.
The MFD deficiency assumption by the theorem implies that,
for any $c'_A \ne c_A=0$, $c'_A \in \cC \sim \cC_b$, $c'_A$ does confer the MFD$(\omega, \omega)$ property to the code $\cC$, 
where $\omega$ is the Hamming weight of $c'_A$ in $\F_2^{mL}$. Hence, $c'_A$ can only make the code MFD$(\omega, L_1)$, where $0 \le L_1 < \omega$. Now, we can write $L=L_0+L_1+L_2$ as depicted in Figure~\ref{fig_MFD_omega_2}, where $L_1$ is the number of bins corresponding to a single-bin jump (neighboring labels from the MFD property), $L_2$ is the number of bins with multiple jumps (two or more), and the remaining $L_0$ bins correspond to identical labels in $c_A$ and $c'_A$.
Since $L_1<\omega$, we have $L_2\ge 1$. We also have
$L_1 \ge 0$ and $L_0 \ge 0$. The Hamming weight in the first $L_1+L_2$ bins of $c'_A$ is $\omega$,
then we conclude that $L_1+L_2\le\omega$. 
Thus, we do not know whether $L$ is less than $\omega$ or not, it depends on $L_0$. Fortunately, the value of $L$ with respect to $\omega$ is not required to prove infinite diversity in Statement 1 of this theorem. The proof relies on $L_2\ge 1$ and $L_1+L_2 \le \omega$ as we will see below.\\ 

Given the vector of photon positions observed by Bob, $Y=y=(y_1, \ldots, y_L) \in [0,N[^L$, we aim at finding an upper bound for the probability of error $\pr(c_A \rightarrow c'_A)=\pr(APP(c_A)<APP(c'_A))$ to prove that diversity is infinite without necessarily determining the exact shape of the decision region boundary. The approach is identical to the $[6,3,3]_2$ code example described at the beginning of this section.
The upper bound of $\pr(c_A \rightarrow c'_A)$ is obtained from the cube of side $\Delta$ in $\R^{L_1+L_2}$. The cube should have $(\hx_1+1, \hx_2+1, \ldots, \hx_{L_1+L_2}+1)$ as its left bottom corner. The side length $\Delta$ is such that the cube belongs entirely to the decision region of $c_A$, as illustrated for $L=2$ in Figure~\ref{fig_decision_boundary}-b. To find the largest $\Delta$, we start by defining the positions of $\hx$, $\hz$, and $y$, given $L_1$ and $L_2$.
%%The worst case corresponds to $\omega=d_{Hmin}$ and all $L_2$ bins with two-bin jumps only. 
We will write the proof in a way such that it is valid for any $\omega\ge d_{Hmin}$ and any number of double-bin jumps within the $L_2$ bins.\\ 

%%\nrd{I reached this point. Joe. 5 July.\\}
%%\nrd{SL updated the equations below. \\}
%%\nrd{Thank You, SL. \\}
For the first $L_1$ bins with a single jump, we have $\hz_{\ell}=\hx_{\ell}+1$, for $\ell=1 \ldots L_1$. For the $L_2$ bins with multiple jumps, we take $\hz_{\ell}=\hx_{\ell}+2$ (the worst case since the number of triple or quadruple jumps is unknown), for $\ell=L_1+1 \ldots L_1+L_2$. To find $\Delta$, we place Bob's positions in the middle bin, $y_{\ell} \in [x_{\ell}+1, x_{\ell}+2]$, for $\ell=1 \ldots L_1+L_2$.
These defined positions are shown in Figure~\ref{fig_y_pos_MFD_omega_2}. 

\begin{figure}[!h]
\centerline{\includegraphics[width=0.95\linewidth]{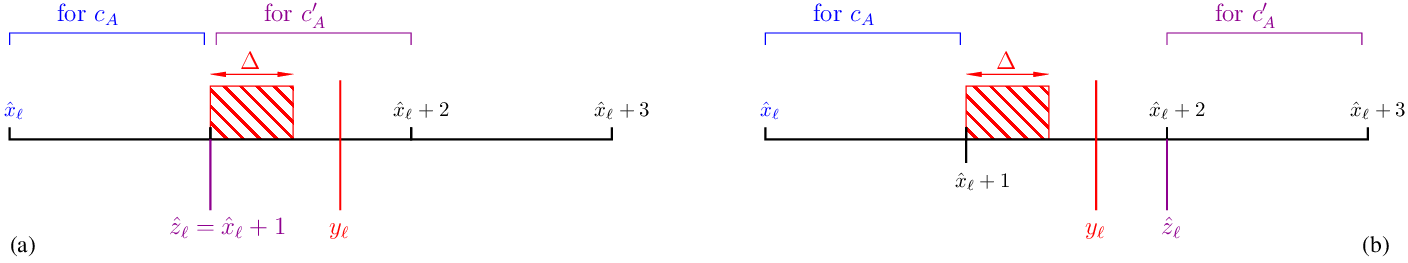}}
\begin{center}
\caption{Photon positions $y_{\ell}$ of Bob for bins $1 \le \ell \le L_1$ (left in (a)) and for bins $L_1< \ell \le L_1 + L_2$ (right in (b)). The competing word~$c'_A$ makes the code MFD$(\omega,L_1)$, where $0 \le L_1 < \omega$,
$0 < L_2$, and $0 < L_1+L_2 \le \omega$.\label{fig_y_pos_MFD_omega_2}}
\end{center}
\end{figure}

Bob's photons positions for the $L_0$ bins with identical labels have no effect on the APPs, and hence can be placed anywhere, e.g., in their most likely place: $y_{\ell} \in [\hx_{\ell}, \hx_{\ell}+1]$, for $\ell=L_1+L_2+1, \ldots, L$.
Given $c_A$, $c'_A$, and $y$, the {\em a posteriori} probabilities follow from
(\ref{equ_app1_approx}) and (\ref{equ_app2_approx}),
\begin{align*}
APP(c_A) & 
\propto 
\prod_{\ell=1}^{L_1+L_2} 
\left( \tfrac{1}{2} e^{- \frac{(y_{\ell} - \hat{x}_{\ell} -1)^2 }{2 \sigma^2}} 
-\tfrac{1}{2} e^{- \frac{(y_{\ell} - \hat{x}_{\ell} )^2 }{2 \sigma^2}}
\right) 
\cdot
\prod_{\ell=L_1+L_2+1}^{L} 
\left( 1-\tfrac{1}{2} e^{- \tfrac{(y_{\ell} - \hat{x}_{\ell} -1)^2 }{2 \sigma^2}} 
-\tfrac{1}{2} e^{- \tfrac{(y_{\ell} - \hat{x}_{\ell} )^2 }{2 \sigma^2}}
\right) 
\\&  
\propto 
\prod_{\ell=1}^{L_1+L_2} 
\left( \tfrac{1}{2} e^{- \tfrac{(y_{\ell} - \hat{x}_{\ell} -1)^2 }{2 \sigma^2}} 
~+~ \cO(e^{-\tfrac{\gamma}{2}})
\right) 
\cdot
\prod_{\ell=L_1+L_2+1}^{L} 
\left( 1-\tfrac{1}{2} e^{- \tfrac{(y_{\ell} - \hat{x}_{\ell} -1)^2 }{2 \sigma^2}} 
-\tfrac{1}{2} e^{- \tfrac{(y_{\ell} - \hat{x}_{\ell} )^2 }{2 \sigma^2}}
\right), 
\end{align*}
\begin{align*}
APP(c'_A) &
\propto \prod_{\ell=1}^{L_1} 
\left( 1 - \tfrac{1}{2} e^{- \tfrac{(y_{\ell} - \hat{x}_{\ell} -1)^2 }{2 \sigma^2}} 
-\tfrac{1}{2} e^{- \tfrac{(y_{\ell} - \hat{x}_{\ell} -2 )^2 }{2 \sigma^2}}
\right) 
\cdot
\prod_{\ell=L_1+1}^{ L_1+L_2} 
\left(  \tfrac{1}{2} e^{- \tfrac{(y_{\ell} - \hat{x}_{\ell} -2 )^2 }{2 \sigma^2}} 
- \tfrac{1}{2} e^{- \tfrac{(y_{\ell} - \hat{x}_{\ell} - 3 )^2 }{2 \sigma^2}} \right)
\\& 
\quad \cdot \prod_{\ell=L_1+ L_2 +1 }^{L} 
\left( 1-\tfrac{1}{2} e^{- \tfrac{(y_{\ell} - \hat{x}_{\ell} -1)^2 }{2 \sigma^2}} 
-\tfrac{1}{2} e^{- \tfrac{(y_{\ell} - \hat{x}_{\ell} )^2 }{2 \sigma^2}}
\right) 
\\&
\propto \prod_{\ell=1}^{L_1} 
\left( 1 - \tfrac{1}{2} e^{- \tfrac{(y_{\ell} - \hat{x}_{\ell} -1)^2 }{2 \sigma^2}} 
-\tfrac{1}{2} e^{- \tfrac{(y_{\ell} - \hat{x}_{\ell} -2 )^2 }{2 \sigma^2}}\right) 
\cdot
\prod_{\ell=L_1+1}^{ L_1+ L_2} 
\left( \tfrac{1}{2} e^{- \tfrac{(y_{\ell} - \hat{x}_{\ell} -2 )^2 }{2 \sigma^2}} 
~+~ \cO(e^{-\tfrac{\gamma}{2}})
\right)
\\&
\quad \cdot \prod_{\ell=L_1 + L_2 +1}^{L} 
\left( 1-\tfrac{1}{2} e^{- \tfrac{(y_{\ell} - \hat{x}_{\ell} -1)^2 }{2 \sigma^2}} 
-\tfrac{1}{2} e^{- \tfrac{(y_{\ell} - \hat{x}_{\ell} )^2 }{2 \sigma^2}}
\right).
\end{align*}

To find the largest side length $\Delta$ of the cube 
$\prod_{\ell=1}^{L_1+L_2} [\hx_{\ell}+1, \hx_{\ell}+1+\Delta]$ that is included in the decision region of $c_A$, we write $APP(c_A)=APP(c'_A)$ with 
$y_{\ell}=\hx_{\ell}+1+\Delta$, for $1 \le \ell \le L_1+L_2$. We obtain
\begin{align*}
\prod_{\ell=1}^{L_1+L_2} 
\left( \frac{1}{2} e^{- \frac{\Delta^2}{2 \sigma^2}} 
~+~ \cO(e^{-\tfrac{\gamma}{2}})
\right)
~&=~
\prod_{\ell=1}^{L_1} 
\left( 1 - \frac{1}{2} e^{- \frac{\Delta^2 }{2 \sigma^2}} 
-\frac{1}{2} e^{- \frac{(1-\Delta)^2 }{2 \sigma^2}}\right) 
\prod_{\ell=L_1+1}^{L_1+ L_2} 
\left( \frac{1}{2} e^{- \frac{(1-\Delta)^2 }{2 \sigma^2}} 
~+~ \cO(e^{-\tfrac{\gamma}{2}})
\right).
\end{align*}
With a vanishing $\sigma^2$ and after applying the $\log()$ to both sides of the above equality, we reach
\begin{equation}
\label{equ_Delta_1}
(L_1 + L_2)\Delta^2 ~=~ L_2 (1-\Delta)^2, ~~~\text{for}~0 < \Delta <1,
\end{equation}
which gives the following expression
\begin{equation*}
\Delta ~=~\frac{\sqrt{L_2}}{\sqrt{L_1 + L_2}+\sqrt{L_2}}.
\end{equation*}
Note that (\ref{equ_Delta_1}) becomes 
$\tfrac{\Delta}{1-\Delta}=\sqrt{ \tfrac{L_2}{L_1 + L_2}} \le 1$, hence $\Delta \le 1/2$.
Finally, note that $L_2 \ge 1$ leading to the inequality
\begin{equation}
\Delta ~\ge~\frac{1}{1+\sqrt{L}}.
\end{equation}
%%\nrd{Good. I reached here. Joe. July 6, at 2:45am.\\}

The above lower bound of $\Delta$ is non-vanishing and does not depend on $L_1$, i.e., does not depend on the choice of the codeword $c'_A$. The application of Lemma~\ref{lem_DMZ} shall complete the proof of the statements~1) and~2) of the theorem by a union bound on the probability of error outside the cube $\prod_{\ell=1}^{L_1+L_2} [\hx_{\ell}+1, \hx_{\ell}+1+\Delta]$:  
\begin{align*}
\pr(c_A \rightarrow c'_A)
&\le \sum_{\ell=1}^{L_1+L_2} \pr(Y_{\ell} \ge \hx_{\ell}+1+\Delta~|~X_{\ell} \in [\hx_{\ell},\hx_{\ell}+1[) \\ 
& = (L_1+L_2) \times \pr(Y_1 \ge \hx_1+1+\Delta~|~X_1 \in [\hx_1,\hx_1+1[) \\
\pr(c_A \rightarrow c'_A) &=\cO\left(e^{- \frac{\Delta^2 \gamma}{4}}\right),
\end{align*}
proving that diversity is infinite.

The proof of statement 3) is straightforward by the application of Proposition~\ref{prop_w_w_L} followed by the condition in statement 1) of the current theorem. 
\end{IEEEproof}
%%\nrd{SL edited on 25 June.}

As for Lemma~\ref{lem_necessary_soft}, Theorem~\ref{thm:infinite-diversity-soft} does not establish a tight expression for 
the pairwise error probability $\pr(c_A \rightarrow c'_A)$
neither it analyzes the boundary of the decision regions.
Fortunately, Theorem~\ref{thm:infinite-diversity-soft} proves that the probability of error is $\cO(e^{- \frac{\Delta^2 \gamma}{4}})$ which implies that diversity has an infinite order. 
Expressions or tight bounds for $\pr(c_A \rightarrow c'_A)$ could be the task of a future paper. The next section shows practical examples of codes and their expected diversity when used in TE-QKD reconciliation.

\begin{corollary}
\label{cor_rate_loss_soft}
Consider a linear $\cC[n, k, d_{Hmin}, t]_q$  code employed for soft reconciliation over a TE-QKD frame of $N=2^m$ bins with rate $R_c = k/n$.
The finite field $\F_q$ over which the code is defined has characteristic 2.
Assume that $\gcd(n\log_2(q),m)=\min(n\log_2(q),m)$ and that $L=\tfrac{n\log_2(q)}{m}<d_{Hmin}$, $L \in \N^*$. So, $\cC$ achieves infinite diversity as per statement 3) in Theorem~\ref{thm:infinite-diversity-soft}.
Then, we have 
\begin{equation}
\label{equ_Rc_bound_soft}
R_c \le 1-\frac{\log_2(q)}{m}.
\end{equation}
\end{corollary}
The proof of the upper bound in (\ref{equ_Rc_bound_soft}) is straightforward from the Singleton bound as for the proof of Corollary~\ref{cor_rate_loss_hard} and could be improved by the use of tighter bounds on $d_{Hmin}$.
The rate penalty of soft reconciliation of Corollary~\ref{cor_rate_loss_soft} is half the rate penalty of algebraic reconciliation of Corollary~\ref{cor_rate_loss_hard}. 
Thus, soft reconciliation can achieve infinite diversity at higher coding rates, or equivalently with fewer bins per frame. 
For the binary case $q=2$, soft decoding allows $R_c \le 1 - 1/m$ against the algebraic decoding $R_c \le 1 - 2/m$.
For $m=3$, soft decoding can support $R_c \le \tfrac{2}{3}$ whereas algebraic decoding requires $R_c \le \tfrac{1}{3}$ to achieve infinite diversity. This matches the case of the $[6,3,3]_2$ code: with $m=3$, the rate is $R_c = \tfrac{1}{2}$, which violates the algebraic decoding bound in~\eqref{equ_Rc_bound} but satisfies the soft-decision bound in~\eqref{equ_Rc_bound_soft}.

\section{Examples of Short Codes and Their Diversity \label{sec_code_examples}}
%%\nrd{Good. I reached here. Joe. July 8, at 5:30pm.\\}
%%\nrd{SL edited}
This section illustrates the diversity conditions derived in Theorems~\ref{thm:infinite-diversity-hard} and~\ref{thm:infinite-diversity-soft} through representative binary and non-binary linear codes. 
As discussed in Section~\ref{sec_infinity-div}, for a code $\cC[n,k, d_{Hmin}, t]_q$ used over a TE-QKD channel model described in Section~\ref{sec_model} with $N = 2^m$ bins, the number of photons per codeword is $L = \tfrac{n \log_2(q)}{m}$. 
Under algebraic decoding, infinite diversity is achieved if and only if $L \le t$, whereas under soft-decision decoding, infinite diversity is achieved if and only if $\cC$ is MFD deficient. Also, $L<d_{Hmin}$ is sufficient to make
the code MFD deficient. 
If these conditions are not satisfied, the dominant error events consist only of single-bin jumps, i.e., diversity is finite and bounded from below by $\tfrac{1}{2}(t+1)$ or $\tfrac{1}{2}d_{Hmin}$ for hard and soft reconciliation respectively.\\

All codes listed below are considered to be short codes.
A short code length and a large number of coded bits per photon usually meet
the conditions of infinite diversity. 

%%------------------------------------------------------------------
%%------------------------------------------------------------------
\subsection{The Binary $[24,12,8]$ Extended Golay Code}
We consider the famous self-dual binary Golay code of length $24$.
Firstly, the $[23,12,7]_2$ cyclic Golay code is built from the generator
polynomial $g(x)=x^{11}+x^{10}+x^6+x^5+x^4+x^2+1$ by creating all $2^{12}$
codewords via $c(x)=a(x)g(x)$, where $a(x) \in \F_2[X]$ has degree less than or equal to $11$. 
Then, an extra parity bit is computed as $c_{23}=\sum_{i=0}^{22} c_i$.
Table~\ref{tab_golay} shows the attained diversity for different TE-QKD frame parameters. The MFD distribution of the code, $\ell_{max}$ for a given $\omega$, is also shown in the fourth column. The weight enumerator polynomial of the length-24 self-dual binary code is $1+759x^8+2576x^{12}+759x^{16}+x^{24}$. 

\begin{table}[!h]
\begin{center}
\caption{Achieved TE-QKD diversity for a Golay $[24,12,8]_2$ code.\label{tab_golay}}
\includegraphics[width=0.85\columnwidth]{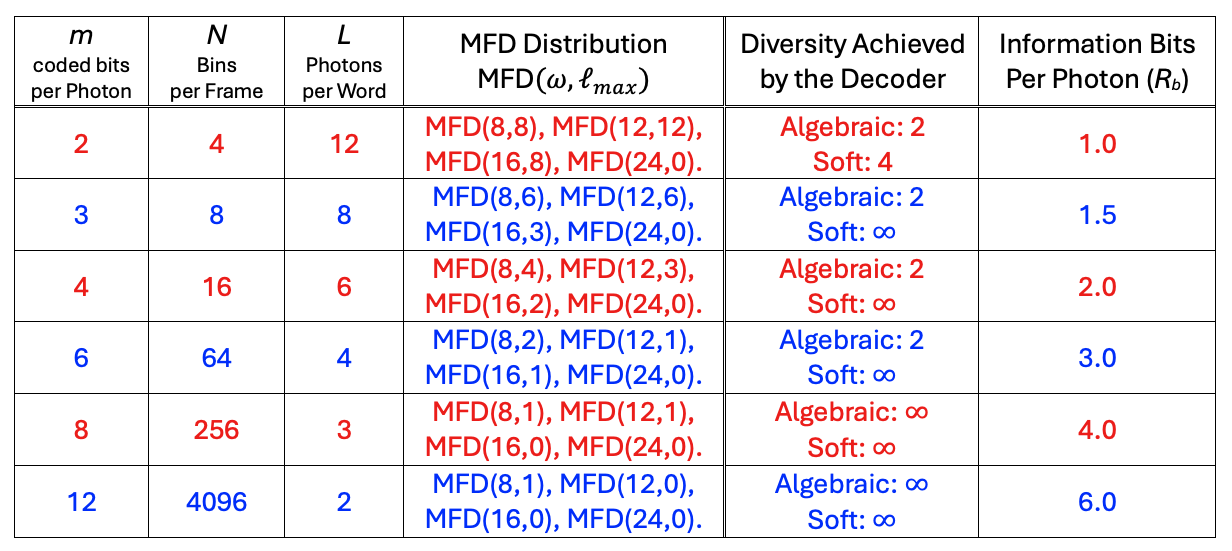}
\end{center}
\end{table}

For practical applications, the two rows for $m=4$ and $m=6$ lead to a reasonable 
lab equipment complexity. Notice that $MFD(\omega, \omega)$ only exists at $m=2$.
Hence, soft-decision decoding reaches infinite diversity early enough for $m\ge 3$. Algebraic decoding is not recommended for the Golay code since infinite diversity requires huge frames at $m=8$ and $N=256$. We recommend TE-QKD soft reconciliation via the $[24,12,8]_2$ Golay code with $4$ and $6$ coded bits per photon.

The bit error rate (BER) performance of hard and soft TE-QKD reconciliation is shown in Figures~\ref{fig_golay_hard} and \ref{fig_golay_soft} versus the signal-to-noise ratio (SNR) $\gamma=1/\sigma^2$ in decibels. As predicted by our theory, the diversity order is $2$ for $m=3,4,6$ for hard reconciliation in Figure~\ref{fig_golay_hard}. 
Diversity becomes infinite at $m=8$ coded bits per photon.
Effective diversity reaches $9.1$ at finite SNR between BER of
$10^{-6}$ and $10^{-7}$. 
For soft reconciliation in Figure~\ref{fig_golay_soft}, the diversity (effective and exact) are equal to $4$ at $m=2$. Soft-decision decoding attains infinite diversity at $m=3,4,6$.
Effective diversity at finite SNR measured between 
BER of $10^{-6}$ and BER of $10^{-7}$ is about $8.5$.  

\begin{figure}[!h]
\begin{center}
\vspace{-9mm}
\includegraphics[angle=270, width=0.7\linewidth]{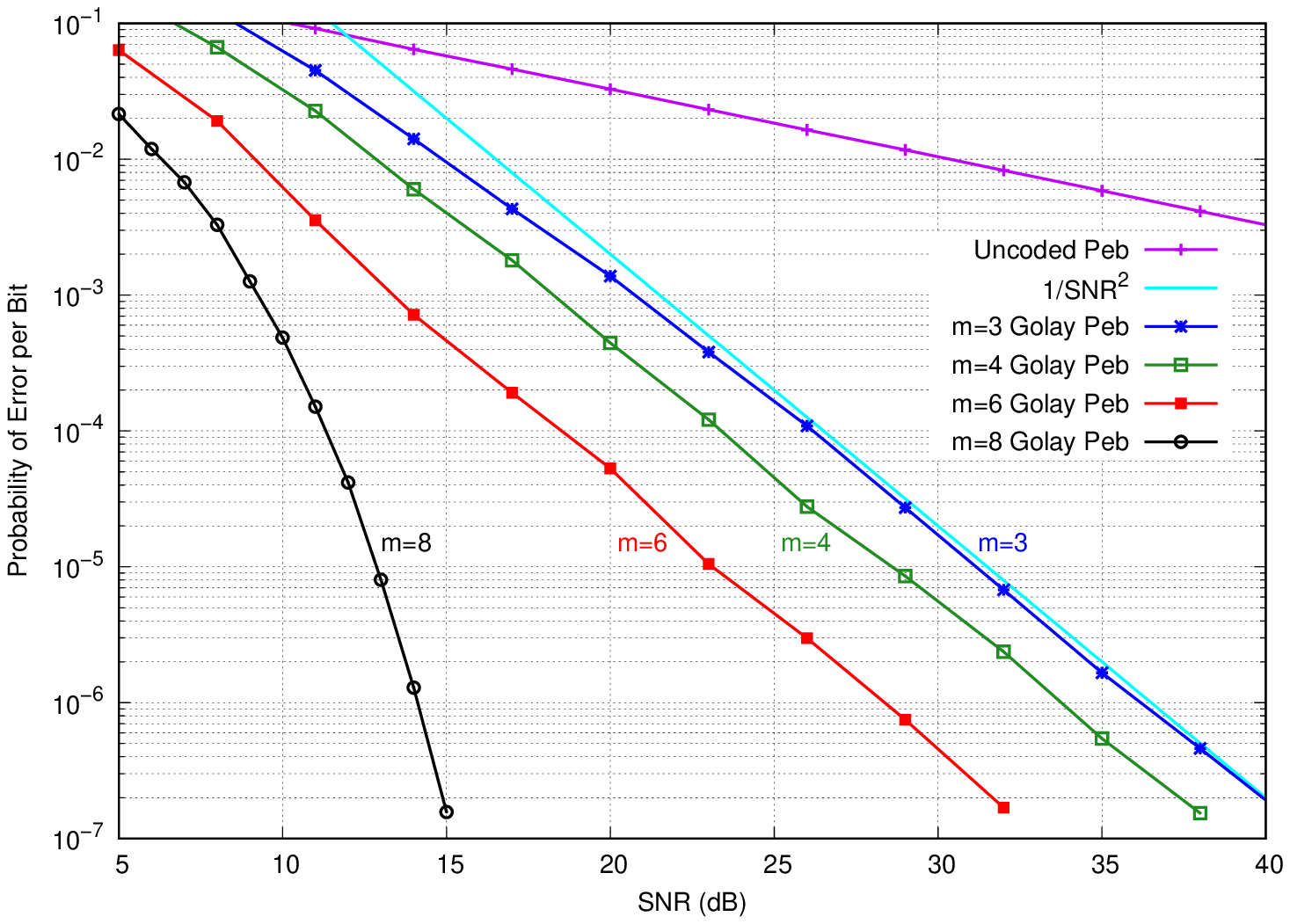}
\vspace{9mm}
\caption{Hard reconciliation of TE-QKD via complete algebraic decoding of the extended binary $[24,12,8]_2$ Golay code. $N=8, 16, 64$, and $256$ bins per frame, $N=2^m$, $m$ coded bits per photon. Diversity follows the result of Theorem~\ref{thm:infinite-diversity-hard} and Table~\ref{tab_golay}. 
\label{fig_golay_hard}}
\end{center}
\end{figure}

\begin{figure}[!h]
\begin{center}
\vspace{-9mm}
\includegraphics[angle=270, width=0.7\linewidth]{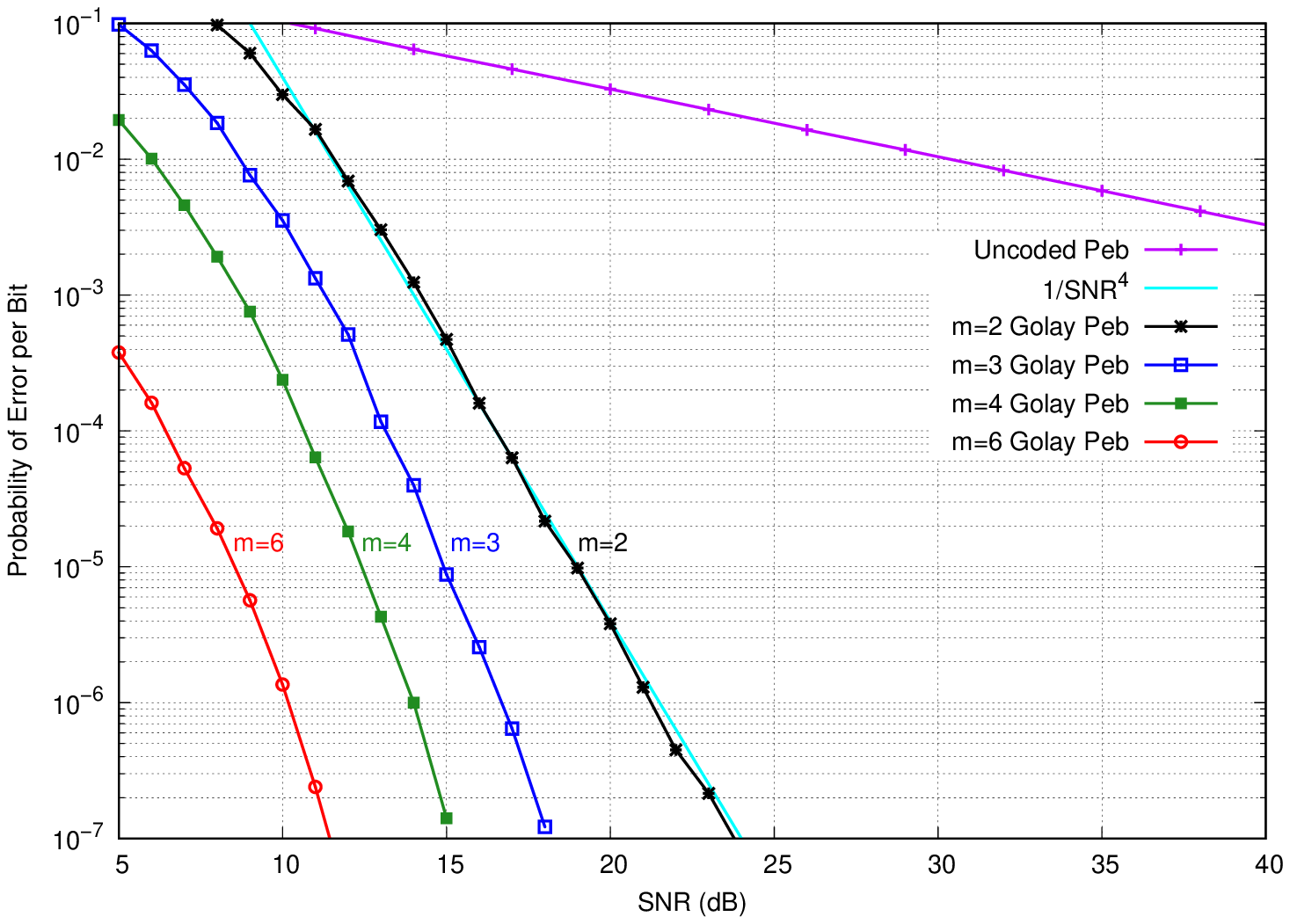} 
\vspace{9mm}
\caption{Soft reconciliation of TE-QKD via soft-decision decoding of the extended binary $[24,12,8]_2$ Golay code. $N=4, 8, 16$, and $64$ bins per frame, $N=2^m$, $m$ coded bits per photon. Diversity follows the result of Theorem~\ref{thm:infinite-diversity-soft} and Table~\ref{tab_golay}. 
\label{fig_golay_soft}}
\end{center}
\end{figure}

%%------------------------------------------------------------------
%%------------------------------------------------------------------
\subsection{Reed-Solomon $RS_1[8,4,5]_8$ and $RS_2[12,6,7]_{16}$ Codes}
The first Reed-Solomon code, $RS_1$, is defined over $\F_8$ and has length $n=8$
and dimension $k=4$. The second Reed-Solomon code, $RS_2$, is defined over $\F_{16}$ and has length $n=12$ and dimension $k=6$. Both RS codes are generated by the sequence $(0, 1, \alpha, \alpha^2, \ldots, \alpha^{n-2})$,
where $\alpha$ is a primitive element of $\F_q$, $q=8$ or $q=16$,
with minimal polynomial $x^3+x+1$ and $x^4+x+1$ respectively. 
The binary image $[24,12,6]_2$ of $RS_1$ has Hamming weight enumerator 
$1+ 56x^6 + 423x^8 + 840x^{10} + 1456x^{12} + 840x^{14} + 423x^{16} + 56x^{18}+ x^{24}$, and $d_{Hmin}(\cC_b)=6$. The binary image $[48,24,8]_2$ of $RS_2$ has Hamming weight enumerator 
$1+ 4x^8 + 73x^9 + 346x^{10} + 1414x^{11} + 4360x^{12} + 11712x^{13} + 28758x^{14} + 64613x^{15} + 133607x^{16} + 253027x^{17} + 436834x^{18} + 688212x^{19} + 997816x^{20} + 1330822x^{21} + 1629166x^{22} + 1844431x^{23} + 1926824x^{24} + \ldots + 4x^{40} + x^{48}$, and $d_{Hmin}(\cC_b)=8$. 

\begin{table}[!h]
\begin{center}
\caption{Achieved TE-QKD diversity for RS codes $[8,4]_8$ and $[12,6]_{16}$.\label{tab_rs}}
\includegraphics[width=0.9\columnwidth]{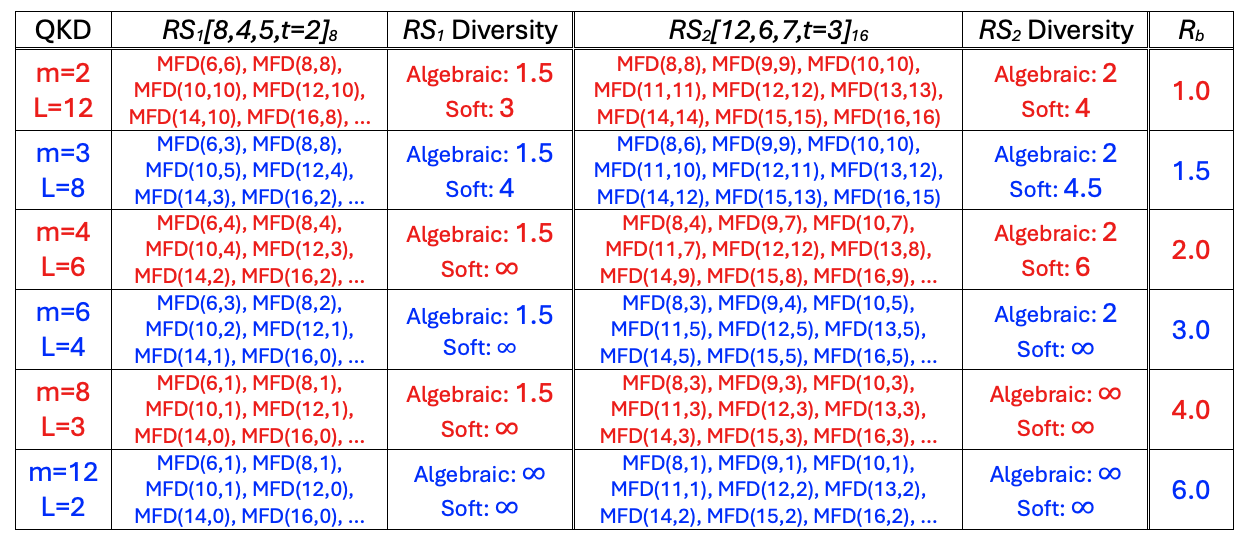}
\end{center}
\end{table}

The MFD distribution MFD$(\omega,\ell_{max})$ for both Reed-Solomon codes is shown in Table~\ref{tab_rs}. For the $RS_1$ code, finite diversity of soft reconciliation at $m=3$ is forced by the term MFD$(8,8)$, where $\omega=8 > d_{Hmin}(\cC_b)=6$. Thus, yielding a diversity $\tfrac{1}{2}\omega$ larger than the minimum diversity from $d_{Hmin}$. A similar surprise with the $RS_2$ code
is the MFD$(9,9)$ and MFD$(12,12)$ terms at $m=3$ and $m=4$ respectively. In all cases, soft reconciliation as for the Golay code attains infinite diversity much earlier than hard reconciliation. For practical applications, we recommend
soft reconciliation with $m=4$ or $m=6$ coded bits per photon.
%%------------------------------------------------------------------
%%------------------------------------------------------------------
\subsection{Binary $BCH_1[30,19,6]_2$ and Quaternary $BCH_2[15,7,7]_4$ Codes}
We build the cyclic binary $[31,20,t=2]_2$ code from the generator polynomial $g(x)=x^{11}+x^8+x^7+x^5+x^4+x^3+x^1+1$ whose roots are $1,\alpha,\alpha^2, \alpha^3$ in $\F_{16}$, $\alpha$ being a primitive element in $\F_{16}$ with minimal polynomial $x^4+x+1$. Then, we shorten by one information bit to get the $BCH_1[30,19,d_{Hmin}=6,t=2]_2$ shortened BCH code. Its Hamming weight enumerator is $1+650x^6+5865x^8+28182x^{10}+87400x^{12}+137700x^{14}
+146115x^{16}+81900x^{18}+30360x^{20}+5490x^{22}+595x^{24}+30x^{26}$.  
For the quaternary code, we directly build $BCH_2[15,7,7]_4$ from its
generator polynomial $g(x)=x^7+\beta x^6+\beta x^5+\beta^2 x^3+\beta^2 x+1$,
where $\beta$ is primitive in $\F_4$, $\beta^2+\beta+1=0$, $\beta=\alpha^5$,
and $g(x)$ has $1,\alpha,\alpha^2,\alpha^3,\alpha^4,\alpha^5$ as roots in 
$\F_{16}$. The Hamming weight enumerator of the binary image of the $BCH_2$ code
is $1+165x^8+1200x^{10}+1880x^{12}+5520x^{14}+3435x^{16}
+3280x^{18}+648x^{20}+240x^{22}+15x^{24}$, and $d_{Hmin}(\cC_b)=8$. 
We could have tested a binary BCH code of length $30$ with $t=3$, but the quaternary BCH code has a thinner weight distribution tail near the minimum distance. 
%%zoom now? Ah!
%%yes :-)
\begin{table}[!h]
\begin{center}
\caption{Achieved TE-QKD diversity for BCH codes $[30,19]_2$ and $[15,7]_4$.\label{tab_bch}}
\includegraphics[width=0.9\columnwidth]{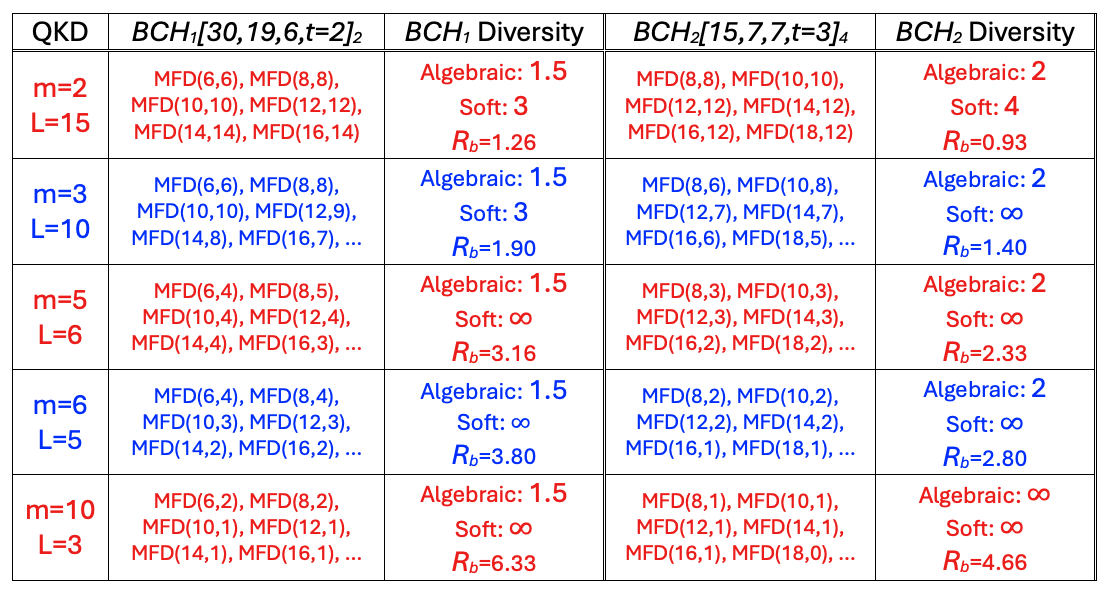}
\end{center}
\end{table}
The MFD$(\omega, \ell_{max})$ distribution of both BCH codes is illustrated
in Table~\ref{tab_bch}. We did not encounter any term of the form
MFD$(\omega, \omega)$, for $\omega > d_{Hmin}(\cC_b)$, when MFD$(d_{Hmin}(\cC_b), d_{Hmin}(\cC_b))$ does not exist. As in previous code examples,
soft reconciliation achieves infinite diversity very early at a relatively small number of coded bits per photon. Both codes are recommended for practical TE-QKD information reconciliation applications at $m=5$ and $m=6$ coded bits per photon.\\

\noindent
Remark: The $[30,20,d_{Hmin}=5,t=2]_2$ BCH code used in Section~\ref{sec_coded_div_QKD} has finite diversity at $m=3$. 
It is derived by shortening the very famous $[31,21]_2$ BCH code. 
The diversity of the $[30,20]_2$ BCH code becomes infinite at $m=5$ under soft-decision decoding, as for $BCH_1$. 
In the current sub-section, we preferred to construct $BCH_1[30,19]$ such that $d_{Hmin}=6$ to enhance its MFD deficiency. For soft reconciliation, the $BCH_1$ code attains a diversity order of $3$, instead of $2.5$ for the $[30,20]_2$ code, when $m=3$ coded bits per photon. 

%%------------------------------------------------------------------
%%------------------------------------------------------------------
\subsection{Reed-Muller RM$(2,4)$ and RM$(2,5)$ Codes}
The first considered binary Reed-Muller code is the RM$(2,4)$.
It is equivalent to an extended $[16,11,4]_2$ Hamming/BCH code.
For determining its MFD distribution, we build the RM$(2,4)$ code
from its standard generator matrix $G_1$ where each row is the evaluation of a Boolean basis monomial over $\F_2^4$. The $11$ basis monomials of four variables are $1, x_1, \ldots, x_4, x_1x_2, \ldots, x_3x_4$. We also considered $G_2$, the generator matrix of the equivalent extended Hamming/BCH code. $G_2$ is built by shifting the generator polynomial $x^4+x+1$ and extending via an all-1 column.

The second binary Reed-Muller code is the RM$(2,5)$; it is self-dual and equivalent to the extended BCH$[32,16,8]_2$ code. Similarly, we build the code from its generator matrix $G_1$ from $16$ Boolean basis monomials of five variables, $1, x_1, \ldots, x_5, x_1x_2, \ldots, x_4x_5$, evaluated over $\F_2^5$. The generator matrix $G_2$ is built by shifting the generator polynomial $x^{15}+x^{11}+x^{10}+x^9+x^8+x^7+x^5+x^3+x^2+x+1$ and extending 
via an all-1 column. 

\begin{table}[!h]
\begin{center}
\caption{Achieved TE-QKD diversity for the Reed-Muller code RM$(2,4)$ in three different versions.\label{tab_RM_2_4}}
\includegraphics[width=0.9\columnwidth]{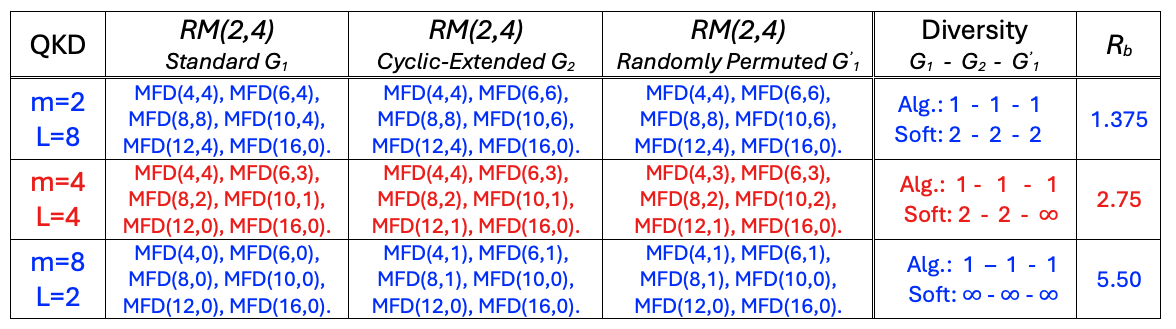}
\end{center}
\end{table}

\begin{table}[!h]
\begin{center}
\caption{Achieved TE-QKD diversity for the Reed-Muller code RM$(2,5)$ in three different versions.\label{tab_RM_2_5}}
\includegraphics[width=0.9\columnwidth]{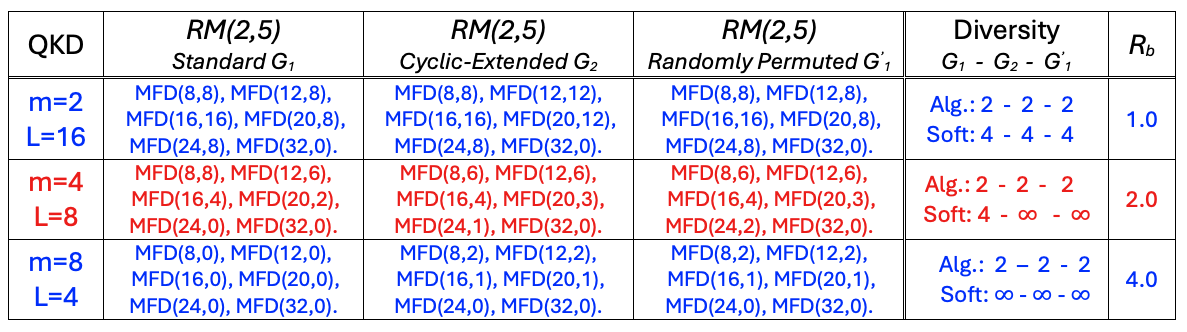}
\end{center}
\end{table}

Tables~\ref{tab_RM_2_4} and \ref{tab_RM_2_5} show the complete MFD$(\omega, \ell_{max})$ distributions of RM$(2,4)$ and RM$(2,5)$ codes respectively.
The information rate per photon ($R_b$ information bits/photon)
is given in the last column. The number of photons per word, $L=n/m$, is considered only for $n$ multiple of $m$ as in previous tables in this section. 
In both cases, for RM$(2,4)$ and RM$(2,5)$, the standard generator matrix $G_1$ is very structured blocking MFD deficiency at $m=4$ coded bits per photon. Hence, we propose to permute the binary digits of the code to achieve MFD deficiency and hence infinite diversity under soft-decision decoding. 
The matrix $G'_1$ is obtained by permuting the columns of $G_1$,
with the two permutations $(0, 2, 12, 1, 8, 13, 4, 15, 11, 6, 9, 5, 7, 10, 3, 14)$ and $(0, 3, 24, 2, 16, 23, 6, 31, 25, 11, 19, 10, 12, 17, 14, 5, 20, 13,$ $21, 22, 30, 1, 26, 4, 28, 9, 7, 8, 18, 15, 29, 27)$, for RM$(2,4)$ and RM$(2,5)$ respectively.
All codes versions attain infinite diversity at $m=8$ for soft reconciliation, as shown in Tables~\ref{tab_RM_2_4} and \ref{tab_RM_2_5}. The case $m=4$ coded bits per photon is the only recommended for a reasonable-complexity lab equipment and good performance of information reconciliation via soft-decision decoding of both Reed-Muller codes. 
\clearpage
\section{Conclusions \label{sec_conclusion}}
%%\nrd{=====SL edited on 24 July 2026 ======\\
%We will rewrite the Conclusions after completing all other Sections.
%}
This work studied the diversity of information reconciliation in time-entangled quantum key distribution systems. Diversity is a crucial parameter to indicate how fast the probability of error is vanishing versus the signal-to-noise ratio. Finite diversity is created by a fading process that weakens a signal with respect to noise. 
In time-entangled QKD systems, the random position of the jitter-free photon relative to the bin borders plays the role of fading. 
At high signal-to-noise ratio, raw key disagreements split into two families with very different statistics.
An error confined to adjacent bins is a polynomial event, whereas a displacement of two bins or more is an exponential event, see Proposition~\ref{prop:single-double-jumps}.
Therefore, a coded system reaches infinite diversity when its decoder absorbs every dominant pattern built from adjacent bin errors, a situation with no equivalent in wireless communications.
Indeed, under conditions stated by Theorems~\ref{thm:infinite-diversity-hard} and~\ref{thm:infinite-diversity-soft}, infinite diversity is created out of finite diversity on the QKD channel with the use of finite length error-correcting codes, an infinite diversity amplification never observed in communication theory.\\ 

%Starting from the PPM-based framework, we showed that the hard-output channel exhibits two distinct high-SNR error mechanisms. Single-bin jumps occur with probability on the order of $\gamma^{-1/2}$ and determine the finite uncoded diversity. Multi-bin jumps decay exponentially with $\gamma$ and yield infinite diversity. 
%This separation explains why coded TE-QKD channels can exhibit a behavior that is not observed in conventional fading channels: a finite-diversity uncoded channel can become an infinite-diversity coded system.

Under bounded-decoding algebraic reconciliation, Theorem~\ref{thm:infinite-diversity-hard} and Corollary~\ref{cor_infinite-diversity-hard-nonbinary} identify the threshold at which this absorption occurs: the $L$ photons carrying a codeword must all fit within the correction radius, i.e., $L \le t$. When this threshold is violated, the diversity falls back to $\tfrac{1}{2}(t+1)$, in agreement with the classical wireless result.\\

Under soft reconciliation, the boundary between finite and infinite diversity is no longer set by $t$ but by the MFD property of the code together with the Gray labeling of the frame.
As stated in Theorem~\ref{thm:infinite-diversity-soft} for soft reconciliation, a code has infinite diversity when no non-zero codeword is formed by neighboring labels only, a property referred to as MFD deficiency. On the contrary, the property referred to as full MFD renders finite diversity for soft reconciliation where diversity falls back to $\tfrac{1}{2}\omega$, for some weight $\omega \ge d_{Hmin}$.\\ 
%%The demilitarized zone lemma (Lemma~\ref{lem_DMZ}) provided the geometric tool for this proof, by showing that leaving a zone of fixed width around a bin border is an event of infinite diversity. 
%%The associated rate bounds quantify the redundancy required by the two decoding strategies and explain why soft reconciliation can operate with smaller frames or higher code rates, see Corollaries~\ref{cor_rate_loss_hard} and~\ref{cor_rate_loss_soft}.

%For algebraic hard-decision decoding, we proved that infinite diversity is achieved if and only if the number of photons per codeword does not exceed the correction radius of the code.
%Under this condition, all dominant single-bin jumps can be corrected, and any remaining decoding error must involve at least one multi-bin jump. If this condition is not satisfied, an error event composed only of single-bin jumps gives a polynomial lower bound on the decoding error probability, and the diversity remains finite.
%Infinite diversity is achieved if and only if the number of photons per codeword is smaller than the minimum Hamming distance. This condition ensures that any confusion between two codewords cannot be caused solely by single-bin jumps. At least one photon must undergo a multi-bin displacement, which forces the pairwise error probability to decay exponentially with SNR.
%%The short code examples confirm that moderate frame sizes, around $4$ to $6$ coded bits per photon, already suffice for the considered Golay, RS, BCH, and RM code constructions to achieve infinite diversity. 

Short codes combined with a slightly large number of coded bits per photon allow the infinite diversity conditions of Theorems~\ref{thm:infinite-diversity-hard} and~\ref{thm:infinite-diversity-soft} to be satisfied. The sudden change in the diversity order, from finite to infinite, causes an exponential decrease of the probability of error instead of a polynomial decrease.  
Consequently, we recommend short error-correcting codes for soft information reconciliation in time-entangled quantum key distribution systems where a photon carries three or more coded bits.

%%They also show that conventional parameters such as length, dimension, and minimum distance do not fully determine the diversity performance of TE-QKD reconciliation. Instead, the grouping of coded bits into photon labels, and the coordinate ordering of a code can alter the MFD distribution at no cost in rate.
%%This observation suggests a practical design principle: the code, coordinate permutation, Gray labeling, and the number of bins per frame should be jointly optimized.
%%The efficient MFD testing and permutation search procedures for longer codes, and tighter finite-SNR error estimates for soft-decision decoding are left for future work.

%These results show that the diversity order in TE-QKD is governed jointly by the frame dimension, the number of photons per codeword, and the distance properties of the reconciliation code.
%This suggests a new design principle for TE-QKD information reconciliation. That is, codes should be selected not only for rate and decoding efficiency, but also for their ability to eliminate all dominant single-bin jumps.

%%\nrd{=====SL edited on 3 June 2026 =====}
%%--------------------------------------------------------------
%%--------------------------------------------------------------
\clearpage
\appendices

\section{List of Fundamental Equations \label{app_fundamental}}
Equations (\ref{equ_hpi}), (\ref{equ_Y_cond_hX_tY}), (\ref{equ_pij}), and (\ref{equ_pdfU_valid}) are exact equations of probabilities and densities for the TE-QKD channel from \cite{Boutros2023}. 
We list below a) the a priori probability of a bin given that both frames
are valid, then~b) the conditional density of the photon position
on Bob's side given the bin position of Alice, then~c) the transition probability of the discrete-input discrete-output TE-QKD channel,
and finally~d) the density of the exact photon position 
(before jitter) given that both frames are valid.
Here, the integer set $\Z_N$ is the set $\{0, 1, \ldots, N-1 \}$.\\

\noindent
a) The apriori $\hpi_i$ is the probability that Alice's photon falls in bin $i$ 
given that Bob's photon is inside the frame, $i \in \Z_N$:
\begin{equation}
\label{equ_hpi}
\hpi_i = \pr(\hX=i|\tY \in [0,N)) 
= \frac{\int_0^N \left[ Q\left(\frac{i-u}{\sigma}\right)-Q\left(\frac{i+1-u}{\sigma}\right) \right]\cdot
\left[ Q\left(\frac{-u}{\sigma}\right)-Q\left(\frac{N-u}{\sigma}\right) \right]\, du}
{\int_0^N \left[ Q\left(\frac{-u}{\sigma}\right)-Q\left(\frac{N-u}{\sigma}\right) \right]^2\, du.} 
\end{equation}

\noindent
After a change of variable $u \leftarrow u/N$, the denominator becomes
$N\cdot \int_0^1 f^2_{\sigma/N}(u) du=N (1-\cO(\tfrac{1}{\sqrt{\gamma}}))$ from~d)
in Lemma~1 in~\cite{Boutros2023}, 
where $f_{\sigma}(x)=Q(\tfrac{-x}{\sigma})-Q(\tfrac{1-x}{\sigma})$. 
The numerator becomes
$\int_0^N f_{\sigma}(u-i) f_{\sigma/N}(\tfrac{u}{N}) du$. 
Split the integral into a first integral over the range $[i,i+1]$
and the second outside $[i,i+1]$. 
Using~b) and~c) from Lemma~1 in~\cite{Boutros2023}, you find that
the first integral is $1-\cO(\tfrac{1}{\sqrt{\gamma}})$ and the second integral
is $\cO(\tfrac{1}{\sqrt{\gamma}})$. 
Hence, we have 
\begin{equation}
\label{equ_hpi_simple}
\hpi_i=\frac{1}{N}+\cO\left(\frac{1}{\sqrt{\gamma}}\right),
\end{equation}
where the big $\cO$ term is positive for bins in the middle of the frame
and negative for side bins.\\ 

\noindent
b) The density $p(y|\hx)$ of Bob's photon position given that both frames are valid and Alice's photon is in bin $i$:
\begin{equation}
\label{equ_Y_cond_hX_tY}
    p_{Y|\hX,\tY \in [0,N)}(y|\hx=i) ~=~\frac
    {\int_0^N \frac{1}{\sqrt{2\pi\sigma^2}} ~e^{-\tfrac{(y-u)^2}{2\sigma^2}} \cdot \left[Q\left(\frac{i-u}{\sigma}\right)-Q\left(\frac{i+1-u}{\sigma}\right)\right] \,du}
    {\int_0^N \left[Q\left(\frac{-t}{\sigma}\right)-Q\left(\frac{N-t}{\sigma}\right)\right] \cdot \left[Q\left(\frac{i-t}{\sigma}\right)-Q\left(\frac{i+1-t}{\sigma}\right)\right] \,dt}, 
\end{equation}
where $i \in \Z_N$ and $y \in [0,N[$. At low SNR, $p(y|\hx)$ has a Gaussian shape, whereas it converges towards a square function at high SNR; see Figures~5 and 6 in \cite{Boutros2023}.\\

\noindent
c) Given that Alice's photon is in bin $i$, the probability that Bob's photon is in bin $j$ is given by: 
\begin{equation}
\label{equ_pij}
    p_{ij} ~=~\pr(\hY=j|\hX=i) = \frac
    {\int_0^N \left[Q\left(\frac{j-u}{\sigma}\right)-Q\left(\frac{j+1-u}{\sigma}\right)\right] \cdot \left[Q\left(\frac{i-u}{\sigma}\right)-Q\left(\frac{i+1-u}{\sigma}\right)\right] \,du}
    {\int_0^N \left[Q\left(\frac{-t}{\sigma}\right)-Q\left(\frac{N-t}{\sigma}\right)\right] \cdot \left[Q\left(\frac{i-t}{\sigma}\right)-Q\left(\frac{i+1-t}{\sigma}\right)\right] \,dt},
\end{equation}
where $i, j \in \Z_N$.
From Proposition~2 in \cite{Boutros2023}, we know that $p_{ij}=\cO(\sigma)=\cO(\tfrac{1}{\sqrt{\gamma}})$
for $|i-j|=1$, i.e., a single-bin jump from $\hX$ to $\hY$. Such events dominate the performance
without coding and yield a $1/2$-diversity uncoded performance. For $|i-j|=2$,
it turns into $p_{ij}=\cO(e^{-\tfrac{\gamma}{4}})$. These double-jump events create infinite diversity if the error-correcting code wipes out all single-jump errors.\\ 

\noindent
d) The density $p_U(u)$ of the exact photon position $U \in [0,N[$, conditioned on both frames being valid, is given by the expression: 
\begin{equation}
\label{equ_pdfU_valid}
p(u|\tX,\tY \in [0,N))= \frac{\left[Q\left(\frac{-u}{\sigma}\right)-Q\left(\frac{N-u}{\sigma}\right)\right]^2}
{\int_0^N \left[Q\left(\frac{-t}{\sigma}\right)-Q\left(\frac{N-t}{\sigma}\right)\right]^2 \,dt}. 
\end{equation}
From Lemma~1 in \cite{Boutros2023}, the numerator in (\ref{equ_pdfU_valid}) is
\[
\left[ Q\left(\frac{-u}{\sigma}\right) - Q\left(\frac{N-u}{\sigma}\right) \right]^2 = 
1-\cO(\exp(-u^2 \cdot \tfrac{\gamma}{2})),~~~~\text{for}~~u \in ]0,\tfrac{N}{2}].
\]
Then, with the properties of the Gaussian tail function, we have
\[
\int_0^{N/2} \exp(-u^2 \cdot \tfrac{\gamma}{2}) du = \cO(\tfrac{1}{\gamma}).
\]
Hence, at high SNR, for $u \in ]0,\tfrac{N}{2}]$, the density of $U$ becomes
\begin{equation}
\label{equ_pdfU_highsnr}
p(u)=p(N-u)=\frac{1-\cO(\exp(-u^2 \cdot \tfrac{\gamma}{2}))}{N-\cO(\tfrac{1}{\gamma})} ~\sim~\frac{1}{N}.
\end{equation}

%%-------------------------------------------------------------------
%%-------------------------------------------------------------------
\clearpage
\section{Big $\cO$ Notations of Useful Integrations\label{app:bigO}}
In~\cite{Boutros2023}, we have shown that
\begin{align}
    I_1 = \int_0^1 Q\left( \frac{v}{\sigma} \right) ~dv =
    \int_0^1 Q\left( \frac{1-v}{\sigma} \right) ~dv  ~=~ \frac{1}{ \sqrt{2\pi \gamma} } + \cO\left( e^{-\frac{\gamma}{2}}\right),
    \label{eq:v/s}
    \\
    I_2 = \int_0^1 \left[ Q\left( \frac{v}{\sigma} \right) \right]^2 ~dv =
    \int_0^1 \left[ Q\left( \frac{1-v}{\sigma} \right) \right]^2 ~dv
    ~=~ \frac{(\sqrt{2} -1)}{ 2 \sqrt{\pi \gamma} } + \cO\left( e^{-\gamma}\right),
    \label{eq:v^2/s^2}
    \\
    I_3 = \int_0^1 Q\left( \frac{v}{\sigma} \right) Q\left( \frac{1- v}{\sigma} \right)~dv 
    ~=~  \cO\left( e^{-\frac{\gamma}{4}}\right),
    \label{eq:v/s(1-v)/s}
\end{align}
where $\gamma = \frac{1}{\sigma^2}$. The integrals $I_1$ and $I_2$ correspond to finite diversity of order $1/2$, while $I_3$ has infinite diversity. We propose in this appendix a general technique to cover almost all integrals over a finite interval involving one or two Gaussian tail functions, useful in the infinite diversity case.\\

Let $a$, $b$, $c$, $d$, $e$, $f$ be finite real numbers. Also assume that $e\ge 0$, $f\ge 0$, and for $v \in [e,f]$, we have $av+b\ge 0$
and $cv+d \ge 0$. Consider
the integral
\begin{equation}
I=\int_e^f Q\left(\frac{av+b}{\sigma}\right) \cdot Q\left(\frac{cv+d}{\sigma}\right)~dv.
\end{equation}
Now we use the inequality $Q(x) \le \exp(-x^2/2)$, for $x\ge 0$, to get an upper bound of $I$,
\begin{equation*}
I \le \int_e^f \exp\left(-\frac{(av+b)^2}{2\sigma^2}-\frac{(cv+d)^2}{2\sigma^2}\right) ~dv
= \int_e^f \exp\left(-\frac{(a^2+c^2)v^2+2(ab+cd)v+b^2+d^2}{2\sigma^2}\right) ~dv
\end{equation*}
The parabola $g(v)=(a^2+c^2)v^2+2(ab+cd)v+b^2+d^2$ in the numerator of the exponent in the integrand has its minimum at $v_{min}=-\frac{(ab+cd)}{(a^2+c^2)}$ and its value $g_{min}=(a^2d^2+b^2c^2-2abcd)/(a^2+c^2)$.
Three cases are to be distinguished:\\
Case 1: $e \le v_{min} \le f$.
\begin{equation}
I \le (f-e) \cdot \exp\left(-\frac{g_{min}}{2\sigma^2}\right) ~=~ \cO\left(\exp\left(-\frac{\gamma}{2} \cdot \frac{a^2d^2+b^2c^2-2abcd}{a^2+c^2}\right)\right).
\end{equation}
Case 2: $v_{min} \le e$.
\begin{equation}
I \le (f-e) \cdot \exp\left(-\frac{g(e)}{2\sigma^2}\right) ~=~
\cO\left(\exp\left(-\frac{\gamma}{2} \cdot ((a^2+c^2)e^2+2(ab+cd)e+b^2+d^2)\right)\right).
\end{equation}
Case 3: $f \le v_{min}$.
\begin{equation}
I \le (f-e) \cdot \exp\left(-\frac{g(f)}{2\sigma^2}\right) ~=~
\cO\left(\exp\left(-\frac{\gamma}{2} \cdot ((a^2+c^2)f^2+2(ab+cd)f+b^2+d^2)\right)\right).
\end{equation}
We use the above $\cO()$ results to cover infinite-diversity integrals encountered in this paper. The reader can check that $I_3$ can be solved via Case~1. We also obtain
\begin{equation}
\label{eq:2-v-O}
\int_0^1 Q\left( \frac{2- v}{\sigma } \right) dv ~=~ \cO\left( e^{- \frac{\gamma}{2}} \right),~~~\text{and}~~~
\int_0^1 \left(Q\left( \frac{2- v}{\sigma } \right)\right)^2 dv ~=~ \cO\left( e^{-\gamma} \right).
\end{equation}
Similarly, 
\begin{align}
\int_0^1 Q\left( \frac{3 - v}{\sigma } \right) dv
= \cO\left( e^{- 2\gamma } \right),
\label{eq:3-v-O}\\
        %^^^^^^^^^^^^^^^^^^^^^^^^^^^
\int_0^1 Q\left( \frac{v}{\sigma} \right) Q\left( \frac{2- v}{\sigma} \right) dv  = \cO\left( e^{-\gamma}\right), 
\label{eq:v-2-v-O}\\
        %^^^^^^^^^^^^^^^^^^^^^^^^^^^
    \int_0^1 Q\left( \frac{1-v}{\sigma} \right) Q\left( \frac{1 + v}{\sigma} \right) dv 
    = \cO\left( e^{-\gamma }\right),
    \label{eq:1-v-1+v-O}
        %^^^^^^^^^^^^^^^^^^^^^^^^^^^
    \\
       \int_0^1 Q\left( \frac{v}{\sigma} \right) Q\left( \frac{3 - v}{\sigma} \right) dv 
    = \cO\left( e^{-\frac{5}{2}\gamma }\right),
    \label{eq:v-3-v-O}
        %^^^^^^^^^^^^^^^^^^^^^^^^^^^
         \\
       \int_0^1 Q\left( \frac{2-v}{\sigma} \right) Q\left( \frac{ 1 + v}{\sigma} \right) dv 
    = \cO\left( e^{- \frac{9}{2}\gamma }\right),
    \label{eq:2-v-1+v-O}
        %^^^^^^^^^^^^^^^^^^^^^^^^^^^
        \\
     \int_0^1 Q\left( \frac{1-v}{\sigma} \right) Q\left( \frac{2- v}{\sigma} \right) dv 
     =
     \int_0^1 Q\left( \frac{v}{\sigma} \right) Q\left( \frac{1+ v}{\sigma} \right) dv
     = \cO\left( e^{-\frac{\gamma}{2}}\right),
    \label{eq:1-v-2-v-O}
    %^^^^^^^^^^^^^^^^^^^^^
    \\
    \int_0^1 Q\left( \frac{1-v}{\sigma} \right) Q\left( \frac{3- v}{\sigma} \right) dv 
     =
     \int_0^1 Q\left( \frac{v}{\sigma} \right) Q\left( \frac{2 + v}{\sigma} \right) dv
     = \cO\left( e^{-2\gamma }\right),
    \label{eq:1-v-3-v-O}
    %^^^^^^^^^^^^^^^^^^^^^^^^^^
    \\
    \int_{0}^1 Q\left(\frac{2-v}{\sigma}\right)
    Q\left(\frac{3-v}{\sigma}\right) dv
    = \cO\left(e^{-\frac{5}{2} \gamma} \right).
\end{align}
%\end{itemize}

%%\nrd{Hello Siyao, I made Appendix B general for all integrals of infinite diversity. You typed most of the equations! Thank You. June 2, 2026. ===================}
%%\nrd{SL: The idea of extending the technique we used for the general case is brilliant! I'm glad these equations were helpful :)}
%=========================================

%%-------------------------------------------------------------------
%%-------------------------------------------------------------------
%%----------------------- Appendix C --------------------------------
%%\section{Decision Boundary in the Finite-Diversity Case\label{app_C}}
%%\nrd{Added on June 26.\\}
%%Fill here. 

%%\begin{figure}[!h]
%%\centerline{\includegraphics[width=0.5\linewidth]{boundary_Lemma6.png}}
%%\begin{center}
%%\caption{Boundary between the decision regions of $c_A=000~000~000$ and $c'_A=100~100~100$ in dimension $3$, in the finite diversity case.\label{fig_boundary_Lemma6}}
%%\end{center}
%%\end{figure}
%%-------------------------------------------------------------------
%%------------------- End of Appendix C -----------------------------
%%-------------------------------------------------------------------

%-------------------------- Comment out -----------------------------

%=================================================

%%-----------------------------------------------------------------------
%%-----------------------------------------------------------------------
\clearpage
\bibliographystyle{IEEEtran}

%%--------------------------------------------------------------
\end{document}